\documentclass[preprint,review,3p]{elsarticle}

\usepackage{amssymb}
\usepackage{amsmath}
\usepackage{amsthm}
\theoremstyle{plain}
\newtheorem{theorem}{Theorem}[section]
\newtheorem{corollary}[theorem]{Corollary}
\newtheorem{lemma}[theorem]{Lemma}
\newtheorem{proposition}[theorem]{Proposition}
\theoremstyle{definition}
\newtheorem{definition}{Definition}[section]
\theoremstyle{remark}
\newtheorem{remark}{Remark}[section]
\theoremstyle{definition}
\newtheorem{assumption}{Assumption}
\usepackage{algorithm}
\usepackage{algpseudocode}
\algrenewcommand\algorithmicrequire{\textbf{Input:}}
\algrenewcommand\algorithmicensure{\textbf{Output:}} 
\usepackage{mathrsfs}
\usepackage{graphicx}
\usepackage{caption}
\usepackage{subcaption}
\usepackage{float}
\usepackage{tikz}
\usetikzlibrary{
    decorations.pathreplacing,
    arrows.meta,      
    positioning,      
    calc,             
    shapes.geometric, 
    shapes.misc,      
    fit,              
    backgrounds        
}

\begin{document}

\begin{frontmatter}



\title{Direct Band-Energy Estimation from Finite Samples via Continuous Slepian Multitapers}

\author[ntu]{Chun-Jen Shih} 
\ead{f10942045@ntu.edu.tw}
\author[ntu]{Jian-Jiun Ding\corref{cor1}}
\cortext[cor1]{Corresponding author.}
\ead{jjding@ntu.edu.tw}

\address[ntu]{Graduate Institute of Communication Engineering, National Taiwan University, No. 1, Sec. 4, Roosevelt Rd., Taipei 106319, Taiwan (R.O.C.)}

\begin{abstract}
In this work, a direct band-energy estimation approach using finite sampled data under explicit time- and frequency-domain regularity and concentration assumptions is proposed. Standard two-stage methods always estimate the power spectral density and then integrate it over the target band, thereby accumulating error through spectral smearing, leakage, and numerical band integration. Instead, we formulate the band-energy estimation as a finite-sample approximation problem. We then construct a computable direct estimator based on sampled continuous prolate spheroidal wave functions (PSWFs) in a multitaper-type framework, which links continuous-time/frequency assumptions to discrete observations. The estimator is obtained by first constructing a band $L^2$-norm estimator and then squaring it to estimate the band energy. We derive deterministic and non-asymptotic error bounds that explicitly separate four sources of approximation error: time limitation, sampling, band limitation, and finite-term truncation. In addition, several theories are derived to show how the accuracy of the estimation depends on the length, bandwidth, sampling step, time-frequency smoothness, and concentration. We also establish deterministic stability under additive noise, analyze computational complexity, and show that the method admits an efficient multiple-band implementation through offline reuse of PSWF-related computations. Numerical experiments across multiple observation lengths, bands of interest, and noise levels are conducted to show that the proposed approach achieves excellent performance compared to integrated-periodogram-type methods.
\end{abstract}


\begin{keyword}
band energy estimation \sep direct functional estimation \sep prolate spheroidal wave functions \sep multitaper methods \sep deterministic non-asymptotic error bounds \sep finite-sample approximation
\end{keyword}

\end{frontmatter}



\section{Introduction}
{Approximating spectral functionals from finite sampled data is a critical task in computational spectral analysis. In many applications, the quantity of interest is not the full power spectral density (PSD), but an integrated band quantity such as the band energy over a prescribed frequency interval. This occurs, for example, in physiological band-energy computation \cite{Application1, Application2, Application3}, subband/filter-bank analysis \cite{Application4}, and spectrum sensing \cite{Application5}, where the target is an aggregate band quantity rather than fine spectral-line resolution. These settings motivate the study of \emph{direct} finite-sample approximation procedures for spectral functionals.\\
\smallskip
\noindent
\textbf{Direct band-energy estimation.} A common way for band-energy estimation is the plug-in approach: one first estimates the PSD and then numerically integrates the estimated spectrum over the band of interest \cite{Integratedperiodogram1,Integratedperiodogram4,Integratedperiodogram6}. This is inherently a two-stage procedure for a one-stage target. Under finite observation, the estimation of the full spectrum is affected by time truncation, spectral smearing, leakage, and taper choice \cite{Window, Spectralanalysis, Welch2}. When the quantity of interest is the band functional itself, rather than the full spectrum, fine spectral-line resolution is unnecessary; what matters is the aggregate energy over the band of interest. From this perspective, direct band-functional estimation is attractive because it avoids the additional error accumulated through full-spectrum estimation followed by numerical band integration.\\
\smallskip
\noindent
\textbf{Addressed research problems.} In this paper, we study direct band-energy estimation from finite uniformly sampled data. More precisely, we address the following questions:
\begin{enumerate}
    \item Can one construct a direct estimator for the band energy from finite discrete samples of an underlying continuous-time signal?
    \item Can one derive deterministic, non-asymptotic error bounds that explicitly separate the effects of time limitation, sampling, band limitation, and finite-term truncation?
    \item Can the estimation procedure be implemented efficiently, with the principled choices of observation length, sampling step, and estimator parameters?
\end{enumerate}
\smallskip
\noindent
\textbf{The proposed approach.} Motivated by the optimal joint time-frequency energy concentration of continuous prolate spheroidal wave functions (PSWFs) \cite{PSWF1,PSWF2,PSWF4} and by the multitaper principle, we construct a multitaper-type estimator for band energy using \emph{sampled continuous PSWFs} as tapers. Classical multitaper methodology is usually formulated with discrete prolate spheroidal sequences (DPSS) \cite{PSWF3} and is primarily used for PSD estimation \cite{Multitaper1,Multitaper2,Multitaper3,Multitaper4}. Here, instead, we start from continuous PSWFs and use their samples as taper sequences, thereby linking discrete observations to continuous-time and continuous-frequency assumptions. The estimator is constructed by first estimating the band $L^2$-norm and then squaring it to obtain the band energy.\\
\smallskip
\noindent
\textbf{Position relative to existing work and novelty.} Much of the existing theory for spectral-functional estimation is stochastic in nature and follows a PSD-integration strategy, as in integrated-periodogram-type estimators \cite{Integratedperiodogram1,Integratedperiodogram4,Integratedperiodogram2,Integratedperiodogram3,Integratedperiodogram7} and their tapered variants \cite{Integratedperiodogram6,Integratedperiodogram5}. In these works, the data are modeled as stationary or locally stationary random processes, and performance is assessed through bias, variance, and asymptotic distributions. Nearby deterministic band-centric methodologies are more commonly developed in filter-bank/subband analysis \cite{Application4,Filterbank} and energy detection \cite{Energydetection}, rather than band-energy estimation problems under explicit continuous-time and continuous-frequency assumptions. Similarly, PSWF-related studies have been developed for spectral-estimation windows \cite{PSWF6}, sampling-error reduction \cite{PSWF7}, approximation of bandlimited or Sobolev-regular functions \cite{PSWF8}, and computational aspects of prolate-type constructions \cite{PSWF9,PSWF10}, but not for direct estimation of band-energy functionals from finite samples. In contrast, we study this problem in a deterministic, non-asymptotic framework and develop sampled continuous PSWF multitapers for direct band-energy estimation. To the best of our knowledge, we are not aware of prior work that combines these elements in such a way.\\
\smallskip
\noindent
\textbf{Theoretical analysis and computational aspects.} Many real-world signals exhibit regularity and concentration in both time and frequency \cite{Uncertainty, Bsplines, Wavelet1, Wavelet2}. This motivates a deterministic finite-sample analysis based on smoothness and energy concentration. The main theoretical contribution is a four-term error decomposition for the proposed estimator. Specifically, the approximation error can be decomposed into four parts: time limitation of the observation window, sampling, band limitation, and finite-term truncation in the PSWF expansion. This decomposition yields interpretable finite-sample guarantees and clarifies how the approximation accuracy depends on observation length, sampling step, and bandwidth under time-frequency Sobolev regularity and energy concentration. From a computational perspective, the estimation procedure consists of PSWF computation, evaluation of PSWF samples on the observation grid, and estimator construction. The latter two stages are relatively inexpensive. In multiple-band settings, the first stage can be carried out offline and reused, making the overall procedure efficient on a per-band basis.\\
\smallskip
\noindent
\textbf{Contributions.} The contributions of this paper can be summarized as follows.
\begin{itemize}
    \item We formulate band-energy estimation as a direct finite-sample approximation problem for a spectral functional.
    \item We construct computable multitaper-type estimators based on sampled continuous PSWFs, thereby linking continuous-time and continuous-frequency assumptions to discrete observations.
    \item We derive deterministic, non-asymptotic error bounds with an explicit four-term decomposition associated with time limitation, sampling, band limitation, and finite-term truncation.
    \item We establish corollaries that clarify the roles of the observation length, the band of interest, and the number of retained tapers, and we prove a deterministic stability bound under additive noise.
    \item We discuss computational complexity and design an efficient implementation strategy for multiple-band estimation.
    \item We provide numerical experiments that support the theory and demonstrate favorable performance relative to integrated-periodogram-type methods.
\end{itemize}
\smallskip
\noindent
\textbf{Organization.} The remainder of the paper is organized as follows. Section 2 introduces the notation, assumptions, and background used in the analysis. Section 3 presents the proposed direct estimators. Section 4 develops the theoretical analysis, including the main results and proof outlines. Section 5 reports numerical experiments and discusses their implications.    
}


\section{Preliminaries}
\subsection{Problem Setup and Basic Operators}
{We use standard notation for $L^p$, $\ell^p$, and Sobolev spaces $W^{s,p}$ \cite{Lp,Sobolev}. Throughout, for an integer $s\ge 1$, we write $[s]:=\{1,2,\dots,s\}$. Whenever the functions are sufficiently smooth and weak and classical derivatives coincide, we use $f^{(k)}$ to denote the $k$th derivative. 

For suitable $f:\mathbb R\to\mathbb C$ (for instance, $f\in L^1(\mathbb R)\cap L^2(\mathbb R)$), we use the Fourier-transform convention
\begin{align}
\hat f(\omega):=\mathcal Ff(\omega)&=\int_{-\infty}^{\infty} f(t)e^{-i\omega t}\,dt,\\
f(t):=\mathcal F^{-1}\hat f(t)&=\frac{1}{2\pi}\int_{-\infty}^{\infty}\hat f(\omega)e^{i\omega t}\,d\omega,
\end{align}
and write $f\leftrightarrow \hat f$. The $L^2$ inner product is
\begin{equation}
\label{eq:Parseval}
\langle f,g\rangle:=\int_{-\infty}^{\infty} f(t)\overline{g(t)}\,dt
=\frac{1}{2\pi}\int_{-\infty}^{\infty}\hat f(\omega)\overline{\hat g(\omega)}\,d\omega,
\end{equation}
where the second equality is the Parseval--Plancherel identity \cite{Fourier}. The total energy is $E_{\mathbb R}(f):=\|f\|_{L^2(\mathbb R)}^2=\langle f,f\rangle$. We also use $f*g$ for convolution on $\mathbb R$ and $\delta(\cdot)$ for the Dirac delta distribution. The sinc kernel and indicator (rectangular) function are denoted as follows. They forms a Fourier transform pair.

\begin{equation}
    \label{eq:sinc}
    \rho_{\Omega}(t):=\frac{\sin \Omega t}{\pi t},\qquad
    \boldsymbol 1_{x\in I}:=
    \begin{cases}
    1,&x\in I,\\
    0,&x\notin I,
    \end{cases}
\end{equation}
\begin{equation}
    \label{eq:sinc_rect_Fourier_pair}
    \rho_{\Omega}(t)\leftrightarrow \boldsymbol 1_{\omega\in\mathcal I_{\Omega} = }
    \begin{cases}
    1,&|\omega|\le \Omega,\\
    0,&|\omega|>\Omega,
    \end{cases}
    \qquad \mathcal I_{\Omega}:=[-\Omega,\Omega].
\end{equation}

\begin{definition}[Band-limiting operator]
Let $f:\mathbb R\to\mathbb C$ and let $\mathcal B\subset\mathbb R$ be measurable. The band-limiting operator $P_{\mathcal B}$ is defined by
\begin{equation}\label{eq:band_limiting_operator}
P_{\mathcal B}f := \mathcal F^{-1}\boldsymbol 1_{\mathcal B}\mathcal F f.
\end{equation}
The band-limiting operation plays a central role in our analysis. Its corresponding band energy is
\begin{equation}\label{eq:band_energy}
E_{\mathcal B}(f):=\|P_{\mathcal B}f\|_{L^2(\mathbb R)}^2=\langle f,P_{\mathcal B}f\rangle,
\end{equation}
and the band $L^2$-norm is
\begin{equation}
\label{eq:square_root_energy}
A_{\mathcal B}(f):=\|P_{\mathcal B}f\|_{L^2(\mathbb R)}=\sqrt{E_{\mathcal B}(f)}.
\end{equation}
\end{definition}
\noindent We use $\mathscr B(\mathcal B):=\{f\in L^2(\mathbb R): f=P_{\mathcal B}f\}$ to denote the class of functions band-limited to $\mathcal B$.

We then consider a continuous-time signal $x^\star:\mathbb R\to\mathbb C$, observed through a finite vector of uniform samples $y\in\mathbb C^m$.
\begin{definition}[Uniform sampling timestamps]
For $T>0$ and $\Delta>0$, define
\begin{equation}\label{eq:sampling_timestamps}
\mathcal G_{T,\Delta}:=\Bigl\{t_j=-\frac{T}{2}+(j-1)\Delta \,\Big|\, j\in[m]\Bigr\},
\end{equation}
and assume that $t_m=T/2$. The corresponding half sampling frequency (in rad/s) is
\begin{equation}\label{eq:whs}
W_{\mathrm{hs}}:=\frac{\pi}{\Delta}.
\end{equation}
\end{definition}
\noindent For notational convenience, we assume that $m$ is odd. If $m$ is even, one can discard one endpoint sample and work with $m-1$ samples; this has a negligible effect when $\Delta$ is small.

In the sampling-error analysis, we use the Peano kernel theorem in the standard form \cite{Numericalanalysis}. If $L$ is a linear functional that annihilates all polynomials of degree at most $k$ and $f\in C^{k+1}([a,b])$, then
\begin{equation}\label{eq:peano_theorem}
L(f)=\frac{1}{k!}\int_a^b K(\theta)f^{(k+1)}(\theta)\,d\theta,
\end{equation}
where the kernel $K(\theta)=L((x-\theta)_+^k)$ is the associated Peano kernel.
\subsection{Time- and Frequency- Domain Signal Assumptions}\label{sec:assumptions}
{To formulate the deterministic signal class used in the error analysis, we introduce the effective time and frequency supports
\begin{equation}
    \label{eq:I_T}
    \mathcal I_T:=\Bigl[-\frac{T}{2},\frac{T}{2}\Bigr]\subset\mathbb R,\quad T>0,
\end{equation}
and
\begin{equation}
    \label{eq:I_W_cut}
    \mathcal I_{W_{\mathrm{cut}}}:=[-W_{\mathrm{cut}},W_{\mathrm{cut}}]\subset\mathbb R,\quad W_{\mathrm{cut}}>0.
\end{equation}
Here $\mathcal I_T$ is the essential observation interval of $x^\star$, while $\mathcal I_{W_{\mathrm{cut}}}$ is the essential spectral support of $\hat x^\star$.
\begin{assumption}[Time regularity outside $\mathcal I_T$]
    \label{assumption:TR}
    $x^\star\in W^{2,1}(\mathbb R\setminus\mathcal I_T)$. That is, $x^\star\in L^1(\mathbb R\setminus\mathcal I_T)$, it is second-order differentiable, and its first and second derivatives belong to $L^1(\mathbb R\setminus\mathcal I_T)$.
\end{assumption}
\begin{assumption}[Time localization inside $\mathcal I_T$]
    \label{assumption:TL}
    $x^\star$ is essentially time-limited in $\mathcal I_T$ in the sense that 
    \begin{equation}
        \|x^\star\|_{L^2(\mathcal I_T)}^2=(1-\epsilon_T^2)\|x^\star\|_{L^2(\mathbb R)}^2,
    \end{equation}
    where $0<\epsilon_T^2<1$ and $\epsilon_T$ is small when $T$ is sufficiently large.
\end{assumption}
\begin{assumption}[Spectral regularity inside $\mathcal I_{W_{\mathrm{cut}}}$]
    \label{assumption:SR}
    $\hat x^{\star}\in W^{s,\infty}(\mathcal I_{W_{\mathrm{cut}}})$ for some integer $s\geq2$. That is, $\hat x^{\star}\in L^{\infty}(\mathcal I_{W_{\mathrm{cut}}})$, its derivatives up to order $s$ exist, and all such derivatives also lie in $L^{\infty}(\mathcal I_{W_{\mathrm{cut}}})$.
\end{assumption}
\begin{assumption}[Spectral localization with polynomial decay outside $\mathcal I_{W_{\mathrm{cut}}}$]
    \label{assumption:SL}
    There exist $\alpha_{fd}>0$ and an integer $r\ge 2$ such that
    \begin{equation}
        |\hat x^\star(\omega)|\leq\alpha_{fd}|\omega|^{-r},\quad\omega\in\mathbb R\setminus\mathcal I_{W_{\mathrm{cut}}}
    \end{equation}
\end{assumption}
Assumptions \ref{assumption:TR} and \ref{assumption:TL} control the effect of time truncation, while Assumptions \ref{assumption:SR} and \ref{assumption:SL} control the in-band smoothness and out-of-band spectral decay needed in the sampling and band-limiting analysis. In particular, to avoid aliasing, the sampling step should satisfy
\begin{equation}
    W_{\mathrm{hs}}\geq W_{\mathrm{cut}}.
\end{equation}
Assumptions \ref{assumption:TR} and \ref{assumption:SR} are expressed in Sobolev spaces \cite{Sobolev}; Assumption \ref{assumption:TL} is the standard energy-concentration assumption from the time--band limiting framework \cite{PSWF2}; and Assumption \ref{assumption:SL} is a standard polynomial spectral-tail condition consistent with classical Fourier-analytic decay principles \cite{Fourier}.
}
\subsection{PSWF Ingredients}
{Next, we collect the PSWF properties used in the construction and theoretical analysis of the estimator. General background on PSWFs can be found in \cite{PSWF1,PSWF4,PSWF5}. 

For a given time interval $\mathcal I_T$ and half-bandwidth $\Omega$, let $\psi_n(c,t)$ denote the $n$th PSWF, where the time--bandwidth parameter is
\begin{equation}\label{eq:c}
c:=\frac{\Omega T}{2},\quad \Omega,T>0.
\end{equation}
The associated concentration eigenvalues are denoted by $\lambda_n(c)$ and satisfy
\begin{equation}\label{eq:lambda_n}
\lambda_n(c):=\frac{2c}{\pi}\bigl[R_{0n}^{(1)}(c,1)\bigr]^2,
\end{equation}
where $R_{0n}^{(1)}$ is the radial spheroidal function of the first kind. The eigenvalues satisfy $\lambda_0(c)>\lambda_1(c)>\cdots>0$ and decay rapidly once $n\ge\lceil2c/\pi\rceil$.

PSWFs are real-valued, band-limited to $\mathcal I_{\Omega}=[-\Omega,\Omega]$, and satisfy the Fourier-domain identity
\begin{equation}\label{eq:PSWF_FT}
    \frac{i^n\Omega R_{0n}^{(1)}(c,1)}{\pi}\psi_n(c,t)=
    \frac{1}{2\pi}\int_{-\Omega}^{\Omega} e^{i\omega t}\psi_n\left(c,\frac{\omega T}{2\Omega}\right)d\omega.
\end{equation}
Equivalently, their Fourier transforms take the form
\begin{equation}\label{eq:ell_n}
\ell_n\psi_n\left(c,\frac{\omega T}{2\Omega}\right)\boldsymbol 1_{\omega\in\mathcal I_{\Omega}},
\quad
\ell_n:=\frac{\pi}{i^n\Omega R_{0n}^{(1)}(c,1)}.
\end{equation}
They are also eigenfunctions of the sinc kernel in \eqref{eq:sinc}:
\begin{equation}\label{eq:PSWF_eigen}
\lambda_n(c)\psi_n(c,t)=\int_{-T/2}^{T/2}\rho_{\Omega}(t-s)\psi_n(c,s)\,ds,\qquad n=0,1,2,\dots.
\end{equation}
Moreover, $\{\psi_n(c,\cdot)\}_{n\ge 0}$ is orthogonal and complete in $L^2(\mathcal I_T)$:
\begin{equation}\label{eq:PSWF_orthogonal}
\int_{-T/2}^{T/2}\psi_i(c,t)\psi_j(c,t)\,dt=\lambda_i(c)\delta_{ij},
\end{equation}
and orthonormal and complete in the Paley--Wiener space $\mathscr B(\mathcal I_{\Omega})$:
\begin{equation}\label{eq:PSWF_orthonormal}
\int_{-\infty}^{\infty}\psi_i(c,t)\psi_j(c,t)\,dt=\delta_{ij}.
\end{equation}

Consequently, every $f\in\mathscr B(\mathcal I_{\Omega})$ admits a PSWF expansion. The following standard approximation result will be used to control the finite-term truncation error.
\begin{lemma}[Approximation of band-limited functions by PSWFs \cite{PSWF2}]\label{lemma:function_approximation_PSWFs}
    For $f\in\mathscr B(I_{\Omega})$,
    \begin{equation}
        \left\|f-\sum\limits_{n=0}^{\lfloor\frac{2c}{\pi}\rfloor}\langle f,\psi_n\rangle\psi_n\right\|_{L^2(\mathbb R)}^2\leq12\|f\|_{L^2(\mathbb R\setminus\mathcal I_T)}^2.
    \end{equation}
\end{lemma}
}


\section{Multitaper-Type Band $L^2$-Norm and Band-Energy Estimators}
{The goal of this section is to construct a direct estimator for the band energy of $x^\star$ over a prescribed frequency band. The construction proceeds by first defining a band $L^2$-norm estimator and then squaring it to obtain the band-energy estimator. Let the band of interest have half-bandwidth $\Omega>0$ and center frequency $\omega_0\in\mathbb R$:
\begin{equation}
    \label{eq:band_of_interest}
    \mathcal B(\omega_0,\Omega):=[\omega_0-\Omega,\omega_0+\Omega],
\end{equation}
and assume that $\mathcal B(\omega_0,\Omega)\subset \mathcal I_{W_{\mathrm{cut}}}$. This condition ensures that the target band lies within the essential spectral support introduced in Section 2.2.

Motivated by the PSWF properties summarized in Section 2.3, we define multitaper-type estimators based on sampled continuous PSWFs with time--bandwidth parameter $c=\Omega T/2$ as follows:

\begin{definition}[Multitaper-Type Band $L^2$-Norm and Band-Energy Estimators]
Given the observation vector $y\in\mathbb C^m$, the sampling grid $\mathcal G_{T,\Delta}=\left\{t_j=-\frac{T}{2}+(j-1)\Delta \mid j\in[m]\right\}$,
and the band of interest $\mathcal B(\omega_0,\Omega)=[\omega_0-\Omega,\omega_0+\Omega]$, we first define the band $L^2$-norm estimator by
\begin{equation}
    \label{eq:band_L2_norm_estimator}
    A_{\mathcal B(\omega_0,\Omega)}^{\mathrm{est}}(y):=\Delta\sqrt{\sum\limits^{\lfloor\frac{2c}{\pi}\rfloor}_{n=0}\Big|\sum\limits^m_{j=1}y_j\psi_n(c,t_j)e^{-i\omega_0t_j}\Big|^2},
\end{equation}
where $y_j$ denotes the $j$th component of $y$, and $\psi_n(c,t_j)$ is the $n$th PSWF, with $c=\Omega T/2$, evaluated at the sampling point $t_j$. The band-energy estimator is then obtained by squaring \eqref{eq:band_L2_norm_estimator}:
\begin{equation}
    \label{eq:band_energy_estimator}
    E^{\mathrm{est}}_{\mathcal B(\omega_0,\Omega)}(y):={A^{\mathrm{est}}_{\mathcal B(\omega_0,\Omega)}}^2(y)=\Delta^2\sum\limits^{\lfloor\frac{2c}{\pi}\rfloor}_{n=0}\Big|\sum\limits^m_{j=1}y_j\psi_n(c,t_j)e^{-i\omega_0t_j}\Big|^2.
\end{equation}
\end{definition}

The estimator $E^{\mathrm{est}}_{\mathcal B(\omega_0,\Omega)}(y)$ is the primary object of interest and is used to approximate the true band energy $E_{\mathcal B(\omega_0,\Omega)}(x^\star)$. The auxiliary estimator $A^{\mathrm{est}}_{\mathcal B(\omega_0,\Omega)}(y)$ approximates the corresponding band $L^2$-norm and serves as the intermediate quantity from which the band-energy estimator is constructed. Error bounds for both quantities are derived in Section 4. The construction is multitaper-type in the sense that it employs the sampled continuous PSWFs of orders $0,1, \dots,\lfloor 2c/\pi\rfloor$ as taper sequences. For later analysis, it is convenient to introduce the associated finite-sample band-energy coefficients.

\begin{definition}[Finite-Sample Band-Energy Coefficients] For each $n=0,1,\dots,\lfloor 2c/\pi\rfloor$, the $n$th finite-sample band-energy coefficient is defined by
\begin{equation}
\label{eq:C_n_m}
    C_n^{(m)}(y,\omega_0,\Omega):=\sum\limits^m_{j=1}y_j\psi_n(c,t_j)e^{-i\omega_0t_j}.
\end{equation}
Using these coefficients, the band $L^2$-norm estimator \eqref{eq:band_L2_norm_estimator} and the band-energy estimator \eqref{eq:band_energy_estimator} can be written in the form
\begin{equation}
    \label{eq:band_L2_norm_and_energy_estimators}
    A^{\mathrm{est}}_{\mathcal B(\omega_0,\Omega)}(y)=\Delta\sqrt{\sum\limits^{\lfloor\frac{2c}{\pi}\rfloor}_{n=0}\Big|C_n^{(m)}(y,\omega_0,\Omega)\Big|^2},\quad
    E^{\mathrm{est}}_{\mathcal B(\omega_0,\Omega)}(y)=\Delta^2\sum\limits^{\lfloor\frac{2c}{\pi}\rfloor}_{n=0}\Big|C_n^{(m)}(y,\omega_0,\Omega)\Big|^2.
\end{equation}
\end{definition} 

\subsection{Implementation and Computational Issues}
\begin{algorithm}[th]
\caption{Direct PSWF-based estimation of band energy}
\label{alg:pswf-estimator}
\begin{algorithmic}[1]
\Require observation vector $y=\{y_j\}_{j=1}^m$, sampling interval $\Delta$, observation length $T$, center frequency $\omega_0$, \hspace*{0.7cm}half-bandwidth $\Omega$
\Ensure $A^{\mathrm{est}}_{\mathcal B(\omega_0,\Omega)}(y)$ and $E^{\mathrm{est}}_{\mathcal B(\omega_0,\Omega)}(y)$
\State $c \gets \Omega T/2$
\State $K \gets \lfloor 2c/\pi \rfloor + 1$
\State Compute the PSWFs $\{\psi_n(c,t)\}_{n=0}^{K-1}$
\For{$j=1,\dots,m$}
    \State $t_j \gets -T/2 + (j-1)\Delta$
\EndFor
\For{$n=0,\dots,K-1$}
    \For{$j=1,\dots,m$}
        \State Evaluate $\psi_n(c,t_j)$
    \EndFor
\EndFor
\State Compute the finite-sample band-energy coefficients via \eqref{eq:C_n_m}
\State Compute $A^{\mathrm{est}}_{\mathcal B(\omega_0,\Omega)}(y)$ and $E^{\mathrm{est}}_{\mathcal B(\omega_0,\Omega)}(y)$ via \eqref{eq:band_L2_norm_and_energy_estimators}
\State \Return $\left(A^{\mathrm{est}}_{\mathcal B(\omega_0,\Omega)}(y), E^{\mathrm{est}}_{\mathcal B(\omega_0,\Omega)}(y)\right)$
\end{algorithmic}
\end{algorithm}
The computational procedure is summarized in Algorithm \ref{alg:pswf-estimator}. We now discuss its computational complexity, storage requirements and practical implementation aspects.\\
\textbf{Computation and storage of PSWFs.}
The cost of generating continuous PSWFs depends on the numerical routine used for PSWF computation. In our implementation, they are generated by the \texttt{pswf} function in the Chebfun package \cite{Chebfun}. We denote this preprocessing cost abstractly by $C_{\mathrm{PSWF}}(c,K)$. Each continuous PSWF is stored as a Chebfun object which is internally represented by Chebyshev expansion coefficients. Let $L_n$ denote the representation length of the $n$th PSWF. Then the storage cost for retaining the first $K$ PSWF objects is
$O\left(\sum\limits_{n=0}^{K-1} L_n\right)$. This cost is independent of the sampling-grid size $m$, but depends on the time--bandwidth product $c$, the oscillatory complexity of the PSWFs, and the tolerance used in adaptive Chebfun construction. Given a sampling grid of $m$ points, evaluating all retained PSWFs on that grid requires $O\left(m\sum\limits_{n=0}^{K-1} L_n\right)$ operations. By contrast, if one stores discrete PSWF samples directly rather than continuous PSWF objects, then the storage cost becomes $O(mK)$.\\[0.1cm]
\smallskip
\noindent
\textbf{Computational complexity of Algorithm \ref{alg:pswf-estimator}.} Suppose that the discrete PSWF samples $\psi_n(c,t_j)$ have already been generated. The remaining online cost consists of computing the finite-sample band-energy coefficients in \eqref{eq:C_n_m} and then evaluating \eqref{eq:band_L2_norm_and_energy_estimators}. They require $O(mK)=O(m\Omega T)$ operations. Hence, the total computational complexity is
\begin{equation}
\label{eq:computational_complexity}
C_{\mathrm{PSWF}}(c,K)+O\left(m\sum_{n=0}^{K-1} L_n\right)+O(mK).
\end{equation}
If PSWF samples are precomputed and stored, then only the final term $O(mK)$ remains in the online stage.\\[0.1cm]
\smallskip
\noindent
\textbf{Choices of $\Delta$ and $T$.} The sampling step $\Delta$ should be chosen sufficiently small so that aliasing is negligible. In particular, the associated half sampling frequency $W_{\mathrm{hs}}=\pi/\Delta$ should be large enough relative to the essential spectral support of the underlying signal. The observation length $T$ should be chosen so that the time-truncation effect is sufficiently small. We use the term \emph{minimum observation length} to refer to the shortest observation interval for which time truncation remains negligible at the desired accuracy level. In practice, if one wishes to treat multiple bands with different half-bandwidths, or multiple signals with different minimum observation lengths, then it is often convenient to deliberately choose a larger common value of $T$. This point is closely related to the multi-band implementation strategy described in the following.\\[0.1cm]
\smallskip
\noindent
\textbf{Offline--online structure and multiple-band estimation.}
The decomposition in \eqref{eq:computational_complexity} naturally separates the procedure into an offline stage and an online stage. The computation of continuous PSWFs may be carried out offline. Thereafter, for repeated estimation over multiple bands, only the evaluation of PSWF samples on a prescribed grid and estimator construction remain online. This PSWF-reuse strategy reduces the average cost per band. Four representative scenarios are particularly relevant.
\begin{enumerate}
\item \textbf{Fixed $\Omega$ and fixed $T$.} This is the simplest case and commonly arises when one estimates multiple bands of a single signal and all bands have the same width. One computes PSWFs once, stores the corresponding samples on the observation grid, and reuses them for all band-energy evaluations.
\item \textbf{Fixed $\Omega$ and varying $T$.} This situation arises when one analyzes multiple bands from different signals and all bands have the same width. Let $\mathbb{T}$ denote the set of required minimum observation lengths. One may choose a common observation length $\max\mathbb{T}$, compute PSWFs once with $c=\frac{\Omega\max\mathbb{T}}{2}$, and reuse the resulting handlers for all band computations.
\item \textbf{Varying $\Omega$ and fixed $T$.} This case arises when one analyzes a single signal over multiple bands of different widths. Let $\Theta$ denote the set of half-bandwidths under consideration. The PSWFs $\psi_n(c,t)$ with $c=\frac{T\max\Theta}{2}$ may be first constructed such that \eqref{eq:PSWF_eigen}, \eqref{eq:PSWF_orthogonal}, and \eqref{eq:PSWF_orthonormal} hold for the largest half-bandwidth. Then, for any $\Omega\in\Theta$, the scaled family $\sqrt{\frac{\Omega}{\max\Theta}}\psi_n\left(c,\frac{\Omega}{\max\Theta}t\right)$ satisfies the corresponding relations with essential half-bandwidth $\Omega$ and essential observation length $\frac{\max\Theta}{\Omega}T\geq T$. Hence, the precomputed continuous PSWF handlers can be adapted to all $\Omega\in\Theta$ by scaling.
\item \textbf{Varying $\Omega$ and varying $T$.} This is the most general setting. Suppose that the admissible half-bandwidth/observation-length pairs are $(\Omega_i, T_i)$, $i\in[P]$. One may pre-compute PSWFs with a parameter $c$ that satisfies $c\geq \max_{i\in[P]}\frac{\Omega_iT_i}{2}$. The resulting continuous PSWF handlers can then be reused for all target bands by appropriate scaling.
\end{enumerate}
}


\section{Main Results: Deterministic and Non-Asymptotic Error Bounds}
\subsection{Infinite-Sample and Ideal Coefficients, and Associated Benchmarks}
{To analyze the proposed direct band-energy estimator, we first study the corresponding band $L^2$-norm estimator. This intermediate quantity is analytically convenient because it leads to a clean additive error decomposition, from which the band-energy error bound follows immediately. To bridge the computable estimator with the target band $L^2$-norm, we introduce two auxiliary families of coefficients: the \emph{infinite-sample} coefficients and the \emph{ideal} coefficients. These quantities are not directly available in practice, but they isolate different sources of approximation error and therefore play a central role in the analysis.
\begin{definition}[Infinite-Sample and Ideal band-energy coefficients]
The $n$th infinite-sample band-energy coefficient is defined by
\begin{equation}
    \label{eq:C_n_inf}
    C_n^{(\infty)}(x^\star,\omega_0,\Omega):=\sum\limits_{k\in\mathbb Z}x^\star(k\Delta)\psi_n(c,k\Delta)e^{-i\omega_0k\Delta}.
\end{equation}
The $n$th ideal band-energy coefficient is defined by 
\begin{equation}
    \label{eq:C_n_ideal}
    C_n^\star(x^\star,\omega_0,\Omega):=\frac{\overline{\ell_n}}{2\pi\Delta}\int^{\omega_0+\Omega}_{\omega_0-\Omega}\hat x^\star(\omega)\psi_n(c,\frac{T}{2\Omega}(\omega-\omega_0))d\omega,
\end{equation}
where $\ell_n$ is defined in \eqref{eq:ell_n} and $\overline{\ell_n}$ denotes its complex conjugate. 
\end{definition}
The infinite-sample coefficients correspond to the hypothetical case in which all discrete samples of $x^\star$ are available, while the ideal coefficients correspond to the case in which the ground-truth spectrum $\hat x^\star$ is available. These two idealized objects lead naturally to the following band $L^2$-norm benchmarks.
\begin{definition}[Infinite-Sample Band $L^2$-Norm Benchmark]
\begin{equation}
    \label{eq:infinite_sample_band_L^2}
    A^{\infty}_{\mathcal B(\omega_0,\Omega)}(x^\star):=\Delta\sqrt{\sum\limits^{\lfloor\frac{2c}{\pi}\rfloor}_{n=0}\Big|C^{(\infty)}_n(x^\star,\omega_0,\Omega)\Big|^2}.
\end{equation}
\end{definition}
\begin{definition}[Ideal Band $L^2$-Norm Benchmark]
\begin{equation}
    \label{eq:ideal_sample_band_L^2}
    A^{\mathrm{ideal}}_{\mathcal B(\omega_0,\Omega)}(x^\star):=\Delta\sqrt{\sum\limits^{\lfloor\frac{2c}{\pi}\rfloor}_{n=0}\Big|C^{\star}_n(x^\star,\omega_0,\Omega)\Big|^2}.
\end{equation}
\end{definition}
The coefficient in \eqref{eq:C_n_ideal} is called \emph{ideal} for two reasons: it depends on the unavailable quantity $\hat x^\star$, and it yields an exact representation of the band energy.
\begin{proposition}[Exact Representation of the Band Energy]
\label{proposition:exact_representation}
Given the band of interest $\mathcal B(\omega_0,\Omega)$,
\begin{equation}
    E_{\mathcal B(\omega_0,\Omega)}(x^\star)=\Delta^2\sum\limits^{\infty}_{n=0}\Big|C^\star_n(x^\star,\omega_0,\Omega)\Big|^2.
\end{equation}
\emph{Proof.} See $\ref{proof:proposition_exact_representation}$.
\end{proposition}
\noindent Proposition \ref{proposition:exact_representation} motivates the introduction of the ideal coefficients and the corresponding ideal benchmark. 

The quantities $A^{\infty}_{\mathcal B(\omega_0,\Omega)}(x^\star)$ and $A^{\mathrm{ideal}}_{\mathcal B(\omega_0,\Omega)}(x^\star)$ are called \emph{benchmarks} rather than estimators because neither an infinite-length observation vector nor the ground-truth spectrum is available in practice.
Their role is to decompose the discrepancy between the estimator and the target quantity into analytically meaningful terms. By the triangle inequality, $\Big|A^{\mathrm{est}}_{\mathcal B(\omega_0,\Omega)}(y)-A_{\mathcal B(\omega_0,\Omega)}(x^\star)\Big|$ is bounded by $\Big|A^{\mathrm{est}}_{\mathcal B(\omega_0,\Omega)}(y)-A^{\infty}_{\mathcal B(\omega_0,\Omega)}(x^\star)\Big|$+$\Big|A^{\infty}_{\mathcal B(\omega_0,\Omega)}(x^\star)-A^{\mathrm{ideal}}_{\mathcal B(\omega_0,\Omega)}(x^\star)\Big|$+$\Big|A^{\mathrm{ideal}}_{\mathcal B(\omega_0,\Omega)}(x^\star)-A_{\mathcal B(\omega_0,\Omega)}(x^\star)\Big|$. Thus, the total approximation error is decomposed into three contributions. As explained in Section \ref{sec:proof_outline}, the first term reflects time-limitation and sampling effects, the second term reflects band-limitation effects, and the third term reflects finite-term truncation of the ideal coefficient expansion.
}
\subsection{Main Theorem and Corollaries}\label{sec:main_theorem}
{We now state the main non-asymptotic error bound for the band $L^2$-norm estimator, which serves as the main analytical step toward the band-energy estimate. The theorem decomposes the total approximation error into four contributions corresponding to time limitation, sampling, band limitation, and finite-term truncation.
\begin{theorem}[Main Theorem: Non-Asymptotic Error Bound for the Band $L^2$-Norm Estimator]\label{theorem:main_theorem}
Suppose that $x^\star(t)$ is the underlying data-generating mechanism which, when uniformly sampled on $\mathcal G_{T,\Delta}=\{t_i=-\frac{T}{2}+(i-1)\Delta\,|\,i\in[m]\}$, produces the observation vector $y$. Let $W_{\mathrm{hs}}=\frac{\pi}{\Delta}$ denote the half sampling frequency. Assume that $x^\star(t)$ and its Fourier transform $\hat x^\star(\omega)$ satisfy Assumptions \ref{assumption:TR}, \ref{assumption:TL}, and \ref{assumption:SL}. Let $\mathcal I_T$, $\mathcal I_{W_{\mathrm{cut}}}$, and $\mathcal B(\omega_0,\Omega)$ denote the essential time support of $x^\star$, the essential spectral support of $\hat x^\star$, and the band of interest, as defined in \eqref{eq:I_T}, \eqref{eq:I_W_cut}, and \eqref{eq:band_of_interest}, respectively.

Then the approximation error of the band $L^2$-norm estimator satisfies
\begin{equation}
    \label{eq:non_asymptotic_error_bound_L^2}
    \Big|A^{\mathrm{est}}_{\mathcal B(\omega_0,\Omega)}(y)-A_{\mathcal B(\omega_0,\Omega)}(x^\star)\Big|\leq e_{\mathrm{tl}}+e_{\mathrm{sp}}+e_{\mathrm{bl}}+e_{\mathrm{ft}},
\end{equation}
where $e_{\mathrm{tl}}$, $e_{\mathrm{sp}}$, $e_{\mathrm{bl}}$, and $e_{\mathrm{ft}}$ are the time-limiting, sampling, band-limiting, and finite-term truncation errors, respectively.
\begin{enumerate}
\item The time-limiting error $e_{\mathrm{tl}}$ is
\begin{equation}
    \label{eq:etl}
    e_{\mathrm{tl}}:=\gamma_1(\Delta,\Omega, T)\epsilon_T\|x^\star\|_{L^2(\mathbb R)},
\end{equation}
where
\begin{equation}
    \label{eq:gamma1}
        \gamma_1(\Delta,\Omega, T):=\sqrt{\sum\limits^{\lfloor\frac{\Omega T}{\pi}\rfloor}_{n=0}\left(1-\lambda_n(\frac{\Omega T}{2})+\frac{\Omega^2\Delta^2}{2}\right)}\leq\sqrt{\Big(\frac{\Omega T}{\pi}+1\Big)\Big(\frac{\Omega^2\Delta^2}{2}+1\Big)},
\end{equation}
and $\epsilon_{T}$ is the parameter in Assumption \ref{assumption:TL}.
\item The sampling error $e_{\mathrm{sp}}$ is
\begin{equation}
    \label{eq:esp}
    e_{\mathrm{sp}}:=\frac{\Delta\gamma_1(\Delta,\Omega,T)}{\sqrt{8}}\sqrt{\|\frac{d^2}{dt^2}|{x^\star}|^2\|_{L^1(\mathbb R\setminus\mathcal I_T)}}.
\end{equation}
\item The band-limiting error $e_{\mathrm{bl}}$ is
\begin{equation}
    \label{eq:ebl}
    e_{\mathrm{bl}}:=\gamma_2(\Omega, T)\alpha_{fd}W_{\mathrm{hs}}^{-r},
\end{equation}
where
\begin{equation}
    \label{eq:gamma2}
    \gamma_2(\Omega, T):=\sqrt{6\Omega^2(\Omega T+\pi)/\pi}
\end{equation}
and $\alpha_{fd}$ and $r$ are the parameters in Assumption \ref{assumption:SL}.
\item The finite-term truncation error $e_{\mathrm{ft}}$ satisfies
\begin{equation}
    \label{eq:eft}
    e_{\mathrm{ft}}\leq\sqrt{12}\|P_{\mathcal B(\omega_0,\Omega)}x^\star\|_{L^2(\mathbb R\setminus\mathcal I_T)},
\end{equation}
where $P_{\mathcal B(\omega_0,\Omega)}$ is the band-limiting operator defined in \eqref{eq:band_limiting_operator}.
\end{enumerate}
\emph{Proof.} See Section \ref{sec:proof_outline} and the associated appendices.
\end{theorem}
}
The finite-term truncation term $e_{\mathrm{ft}}$ depends on the tail energy
$\|P_{\mathcal B(\omega_0,\Omega)}x^\star\|_{L^2(\mathbb R\setminus\mathcal I_T)}$. The next two corollaries provide explicit bounds for this quantity under additional regularity assumptions.
\begin{corollary}\label{corollary:bounded_tail_energy}
Under the assumptions of Theorem \ref{theorem:main_theorem}, the non-asymptotic error bound \eqref{eq:non_asymptotic_error_bound_L^2} and the error terms \eqref{eq:etl}–\eqref{eq:eft} hold. Assume in addition that Assumption \ref{assumption:SR} is satisfied. Then
\begin{equation}
    \label{eq:eft_bound1}
    e_{\mathrm{ft}}\leq\sqrt{12}\|P_{\mathcal B(\omega_0,\Omega)}x^\star\|_{L^2(\mathbb R\setminus\mathcal I_T)}\leq\sqrt{12}\inf\limits_{\kappa>0}\gamma^+_{\hat x^\star,s,\omega_0,\Omega}(\kappa)Q_{s,T}(\kappa),
\end{equation}
where
\begin{align}
    \label{eq:gamma_ft}
    &\gamma^+_{\hat x^\star,s,\omega_0,\Omega}(\kappa):=\|\hat x^{\star(s)}\|_{L^1(\mathcal B(\omega_0,\Omega+\kappa))}+\sup\limits_{i=0,1,\cdots,s-1}\|\hat x^{\star(i)}\|_{L^{\infty}(\mathcal I^+_{\omega_0,\Omega,\kappa})},\\
    \label{eq:I_ft}
    &\mathcal I^+_{\omega_0,\Omega,\kappa}:=[-\Omega+\omega_0-\kappa,-\Omega+\omega_0]\cup[\Omega+\omega_0,\Omega+\omega_0+\kappa],\\
    \label{eq:Qst_ft}
    &Q_{s,T}(\kappa):=\frac{(2s+1)!}{3\pi s!\sqrt{2s-1}}(2+2\kappa^{-1})^{s-1}T^{1/2-s}+\sqrt{\frac{\kappa}{\pi}},
\end{align}
and $s$ is the parameter in Assumption \ref{assumption:SR}.\\
\emph{Proof.} See \ref{proof:corollary_bounded_tail_energy}.
\end{corollary}
\begin{corollary}\label{corollary:bounded_tail_energy_special_case}
Under the assumptions of Corollary \ref{corollary:bounded_tail_energy} and under the additional condition  
\begin{equation}
    T>T_c:=\Big(\frac{(2s+1)!2^{2s-1}}{3\pi s!\sqrt{2s-1}}\Big)^{\frac{2}{2s-1}},
\end{equation}
the non-asymptotic error bound \eqref{eq:non_asymptotic_error_bound_L^2} and the error terms \eqref{eq:etl}–\eqref{eq:eft} continue to hold. Moreover,
\begin{equation}
    \label{eq:eft_bound2}
    e_{\mathrm{ft}}\leq\sqrt{12}\gamma^+_{\hat x^\star,s,\omega_0,\Omega}(\kappa_T)Q_{s,T}(\kappa_T),
\end{equation}
where
\begin{align}
    &\kappa_T:=\Big(\frac{2^{2s-1}(2s+1)!}{3\pi s!\sqrt{2s-1}}\Big)^{1/s}T^{\frac{1}{2s}-1},\\
    &Q_{s,T}(\kappa_T)\leq\Big(\frac{2^{3s-1}(2s+1)!}{3\pi s!\sqrt{2s-1}}\Big)^{\frac{1}{2s}}T^{\frac{1}{4s}-\frac{1}{2}}.
\end{align}
\emph{Proof.} See \ref{proof:corollary_bounded_tail_energy_special_case}.
\end{corollary}
The upper bound for $e_{\mathrm{ft}}$ in Theorem \ref{theorem:main_theorem} has the simplest form, but it does not explicitly reveal how the band of interest $\mathcal B(\omega_0,\Omega)$ and the observation interval $\mathcal I_T$ affect the truncation error. Under the additional spectral regularity assumption, Corollary \ref{corollary:bounded_tail_energy} makes these dependencies explicit through the factors $\gamma^+_{\hat x^\star,s,\omega_0,\Omega}(\kappa)$ and $Q_{s,T}(\kappa)$. In particular, \eqref{eq:gamma_ft} shows that the bound depends on the behavior of the derivatives of $\hat x^\star$ on the enlarged band $\mathcal B(\omega_0,\Omega+\kappa)$, while \eqref{eq:Qst_ft} shows that the bound decreases as $T$ increases. The drawback is that \eqref{eq:eft_bound1} is expressed as the minimum of a nontrivial function of $\kappa$. Corollary \ref{corollary:bounded_tail_energy_special_case} addresses this issue by selecting a specific margin $\kappa_T$, which yields a more explicit scaling law in $T$ when $T>T_c$.

Because the band-energy estimator is defined as the square of the band $L^2$-norm estimator, the corresponding error bound follows directly from the band $L^2$-norm estimate.
\begin{align*}
    \Big|E^{\mathrm{est}}_{\mathcal B(\omega_0,\Omega)}(y)-E_{\mathcal B(\omega_0,\Omega)}(x^\star)\Big|&=\Big|{A^{\mathrm{est}}_{\mathcal B(\omega_0,\Omega)}(y)}^2-{A_{\mathcal B(\omega_0,\Omega)}(x^\star)}^2\Big|\\&=\Big|A^{\mathrm{est}}_{\mathcal B(\omega_0,\Omega)}(y)-A_{\mathcal B(\omega_0,\Omega)}(x^\star)\Big|\Big[A^{\mathrm{est}}_{\mathcal B(\omega_0,\Omega)}(y)+A_{\mathcal B(\omega_0,\Omega)}(x^\star)\Big]\\&\leq e_{\mathrm{tot}}(2A_{\mathcal B(\omega_0,\Omega)}(x^\star)+e_{\mathrm{tot}}),
\end{align*}
where $e_{\mathrm{tot}}=e_{\mathrm{tl}} +e_{\mathrm{sp}} +e_{\mathrm{bl}} +e_{\mathrm{ft}}$.
\begin{corollary}[Non-Asymptotic Error Bound for the Band-Energy Estimator]\label{corollary:band_energy_estimator_error_bound}
Suppose that $\Big|A^{\mathrm{est}}_{\mathcal B(\omega_0,\Omega)}(y)-A_{\mathcal B(\omega_0,\Omega)}(x^\star)\Big|\leq e_{\mathrm{tot}}$ is guaranteed by Theorem \ref{theorem:main_theorem}, Corollary \ref{corollary:bounded_tail_energy}, or Corollary \ref{corollary:bounded_tail_energy_special_case}. Then 
\begin{equation}
    \label{eq:non_asymptotic_error_bound_energy}
    \Big|E^{\mathrm{est}}_{\mathcal B(\omega_0,\Omega)}(y)-E_{\mathcal B(\omega_0,\Omega)}(x^\star)\Big|\leq e_{\mathrm{tot}}(2A_{\mathcal B(\omega_0,\Omega)}(x^\star)+e_{\mathrm{tot}}).
\end{equation}
\end{corollary}
Finally, we consider additive perturbations in the observation vector. Let $y_{\eta}=y+\eta$, where $\eta\in\mathbb C^m$ satisfies the deterministic noise constraint
\begin{equation}
    \label{eq:noise_setting}
    \|\eta\|_{\ell^2([m])}\le \varepsilon.
\end{equation}
This stability framework is standard in compressive sensing, frame theory, and inverse problems \cite{Noise1,Noise2,Noise3}, and does not impose any probabilistic model on the perturbation.
\begin{corollary}[Noise Perturbation Analysis]\label{corollary:noise_perturbation_analysis}
Let $y_{\eta}=y+\eta$ be the perturbed observation vector, where $\eta$ satisfies \eqref{eq:noise_setting}. Suppose that $\Big|A^{\mathrm{est}}_{\mathcal B(\omega_0,\Omega)}(y)-A_{\mathcal B(\omega_0,\Omega)}(x^\star)\Big|\leq e_{\mathrm{tot}}$ is guaranteed by Theorem \ref{theorem:main_theorem}, Corollary \ref{corollary:bounded_tail_energy}, or Corollary \ref{corollary:bounded_tail_energy_special_case}. Then the perturbed band $L^2$-norm estimator satisfies
\begin{equation}
    \label{eq:noise_error_bound}
    \Big|A^{\mathrm{est}}_{\mathcal B(\omega_0,\Omega)}(y_{\eta})-A_{\mathcal B(\omega_0,\Omega)}(x^\star)\Big|\leq e_{\mathrm{tot}}+e_{\eta},
\end{equation}
where
\begin{align}
    e_{\eta}&:=\varepsilon\gamma_{\psi}(\Omega,\Delta),\\
    \gamma_{\psi}(\Omega,\Delta)&:=\sqrt{\Delta\sum\limits^{\lfloor\frac{2c}{\pi}\rfloor}_{n=0}\Big(\frac{\Omega^2\Delta^2}{2}+\min\{1,\lambda_n(c)+\frac{\Omega\Delta}{\pi}\}\Big)}.
\end{align}
By Corollary \ref{corollary:band_energy_estimator_error_bound}, the perturbed band-energy estimator satisfies
\begin{equation}
    \Big|E^{\mathrm{est}}_{\mathcal B(\omega_0,\Omega)}(y_{\eta})-E_{\mathcal B(\omega_0,\Omega)}(x^\star)\Big|\leq\big(e_{\mathrm{tot}}+e_{\eta}\big)\big(2A_{\mathcal B(\omega_0,\Omega)}(x^\star)+e_{\mathrm{tot}}+e_{\eta}\big).
\end{equation}
\end{corollary}
\noindent The perturbation term is therefore separated from the deterministic approximation terms associated with time limitation, sampling, band limitation, and finite-term truncation. In particular, \eqref{eq:noise_error_bound} shows that the additional error is linear in the noise $\ell^2$-norm, with proportionality constant $\gamma_{\psi}(\Omega,\Delta)$ determined by $\Omega$, $\Delta$, and the PSWF eigenvalues $\lambda_n(c)$.
\subsection{Proof Outline for Theorem \ref{theorem:main_theorem}}\label{sec:proof_outline}
\subsubsection{From $C^{(m)}_n(y,\omega_0,\Omega)$ to $C^{(\infty)}_n(x^\star,\omega_0,\Omega)$ (finite-sample to infinite-sample band-energy coefficients)}
{We begin with the $n$th finite-sample band-energy coefficient $C^{(m)}_n(y,\omega_0,\Omega)=\sum\limits^m_{j=1}y_j\psi_n(c,t_j)e^{-i\omega_0t_j}$ defined in \eqref{eq:C_n_m}. Since the time regularity and localization assumptions (Assumptions \ref{assumption:TR} and \ref{assumption:TL}) hold, it is natural to study the difference between $C^{(m)}_n(y,\omega_0,\Omega)$ and the $n$th infinite-sample band-energy coefficient $C^{(\infty)}_n(x^\star,\omega_0,\Omega)=\sum\limits_{k\in\mathbb Z}x^\star(k\Delta)\psi_n(c,k\Delta)e^{-i\omega_0k\Delta}$ defined in \eqref{eq:C_n_inf}. Let $x^\star_d[k]=x^\star(k\Delta)$ and $u_{d,n}[k]=u_n(k\Delta)$, where $u_n(t):=\psi_n(c,t)e^{i\omega_0 t}$. Then
$C^{(\infty)}_n(x^\star,\omega_0,\Omega)=\sum\limits_{k\in\mathbb Z}x^\star_{d}[k]\overline{u_{d,n}}[k]$. The following lemma gives a rigorous upper bound for the difference between 
$C^{(m)}_n(y,\omega_0,\Omega)$ and $C^{(\infty)}_n(x^\star,\omega_0,\Omega)$.
\begin{lemma}[Difference between $C^{(m)}_n(y,\omega_0,\Omega)$ and $C^{(\infty)}_n(x^\star,\omega_0,\Omega)$]
    \label{lemma:difference_C_m_C_inf}
    \begin{align}
        &\Big|C^{(m)}_n(y,\omega_0,\Omega)-C^{(\infty)}_n(x^\star,\omega_0,\Omega)\Big|\nonumber\\\leq&\,\frac{1}{\Delta}\left(\epsilon_T^2\|x^\star\|^2_{L^2(\mathbb R)}+\frac{\Delta^2}{8}\|\frac{d^2}{dt^2}|{x^\star}|^2\|_{L^1(\mathbb R\setminus\mathcal I_T)}\right)^{1/2}\left(1-\lambda_n(\frac{\Omega T}{2})+\frac{\Omega\Delta^2}{2}\right)^{1/2},
    \end{align}
    where $\epsilon_T$ is the parameter in Assumption $\ref{assumption:TL}$.\\
    \emph{Proof.} See $\ref{proof:lemma_difference_C_m_C_inf}$.
\end{lemma}
\subsubsection{From $C^{(\infty)}_n(x^\star,\omega_0,\Omega)$ to $C^\star_n(x^\star,\omega_0,\Omega)$ (infinite-sample to ideal band-energy coefficients)}
From the discrete-time counterpart of the generalized Parseval theorem \eqref{eq:Parseval}, we obtain $C^{(\infty)}_n(x^\star,\omega_0,\Omega)=\frac{1}{2\pi}\int^{\pi}_{-\pi}\hat x^\star_{d}(\omega)\overline{\hat u_{d,n}}(\omega)d\omega$, where $\hat x^\star_{d}(\omega)$ and $\hat u_{d,n}(\omega)$ denote the discrete-time Fourier transforms of $x^\star_{d}[k]$ and $u_{d,n}[k]$, respectively. We now deal with $\hat x^\star_{d}(\omega)$ and $\hat u_{d,n}(\omega)$ in turn.
\begin{enumerate}
\item For $\hat x^\star_{d}(\omega)$: Define $x^\star_{s}(t):=x^\star(t)\sum\limits_{k\in\mathbb Z}\delta(t-k\Delta)$. By the Fourier multiplication rule and Poisson summation, 
\begin{equation*}
    \hat x^\star_{s}(\omega)=\frac{1}{2\pi}[\hat x^\star(\omega)*(\frac{2\pi}{\Delta}\sum\limits_{k\in\mathbb Z}\delta(\omega-k\frac{2\pi }{\Delta}))]=\frac{1}{\Delta}\sum\limits_{k\in\mathbb Z}\hat x^\star(\omega-k\frac{2\pi}{\Delta}).
\end{equation*}
On the other hand, $x^\star_s$ can be written as $x^\star_s(t)=\sum\limits_{k\in\mathbb Z}x^\star(k\Delta)\delta(t-k\Delta)$. From the Fourier transform pair $\delta(t-k\Delta)\leftrightarrow e^{-i\omega k\Delta}$, we have 
\begin{equation*}
\hat x^\star_s(\omega)
= \sum_{k\in\mathbb Z} x^\star(k\Delta)\,e^{-i\omega k\Delta}
= \hat x^\star_d(\omega\Delta).
\end{equation*}
Hence,
$\hat x^\star_{d}(\omega)=\frac{1}{\Delta}\sum\limits_{k\in\mathbb Z}\hat x(\frac{\omega-2\pi k}{\Delta})$.
\item For $\hat u_{d,n}(\omega)$: Similarly, $\hat u_{d,n}(\omega)=\frac{1}{\Delta}\sum\limits_{k\in\mathbb Z}\hat u_n(\frac{\omega-2\pi k}{\Delta})$, where $\hat u_n(\omega)$ denotes the discrete-time Fourier transform of $u_n(t)$. The identity $\hat u_n(\omega)=\ell_n\psi_n(c,\frac{T}{2\Omega}(\omega-\omega_0))\boldsymbol{1}_{\omega\in\mathcal B(\omega_0,\Omega)}$ follows from the Fourier transform representation of PSWFs in \eqref{eq:PSWF_FT} and the modulation (frequency-shift) rule.
\end{enumerate}
\hspace*{0.5cm}Note that $C^{(\infty)}_n(x^\star,\omega_0,\Omega)$ can be further simplified as:
\begin{equation*}
    \begin{aligned}
        C^{(\infty)}_n(x^\star,\omega_0,\Omega)&=\frac{1}{2\pi\Delta^2}\int^{\pi}_{-\pi}[\sum\limits_{k_1\in\mathbb Z}\hat x^\star(\frac{\omega-2\pi k_1}{\Delta})][\sum\limits_{k_2\in\mathbb Z}\overline{\hat u_n}(\frac{\omega-2\pi k_2}{\Delta})]d\omega\\&=\frac{1}{2\pi\Delta}\int^{\frac{\pi}{\Delta}}_{-\frac{\pi}{\Delta}}[\sum\limits_{k_1\in\mathbb Z}\hat x^\star(\omega-\frac{2\pi k_1}{\Delta})][\sum\limits_{k_2\in\mathbb Z}\overline{\hat u_n}(\omega-\frac{2\pi k_2}{\Delta})]d\omega\\&=\frac{\overline{\ell_n}}{2\pi\Delta}\sum\limits_{k_1\in\mathbb Z}\sum\limits_{k_2\in\mathbb Z}\int^{\frac{\pi}{\Delta}}_{-\frac{\pi}{\Delta}}\hat x^\star(\omega-k_1\frac{2\pi}{\Delta})\psi_n(c,\frac{T}{2\Omega}(\omega-\omega_0-k_2\frac{2\pi}{\Delta}))\boldsymbol{1}_{\omega\in\mathcal B(\omega_0+k_2\frac{2\pi}{\Delta},\Omega)}d\omega.
    \end{aligned}
\end{equation*}  
For $k_2\geq1$, we have $\frac{\pi}{\Delta}-(\omega_0+k_2\frac{2\pi}{\Delta}-\Omega)=(1-2k_2)\frac{\pi}{\Delta}-(\omega_0-\Omega)<0$ because we impose the constraint $\omega_0-\Omega>-W_{\mathrm{cut}}\geq-W_{\mathrm{hs}}=-\frac{\pi}{\Delta}$. For $k_2\leq-1$, $(\omega_0+k_2\frac{2\pi}{\Delta}+\Omega)-(-\frac{\pi}{\Delta})=(\omega_0+\Omega)-(2k_2-1)\frac{\pi}{\Delta}<0$ since we require$\omega_0+\Omega<W_{\mathrm{cut}}\leq W_{\mathrm{hs}}=\frac{\pi}{\Delta}$. Thus, the integrand is non-zero only when $k_2=0$. Therefore,
\begin{align*}
    C^{(\infty)}_n(x^\star,\omega_0,\Omega)&=\frac{\overline{\ell_n}}{2\pi\Delta}\sum\limits_{k\in\mathbb Z}\int^{\omega_0+\Omega}_{\omega_0-\Omega}\hat x^\star(\omega-k\frac{2\pi}{\Delta})\psi_n(c,\frac{T}{2\Omega}(\omega-\omega_0))d\omega\\&=\frac{\overline{\ell_n}}{2\pi\Delta}\int^{\omega_0+\Omega}_{\omega_0-\Omega}\hat x^\star(\omega)\psi_n(c,\frac{T}{2\Omega}(\omega-\omega_0))d\omega\\&+\frac{\overline{\ell_n}}{2\pi\Delta}\sum\limits_{k\not=0}\int^{\omega_0+\Omega}_{\omega_0-\Omega}\hat x^\star(\omega-k\frac{2\pi}{\Delta})\psi_n(c,\frac{T}{2\Omega}(\omega-\omega_0))d\omega\\&=C^\star_n(x^\star,\omega_0,\Omega)+\frac{\overline{\ell_n}}{2\pi\Delta}\sum\limits_{k\not=0}\int^{\omega_0+\Omega}_{\omega_0-\Omega}\hat x^\star(\omega-k\frac{2\pi}{\Delta})\psi_n(c,\frac{T}{2\Omega}(\omega-\omega_0))d\omega.
\end{align*}
Because the spectral localization assumption with sharp frequency decay (Assumption \ref{assumption:SL}) holds, we expect that $C^{(\infty)}_n(x^\star,\omega_0,\Omega)$ is close to $C^\star_n(x^\star,\omega_0,\Omega)$. The following lemma provides an explicit upper bound on their difference.
\begin{lemma}[Difference between $C^{(\infty)}_n(x^\star,\omega_0,\Omega)$ and $C^\star_n(x^\star,\omega_0,\Omega)$]
    \label{lemma:difference_C_inf_C_ideal} 
    \begin{equation}
        \Big|C^{(\infty)}_n(x^\star,\omega_0,\Omega)-C^\star_n(x^\star,\omega_0,\Omega)\Big |\leq\frac{1}{\Delta}\sqrt{6}\Omega\alpha_{fd}W_{\mathrm{hs}}^{-r},
    \end{equation}
    where $r$ and $\alpha_{fd}$ are parameters in Assumption $\ref{assumption:SL}$.\\
    \emph{Proof.} See $\ref{proof:lemma_difference_C_inf_C_ideal}$.
\end{lemma}
\subsubsection{Upper bound for the approximation error $\Big|A^{\mathrm{est}}_{\mathcal B(\omega_0,\Omega)}(y)-A_{\mathcal B(\omega_0,\Omega)}(x^\star)\Big|$}
Finally, we bound the approximation error $\Big|A^{\mathrm{est}}_{\mathcal B(\omega_0,\Omega)}(y)-A_{\mathcal B(\omega_0,\Omega)}(x^\star)\Big|$.
By the triangle inequality,
\begin{align*}
    &\Big|A^{\mathrm{est}}_{\mathcal B(\omega_0,\Omega)}(y)-A_{\mathcal B(\omega_0,\Omega)}(x^\star)\Big|\\\leq&\,\Big|A^{\mathrm{est}}_{\mathcal B(\omega_0,\Omega)}(y)-A^{\infty}_{\mathcal B(\omega_0,\Omega)}(x^\star)\Big|+\Big|A^{\infty}_{\mathcal B(\omega_0,\Omega)}(x^\star)-A^{\mathrm{ideal}}_{\mathcal B(\omega_0,\Omega)}(x^\star)\Big|+\Big|A^{\mathrm{ideal}}_{\mathcal B(\omega_0,\Omega)}(x^\star)-A_{\mathcal B(\omega_0,\Omega)}(x^\star)\Big|
\end{align*}
To control these three terms, we invoke the reverse triangle inequality
\begin{equation}
\label{eq:reverse_triangle_inequality}
\text{for any } v_1, v_2 \text{ in a normed space,}\quad
\bigl|\|v_1\| - \|v_2\|\bigr|
\leq \|v_1 - v_2\|,
\end{equation}
which is an immediate consequence of the triangle inequality.\\ 
\hspace*{0.5cm}We first bound $\Big|A^{\mathrm{est}}_{\mathcal B(\omega_0,\Omega)}(y)-A^{\infty}_{\mathcal B(\omega_0,\Omega)}(x^\star)\Big|$. Because both the vectors $\Big(C^{(m)}_n(y,\omega_0,\Omega)\Big)_{n=0,1,...,\lfloor\frac{2c}{\pi}\rfloor}$ and $\Big(C^{(\infty)}_n(x^\star,\omega_0,\Omega)\Big)_{n=0,1,...,\lfloor\frac{2c}{\pi}\rfloor}$ lie in $\mathbb C^{\lfloor\frac{2c}{\pi}\rfloor+1}$, the reverse triangle inequality yields
\begin{align*}
    \Big|A^{\mathrm{est}}_{\mathcal B(\omega_0,\Omega)}(y)-A^{\infty}_{\mathcal B(\omega_0,\Omega)}(x^\star)\Big|=&\,\Delta\Bigg|\sqrt{\sum\limits^{\lfloor\frac{2c}{\pi}\rfloor}_{n=0}\Big|C^{(m)}_n(y,\omega_0,\Omega)}\Big|^2-\sqrt{\sum\limits^{\lfloor\frac{2c}{\pi}\rfloor}_{n=0}\Big|C^{(\infty)}_n(x^\star,\omega_0,\Omega)}\Big|^2\Bigg|\\\leq&\,\Delta\sqrt{\sum\limits^{\lfloor\frac{2c}{\pi}\rfloor}_{n=0}\Big|C^{(m)}_n(y,\omega_0,\Omega)-C^{(\infty)}_n(x^\star,\omega_0,\Omega)\Big|^2}.
\end{align*}
By Lemma \ref{lemma:difference_C_m_C_inf}, we have
\begin{align}
    &\Big|A^{\mathrm{est}}_{\mathcal B(\omega_0,\Omega)}(y)-A^{\infty}_{\mathcal B(\omega_0,\Omega)}(x^\star)\Big|\nonumber\\\leq&\,\sqrt{\left(\epsilon_T^2\|x^\star \|_{L^2(\mathbb R)}^2+\frac{\Delta^2}{8}\|\frac{d^2}{dt^2}{x^\star}^2\|_{L^1(\mathbb R\setminus\mathcal I_T)}\right)}\sqrt{\sum\limits^{\lfloor\frac{2c}{\pi}\rfloor}_{n=0}\left(1-\lambda_n(\frac{\Omega T}{2})+\frac{\Omega\Delta^2}{2}\right)}\nonumber\\\leq&\,\epsilon_T\|x^\star \|_{L^2(\mathbb R)}\sqrt{\sum\limits^{\lfloor\frac{2c}{\pi}\rfloor}_{n=0}\left(1-\lambda_n(\frac{\Omega T}{2})+\frac{\Omega\Delta^2}{2}\right)}\nonumber\\&+\frac{\Delta}{\sqrt{8}}\sqrt{\|\frac{d^2}{dt^2}{x^\star}^2\|_{L^1(\mathbb R\setminus\mathcal I_T)}}\sqrt{\sum\limits^{\lfloor\frac{2c}{\pi}\rfloor}_{n=0}\left(1-\lambda_n(\frac{\Omega T}{2})+\frac{\Omega\Delta^2}{2}\right)}.
\end{align}
The second inequality follows again from the triangle inequality. The right-hand side consists of two contributions: The first is termed the time-limiting error $e_{\mathrm{tl}}$ since it depends on the time-localization parameter $\epsilon_T$. The second is termed the sampling error $e_{\mathrm{sp}}$, as it depends on the sampling interval $\Delta$.\\
\hspace*{0.5cm}Next, we address the bound of $\Big|A^{\infty}_{\mathcal B(\omega_0,\Omega)}(x^\star)-A^{\mathrm{ideal}}_{\mathcal B(\omega_0,\Omega)}(x^\star)\Big|$. Because both $\Big(C^{(\infty)}_n(x^\star,\omega_0,\Omega)\Big)_{n=0,1,...,\lfloor\frac{2c}{\pi}\rfloor}$ and $\Big(C^{\star}_n(x^\star,\omega_0,\Omega)\Big)_{n=0,1,...,\lfloor\frac{2c}{\pi}\rfloor}$ belong to
$\mathbb C^{\lfloor\frac{2c}{\pi}\rfloor+1}$, the reverse triangle inequality gives
\begin{align*}
    \Big|A^{\infty}_{\mathcal B(\omega_0,\Omega)}(x^\star)-A^{\mathrm{ideal}}_{\mathcal B(\omega_0,\Omega)}(x^\star)\Big|=&\,\Delta\Bigg|\sqrt{\sum\limits^{\lfloor\frac{2c}{\pi}\rfloor}_{n=0}\Big|C^{(\infty)}_n(x^\star,\omega_0,\Omega)}\Big|^2-\sqrt{\sum\limits^{\lfloor\frac{2c}{\pi}\rfloor}_{n=0}\Big|C^{\star}_n(x^\star,\omega_0,\Omega)}\Big|^2\Bigg|\\\leq&\,\Delta\sqrt{\sum\limits^{\lfloor\frac{2c}{\pi}\rfloor}_{n=0}\Big|C^{(\infty)}_n(x^\star,\omega_0,\Omega)-C^{\star}_n(x^\star,\omega_0,\Omega)\Big|^2}.
\end{align*}
By Lemma \ref{lemma:difference_C_inf_C_ideal}, we have
\begin{equation*}
    \begin{aligned}
        \Big|A^{\infty}_{\mathcal B(\omega_0,\Omega)}(x^\star)-A^{\mathrm{ideal}}_{\mathcal B(\omega_0,\Omega)}(x^\star)\Big|\leq\,&\alpha_{fd}W_{\mathrm{hs}}^{-r}\sqrt{6\Omega^2(\Omega T+\pi)/\pi}.
    \end{aligned}
\end{equation*}
The right-hand side is called the band-limiting error, as it depends on the spectral localization parameters $\alpha_{fd}$ and $r$.\\
\hspace*{0.5cm}Finally, we deal with the bound of  $\Big|A^{\mathrm{ideal}}_{\mathcal B(\omega_0,\Omega)}(x^\star)-A_{\mathcal B(\omega_0,\Omega)}(x^\star)\Big|$. Because both $\Big(C^\star_n(x^\star,\omega_0,\Omega)\Big)_{n\geq0}$ and
$\Big(C^\star_0(x^\star,\omega_0,\Omega),C^\star_1(x^\star,\omega_0,\Omega),\cdots,C^\star_{\lfloor\frac{2c}{\pi}\rfloor}(x^\star,\omega_0,\Omega),0,0,\cdots\Big)$ lie in $\ell^2(\mathbb N)$, we obtain
\begin{align*}
    \Big|A^{\mathrm{ideal}}_{\mathcal B(\omega_0,\Omega)}(x^\star)-A_{\mathcal B(\omega_0,\Omega)}(x^\star)\Big|=&\,\Delta\Bigg|\sqrt{\sum\limits^{\lfloor\frac{2c}{\pi}\rfloor}_{n=0}\Big|C^{\star}_n(x^\star,\omega_0,\Omega)}\Big|^2-\sqrt{\sum\limits^{\infty}_{n=0}\Big|C^{\star}_n(x^\star,\omega_0,\Omega)}\Big|^2\Bigg|\\\leq&\,\Delta\sqrt{\sum\limits_{n>\lfloor\frac{2c}{\pi}\rfloor}\Big|C^{\star}_n(x^\star,\omega_0,\Omega)\Big|^2}.
\end{align*}
From Proposition \ref{proposition:exact_representation} and its proof in \ref{proof:proposition_exact_representation}, we know that
\begin{align*}
    \Delta\sqrt{\sum\limits_{n>\lfloor\frac{2c}{\pi}\rfloor}\Big|C^\star_n(x^\star,\omega_0,\Omega)\Big|^2}&=\|P_{\mathcal B(\omega_0,\Omega)}x^\star-\sum\limits^{\lfloor\frac{2c}{\pi}\rfloor}_{n=0}\langle P_{\mathcal B(\omega_0,\Omega)}x^\star,\psi_ne^{i\omega_0t}\rangle\psi_ne^{i\omega_0t}\|_{L^2(\mathbb R)}\\&=\|(P_{\mathcal B(\omega_0,\Omega)}x^\star)e^{-i\omega_0t}-\sum\limits^{\lfloor\frac{2c}{\pi}\rfloor}_{n=0}\langle(P_{\mathcal B(\omega_0,\Omega)}x^\star)e^{-i\omega_0t},\psi_n\rangle\psi_n\|_{L^2(\mathbb R)}.
\end{align*}
Applying Lemma \ref{lemma:function_approximation_PSWFs}, we have
\begin{equation*}
    \Big|A^{\mathrm{ideal}}_{\mathcal B(\omega_0,\Omega)}(x^\star)-A_{\mathcal B(\omega_0,\Omega)}(x^\star)\Big|\leq\,\sqrt{12}\|(P_{\mathcal B(\omega_0,\Omega)}x^\star)e^{-i\omega_0t}\|_{L^2(\mathbb R\setminus\mathcal I_T)}=\sqrt{12}\|P_{\mathcal B(\omega_0,\Omega)}x^\star\|_{L^2(\mathbb R\setminus\mathcal I_T)}.
\end{equation*}
The last term is referred to as the finite-term truncation error, since it arises from truncating the ideal expansion to finite number of band-energy coefficients.
}


\section{Numerical Results and Discussion}
{We report five experiments to assess the proposed band $L^2$-norm and band-energy estimators and to examine the theoretical predictions. Experiment 1 compares the proposed band-energy estimator with classical methods in both noiseless and noisy settings. Experiment 2 illustrates the multiple-band settings in which the proposed implementation is particularly efficient. Experiments 3 and 4 investigate the dependence of the estimators on the observation length $T$ and the band of interest $\mathcal B(\omega_0,\Omega)$, respectively. Experiment 5 studies the stability of the proposed band $L^2$-norm estimator under additive noise across a range of signal-to-noise ratios (SNRs).

To verify the theorems and corollaries, the true band $L^2$-norm and the true band energy must be available as ground truth quantities. We therefore consider three signal families for which the Fourier transforms can be written explicitly.
\begin{enumerate}
\item \textbf{Gaussian mixture.}
\begin{align*}
    x^\star(t)
    &=\sum_{i=1}^3 \frac{A_i}{\sqrt{2\pi\sigma_i^2}}
    \exp\!\left(-\frac{t^2}{2\sigma_i^2}\right)\cos(2\pi f_i t),\\
    \hat x^\star(\omega)
    &=\sum_{i=1}^3 \frac{A_i}{2}
    \left[
    \exp\!\left(-\frac{\sigma_i^2(\omega-2\pi f_i)^2}{2}\right)
    +
    \exp\!\left(-\frac{\sigma_i^2(\omega+2\pi f_i)^2}{2}\right)
    \right].
\end{align*}
Here $A_i>0$, $\sigma_i>0$, and $f_i\in\mathbb R$ denote the amplitudes, standard deviations, and modulation frequencies (in Hz), respectively. Each component of $\hat x^\star(\omega)$ is centered at $\omega=\pm 2\pi f_i$.
\item \textbf{Hyperbolic secant mixture.}
\begin{align*}
    x^\star(t)
    &=\sum_{i=1}^3 A_i\,\operatorname{sech}\!\left(\frac{\pi t}{\tau_i}\right)\cos(2\pi f_i t),\\
    \hat x^\star(\omega)
    &=\sum_{i=1}^3 \frac{A_i\tau_i}{2}
    \left[
    \operatorname{sech}\!\left(\frac{\tau_i(\omega-2\pi f_i)}{2}\right)
    +
    \operatorname{sech}\!\left(\frac{\tau_i(\omega+2\pi f_i)}{2}\right)
    \right].
\end{align*}
Here $A_i>0$, $\tau_i>0$, and $f_i\in\mathbb R$ denote the amplitudes, scalings and modulation frequencies, respectively. The expression for $\hat x^\star(\omega)$ follows from the Fourier transform pair in Chapter 17 of \cite{FT}.
\item \textbf{B-spline mixture.}
\begin{align*}
    x^\star(t)
    &=\sum_{i=1}^3\sum_{j=0}^{r_i+1}
    \frac{A_i}{\tau_i r_i!}(-1)^j\binom{r_i+1}{j}
    \left(\frac{t}{\tau_i}+\frac{r_i+1}{2}-j\right)_+^{r_i}
    \cos(2\pi f_i t),\\
    \hat x^\star(\omega)
    &=\sum_{i=1}^3 \frac{A_i}{2}
    \left[
    \left(\frac{2\pi}{\tau_i}\rho_{\tau_i/2}(\omega-2\pi f_i)\right)^{r_i+1}
    +
    \left(\frac{2\pi}{\tau_i}\rho_{\tau_i/2}(\omega+2\pi f_i)\right)^{r_i+1}
    \right].
\end{align*}
Here $A_i>0$, $\tau_i>0$, $r_i\ge 0$, and $f_i\in\mathbb R$ denote the amplitudes, scalings, B-spline orders, and modulation frequencies, respectively, and $(x)_+^n=x^n$ for $x\ge 0$ and $(x)_+^n=0$ otherwise. The function $\rho_{\cdot}(\cdot)$ is defined in \eqref{eq:sinc_rect_Fourier_pair}. A B-spline of degree $r_i$ is $r_i$-times differentiable, and its Fourier transform is given in \cite{Bsplines}.
\end{enumerate}
In all simulations, the signal parameters are sampled randomly from prescribed intervals. We write
\begin{equation}
    \theta\in\mathscr R(I),
\end{equation}
to indicate that the parameter $\theta$ is drawn from the sampling domain $I$. For example, the condition $A_i\in\mathscr R([5,20])$ means that the amplitude $A_i$ is sampled from the interval $[5,20]$.  
}
\subsection{Comparison with Integrated-Periodogram-Based Methods}
{We compare the proposed band-energy estimator with integrated-periodogram-type methods based on six classical spectral estimators: Thomson’s multitaper (DPSS) periodogram \cite{Multitaper1}, Welch’s averaged periodogram with the Hamming window \cite{Welch1}, and modified periodograms with Hann, Hamming, Blackman, and Kaiser windows \cite{Window}. For each baseline, we first estimated the PSD of the observed signal, then integrated the estimated spectrum over the band of interest, and finally multiplied the result by the observation length $T$ to obtain the estimated band energy. All spectral estimates and band integrals were computed using MATLAB R2024b \cite{Matlab2024b} and the Signal Processing Toolbox \cite{Signalprocessing}, specifically via the commands \texttt{pmtm}, \texttt{pwelch}, \texttt{periodogram}, and \texttt{bandpower}.

We considered all three signal families introduced above. The Gaussian mixtures were parameterized by $\{f_{1j}\in\mathscr R([50,600])\}_{j=1}^3$, $\{\sigma_{1j}\in\mathscr R([0.02,0.1])\}_{j=1}^3$, and $\{A_{1j}\in\mathscr R([6,18])\}_{j=1}^3$; the hyperbolic secant mixtures by $\{f_{2j}\in\mathscr R([50,530])\}_{j=1}^3$, $\{\tau_{2j}\in\mathscr R([0.02,0.05])\}_{j=1}^3$, and $\{A_{2j}\in\mathscr R([2,18])\}_{j=1}^3$; and the B-spline mixtures by $\{f_{3j}\in\mathscr R([110,690])\}_{j=1}^3$, $\{\tau_{3j}\in\mathscr R([0.4,0.65])\}_{j=1}^3$, $\{r_{3j}\in\mathscr R(\{3,4,5\})\}_{j=1}^3$, and $\{A_{3j}\in\mathscr R([7,14])\}_{j=1}^3$. We studied both a noiseless setting and a noisy setting with $\mathrm{SNR}=0$ dB. In the noiseless case, we generated $3\times 15$ observation vectors $y=y_{i,j}$, where $i\in[3]$ indexes the mixture family and $j\in[15]$ indexes the trial. Specifically, $y_{1,j}$ was obtained by sampling Gaussian mixtures with $W_{\mathrm{hs},1}=2\pi\cdot700$ and $T_1=2.5$, $y_{2,j}$ was obtained by sampling hyperbolic secant mixtures with $W_{\mathrm{hs},2}=2\pi\cdot650$ and $T_2=1$, and $y_{3,j}$ was obtained by sampling B-spline mixtures with $W_{\mathrm{hs},3}=2\pi\cdot800$ and $T_3=4$. In the noisy case, the clean observations were generated in the same manner and then corrupted by additive noise to achieve $\mathrm{SNR}=0$ dB. For each simulation, we estimated band energy over 21 bands. The half-bandwidths were set to 1, 3, and 4 Hz for the Gaussian, hyperbolic secant, and B-spline mixtures, respectively. The corresponding band-center frequencies, in Hz, were $\{f_{11}+(i-11)\}_{i=1}^{21}$, $\{f_{23}+(i-11)\}_{i=1}^{21}$, and $\{f_{32}+(i-11)\}_{i=1}^{21}$ for the Gaussian, hyperbolic secant, and B-spline mixtures, respectively.
\begin{figure}[H]
    \centering
    \begin{subfigure}{0.49\textwidth}
        \centering
        \includegraphics[width=\linewidth]{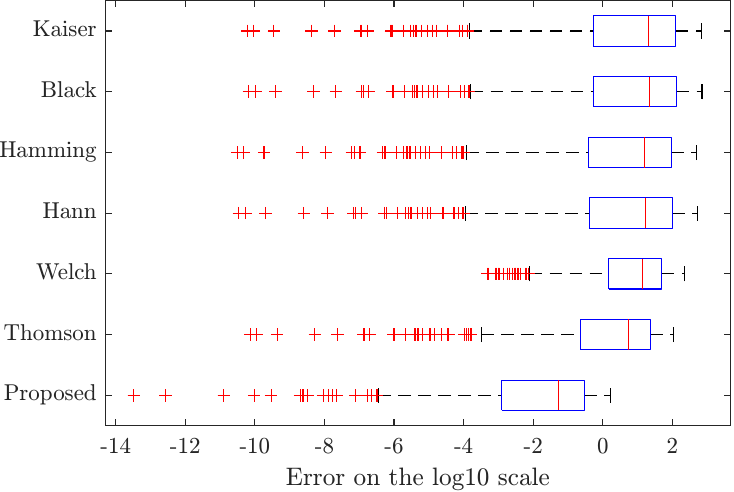}
        \caption{}
        \label{fig:Thomson_experiment1_1}
    \end{subfigure}
    \begin{subfigure}{0.49\textwidth}
        \centering
        \includegraphics[width=\linewidth]{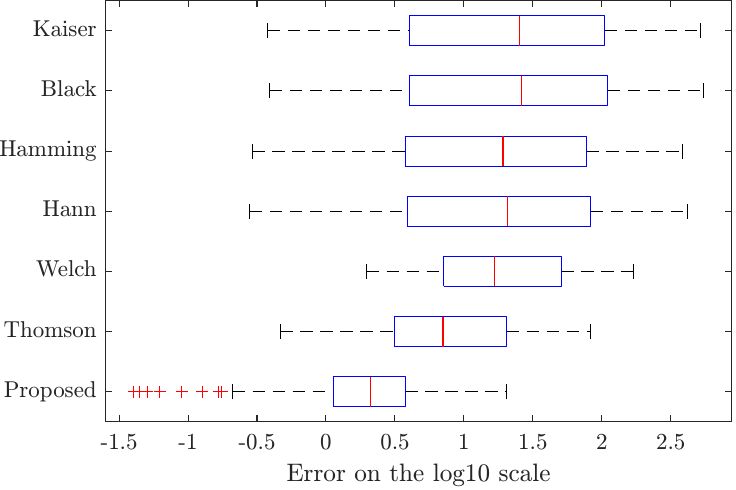}
        \caption{}
        \label{fig:Thomson_experiment1_2}
    \end{subfigure}
    \caption{Boxplots of errors on the $\log_{10}$ scale for the Gaussian mixtures: (a) noiseless, (b) noisy ($\mathrm{SNR}=0$ dB).}
    \label{fig:Thomson_experiment1}
\end{figure}
\begin{figure}[H]
    \centering
    \begin{subfigure}{0.49\textwidth}
        \centering
        \includegraphics[width=\linewidth]{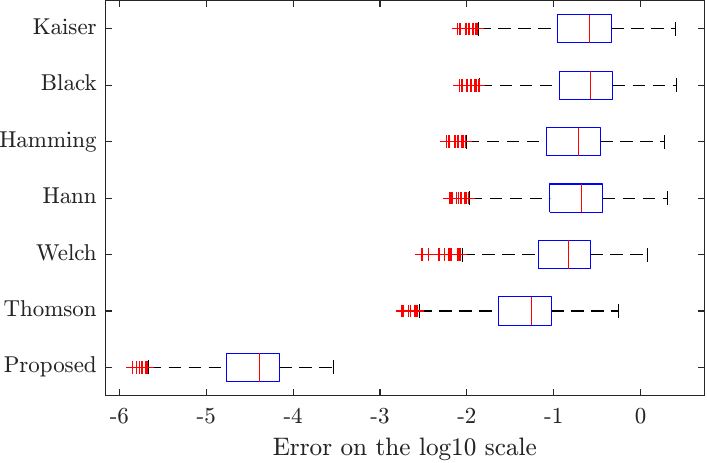}
        \caption{}
        \label{fig:Thomson_experiment2_1}
    \end{subfigure}
    \begin{subfigure}{0.49\textwidth}
        \centering
        \includegraphics[width=\linewidth]{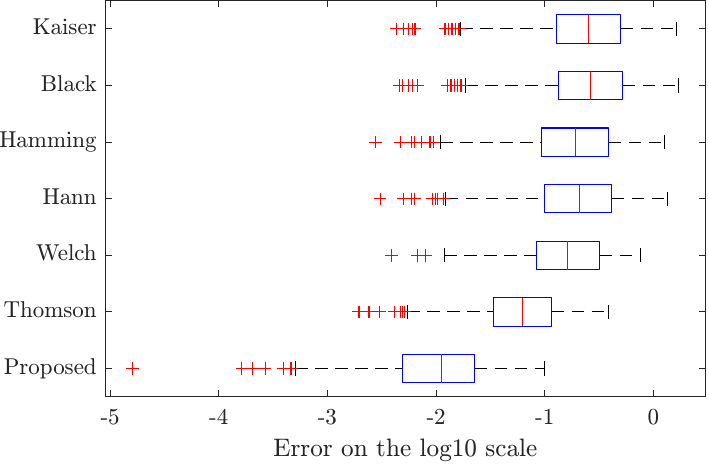}
        \caption{}
        \label{fig:Thomson_experiment2_2}
    \end{subfigure}
    \caption{Boxplots of errors on the $\log_{10}$ scale for the hyperbolic secant mixtures: (a) noiseless, (b) noisy ($\mathrm{SNR}=0$ dB).}
    \label{fig:Thomson_experiment2}
\end{figure}
\begin{figure}[H]
    \centering
    \begin{subfigure}{0.49\textwidth}
        \centering
        \includegraphics[width=\linewidth, height=5.2cm]{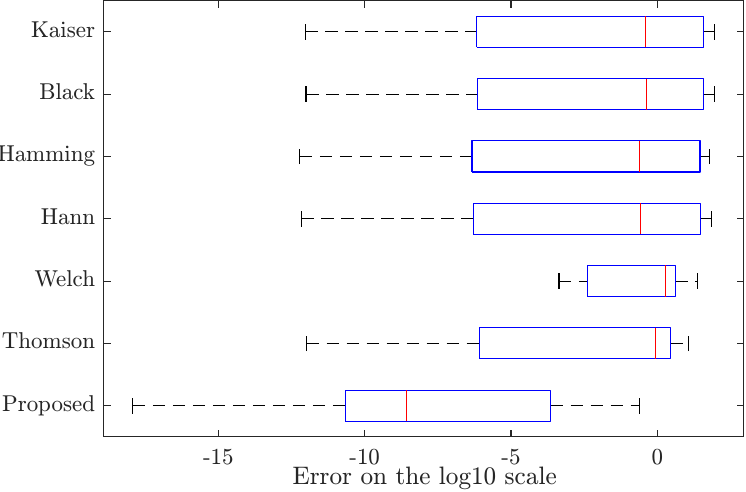}
        \caption{}
        \label{fig:Thomson_experiment3_1}
    \end{subfigure}
    \begin{subfigure}{0.49\textwidth}
        \centering
        \includegraphics[width=\linewidth, height=5.2cm]{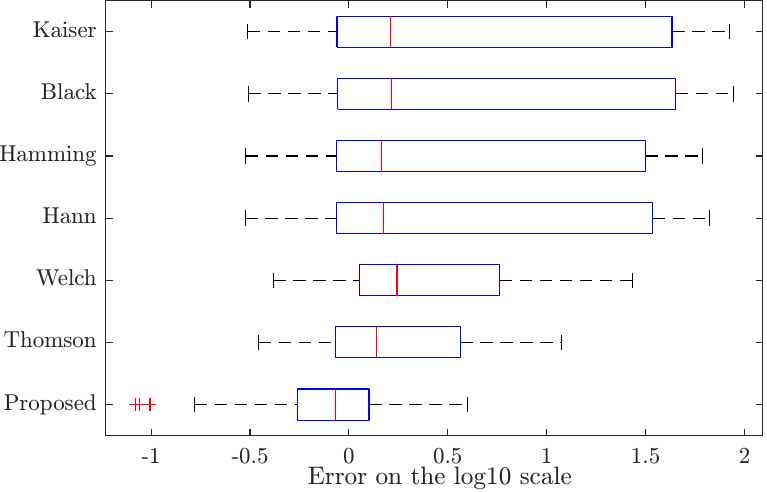}
        \caption{}
        \label{fig:Thomson_experiment3_2}
    \end{subfigure}
    \caption{Boxplots of errors on the $\log_{10}$ scale for the B-spline mixtures: (a) noiseless, (b) noisy ($\mathrm{SNR}=0$ dB).}
    \label{fig:Thomson_experiment3}
\end{figure}
Each boxplot in Figures \ref{fig:Thomson_experiment1}--\ref{fig:Thomson_experiment3} summarizes $15\times 21$ approximation errors on the base-10 logarithmic scale, corresponding to 15 trials over 21 bands. For the integrated-periodogram-type methods, the estimated spectrum is obtained by convolving the true PSD with a window kernel whose finite main lobe and non-negligible side lobes introduce spectral smoothing and leakage. This effect is mild when the true spectrum varies slowly across the kernel width, but becomes much more pronounced when the spectrum contains narrow and sharply localized peaks. Consequently, integrating the smoothed PSD over the target band may lead to substantial deviation from the true band energy. By contrast, the proposed estimator targets the band energy directly through projection onto a time- and band-concentrated subspace spanned by appropriately modulated PSWFs. In the noiseless setting, this yields negligible errors across all three signal families. In the noisy setting, although accuracy degrades for all methods, the proposed estimator remains consistently more accurate than the integrated-periodogram baselines.
}
\subsection{Computational Aspects and Multiple-Band Implementation}
Equation \eqref{eq:computational_complexity} separates the procedure into an offline preprocessing stage and an online estimation stage. The offline stage includes the computation of PSWF handlers. Its cost is implementation-dependent and is incurred only once. The online stage consists of evaluation on the sampling grid and estimator construction and is comparatively lightweight. Since the main computational advantage of the method arises in repeated multiple-band estimation, we focus here on the amortized cost in the reuse setting described in Section 3.1.

We consider the most general reuse setting (Scenario 4), in which both the half-bandwidth $\Omega$ and the observation length $T$ vary. We considered three signal families: Gaussian mixtures parameterized by $\{f_{1j}\in\mathscr R([30,430])\}_{j=1}^{3}$, $\{\sigma_{1j}\in\mathscr R([0.02,0.06])\}_{j=1}^{3}$, and $\{A_{1j}\in\mathscr R([4,5])\}_{j=1}^{3}$; hyperbolic secant mixtures parameterized by $\{f_{2j}\in\mathscr R([100,420])\}_{j=1}^{3}$, $\{\tau_{2j}\in\mathscr R([0.03,0.04])\}_{j=1}^{3}$, and $\{A_{2j}\in\mathscr R([3,6])\}_{j=1}^{3}$; and B-spline mixtures parameterized by $\{f_{3j}\in\mathscr R([50,430])\}_{j=1}^{3}$, $\{\tau_{3j}\in\mathscr R([0.25,0.6])\}_{j=1}^{3}$, $\{r_{3j}\in\mathscr R(\{3,4,5\})\}_{j=1}^{3}$, and $\{A_{3j}\in\mathscr R([3,10])\}_{j=1}^{3}$. We generated $3\times 3$ observation vectors $y=y_{i,j}$, where $i\in[3]$ indexes the mixture family and $j\in[3]$ indexes the $(\Omega,T)$ pair. For Gaussian mixtures, the three $(\Omega,T)$ pairs were $(2\pi\cdot0.5,3)$, $(2\pi\cdot2.5,5)$, and $(2\pi\cdot3,4)$, with sampling frequency $W_{\mathrm{hs}}=2\pi\cdot550$. For hyperbolic secant mixtures, the three $(\Omega,T)$ pairs were $(2\pi\cdot2,2)$, $(2\pi\cdot3,2.5)$, and $(2\pi\cdot4,3)$, with $W_{\mathrm{hs}}=2\pi\cdot500$. For B-spline mixtures, the three $(\Omega,T)$ pairs were $(2\pi\cdot4,3)$, $(2\pi\cdot3,4)$, and $(2\pi\cdot2,5)$, again with $W_{\mathrm{hs}}=2\pi\cdot500$. For each observation vector $y_{i,j}$, we performed estimation over 21 bands. The band-center frequencies, in Hz, were $\{f_{11}+(i-11)\}_{i=1}^{21}$, $\{f_{23}+(i-11)\}_{i=1}^{21}$, and $\{f_{32}+(i-11)\}_{i=1}^{21}$ for the Gaussian, hyperbolic secant, and B-spline mixtures, respectively. In total, the band energy was estimated on $3\cdot 3\cdot 21$ bands.

Among all $(\Omega,T)$ pairs, the largest time--bandwidth product is $c_{\max}=\frac{(2\pi\cdot2.5)\cdot5}{2}$. We therefore computed the PSWF handlers once using $c=c_{\max}$ and then reused them for all target bands through the scaling strategy described in Section 3.1. Grid evaluation was required only nine times because it is performed for each observation vector rather than for each band. Since each identical $(\Omega,T)$ pair is associated with 21 target bands, we report the average elapsed time per band after every 21 processed bands. The results are summarized in Table \ref{tab:avg_time_per_band}.

The average cost per band decreases rapidly as more bands are processed, confirming that the one-time preprocessing cost can be effectively amortized through PSWF reuse. This supports the feasibility of the proposed implementation in repeated multiple-band estimation.
\begin{table}[H]
    \centering
    \caption{Average MATLAB execution time per band in the multiple-band reuse experiment. The processed bands are grouped in blocks of 21, each corresponding to one identical $(\Omega,T)$ pair.}
    \label{tab:avg_time_per_band}
    \begin{tabular}{cc}
    \hline
    Number of processed bands & Average elapsed time per band (ms) \\
    \hline
    21  & 53.9 \\
    42  & 13.9 \\
    63  & 6.2 \\
    84  & 3.5 \\
    105 & 2.3 \\
    126 & 1.6 \\
    147 & 1.2 \\
    168 & 0.891 \\
    189 & 0.706 \\
    \hline
    \end{tabular}
\end{table}
}
\subsection{Effect of the Observation Length $T$}
{This experiment examines how the band $L^2$-norm estimator depends on the observation length $T$. Since the estimator is derived from samples on the interval $\mathcal I_T$, accurate approximation requires that most of the signal energy be concentrated in $\mathcal I_T$, that is, that $\epsilon_T$ in Assumption \ref{assumption:TL} be sufficiently small. The purpose of this experiment is to illustrate the degradation that occurs when $T$ is too short, and the improvement obtained as $T$ increases. We took $x^\star(t)$ to be Gaussian mixtures with parameters $f_1\in\mathscr R([100,120])$, $f_2\in \mathscr R([500,650])$, $f_3\in\mathscr R([1250,1300])$, $\{\sigma_i\in \mathscr R([0.1,0.4]\}_{i=1}^3$, and $\{A_i\in\mathscr R([5,20])\}_{i=1}^3$. We constructed five observation vectors $y=y_i$, $i\in[5]$, using the half sampling frequency $W_{\mathrm{hs}}=2\pi\cdot 1500$ and observation lengths $T_i=i-0.5$. For each $y_i$, we applied the band $L^2$-norm estimator to 21 bands. The half-bandwidths were drawn from $\Omega/(2\pi)\in\mathscr R([1,3])$ Hz and the center frequencies were given by $\{f_1+(i-11)\}_{i=1}^{21}$ Hz.
\begin{figure}[H]
    \centering
    \begin{subfigure}{0.49\textwidth}
        \centering
        \includegraphics[width=\linewidth]{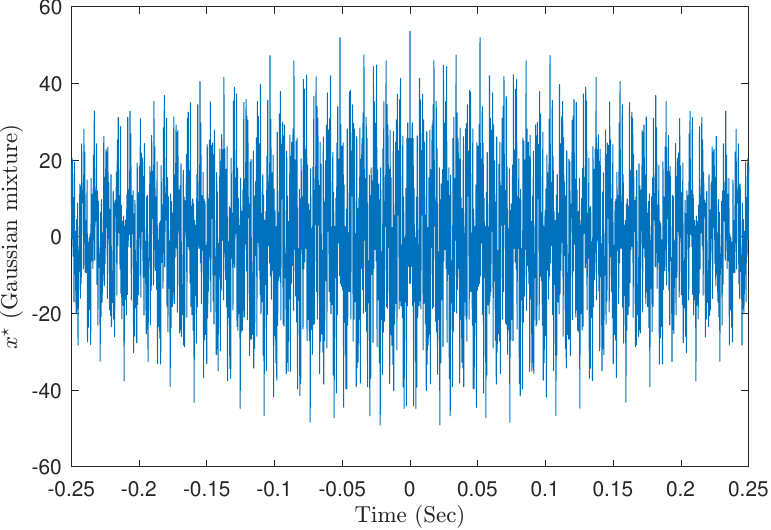}
        \caption{}
        \label{fig:T_1_1}
    \end{subfigure}
    \begin{subfigure}{0.49\textwidth}
        \centering
        \includegraphics[width=\linewidth]{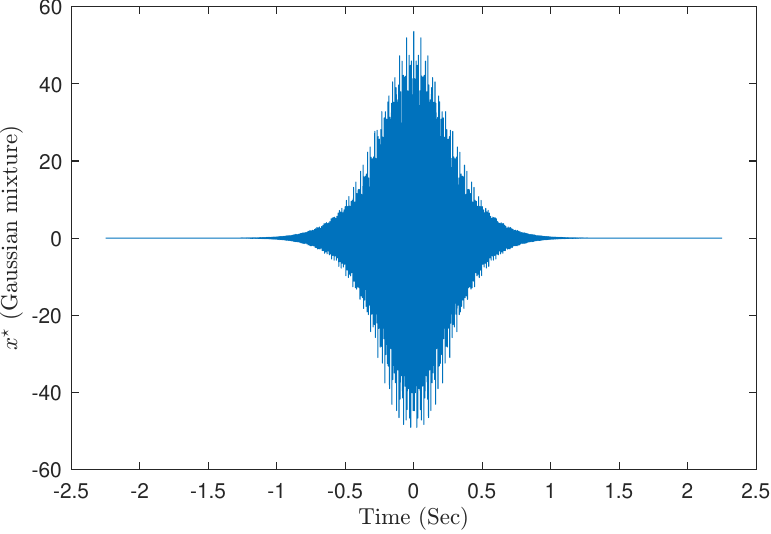}
        \caption{}
        \label{fig:T_1_2}
    \end{subfigure}
    \begin{subfigure}{0.49\textwidth}
        \centering
        \includegraphics[width=\linewidth]{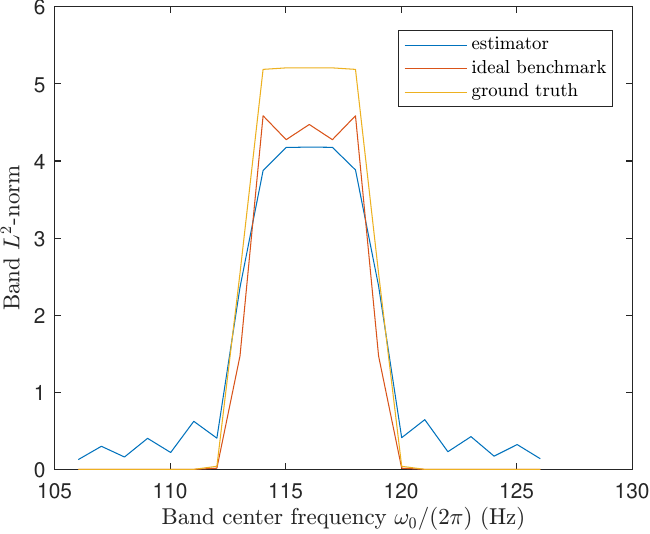}
        \caption{}
        \label{fig:T_1_3}
    \end{subfigure}
    \begin{subfigure}{0.49\textwidth}
        \centering
        \includegraphics[width=\linewidth]{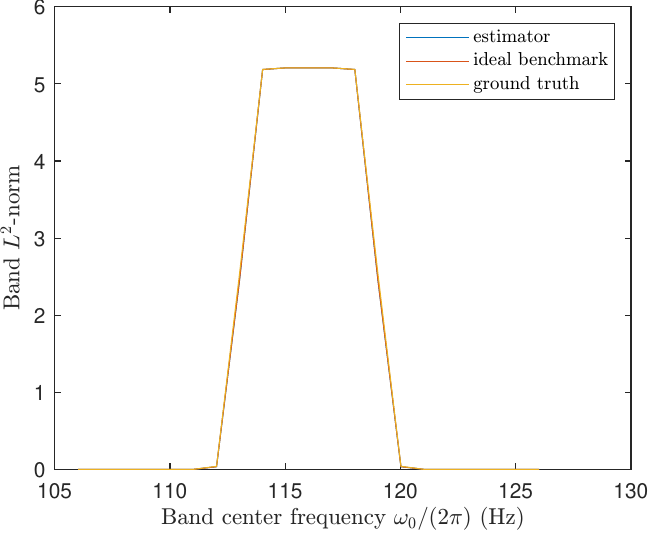}
        \caption{}
        \label{fig:T_1_4}
    \end{subfigure}
    \caption{Panels (a) and (b) display $y_1$ and $y_5$ for $T_1=0.5$ and $T_5=4.5$, respectively. Panels (c) and (d) show the corresponding band $L^2$-norms across the 21 bands computed by the estimator $A^{\mathrm{est}}_{\mathcal B(\omega_0,\Omega)}(y)$, the ideal benchmark $A^{\mathrm{ideal}}_{\mathcal B(\omega_0,\Omega)}(x^\star)$, and the ground truth $A_{\mathcal B(\omega_0,\Omega)}(x^\star)$.}
\end{figure}
Figure \ref{fig:T_1_1} illustrates the case $T_1=0.5$, where the observation interval is too short and the time-localization error is significant. In this regime, both the estimator $A^{\mathrm{est}}_{\mathcal B(\omega_0,\Omega)}(y)$ and the ideal benchmark $A^{\mathrm{ideal}}_{\mathcal B(\omega_0,\Omega)}(x^\star)$ deviate substantially from the true band $L^2$-norm, as shown in Figure \ref{fig:T_1_3}. By contrast, for $T_5=4.5$, most of the signal energy is captured within $\mathcal I_{T_5}$; see Figure \ref{fig:T_1_2}. In this case, the estimator and the ideal benchmark are both nearly indistinguishable from the ground truth across all bands; see Figure \ref{fig:T_1_4}.
\begin{figure}[H]
    \centering
    \begin{subfigure}{0.49\textwidth}
        \centering
        \includegraphics[width=\linewidth]{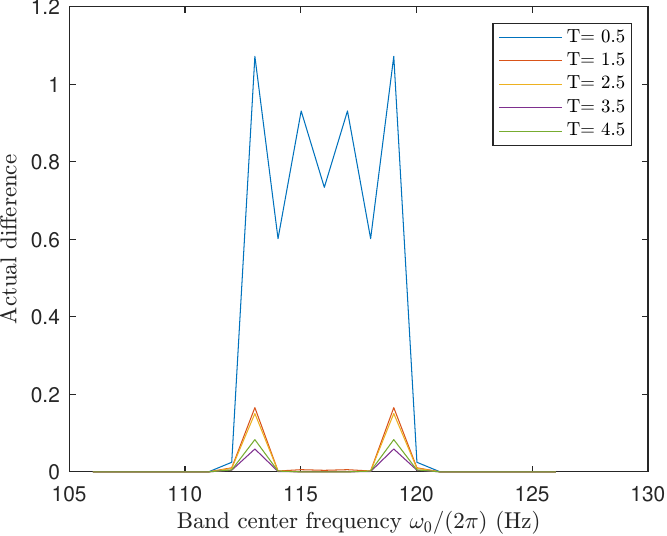}
        \caption{}
        \label{fig:T_2_1}
    \end{subfigure}
    \begin{subfigure}{0.49\textwidth}
        \centering
        \includegraphics[width=\linewidth]{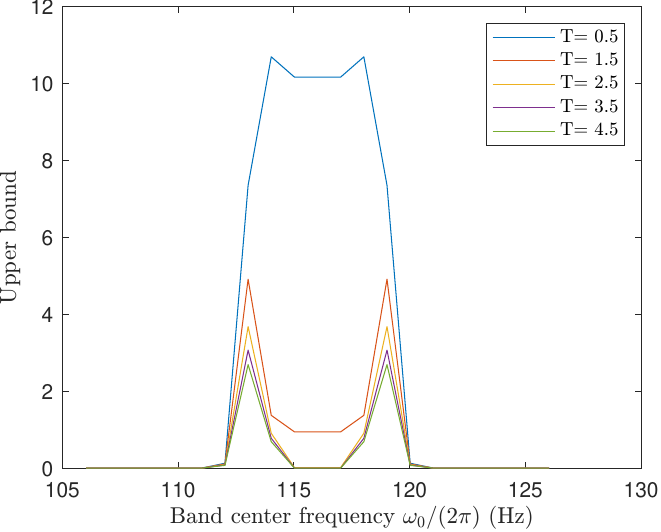}
        \caption{}
        \label{fig:T_2_2}
    \end{subfigure}
    \begin{subfigure}{0.49\textwidth}
        \centering
        \includegraphics[width=\linewidth]{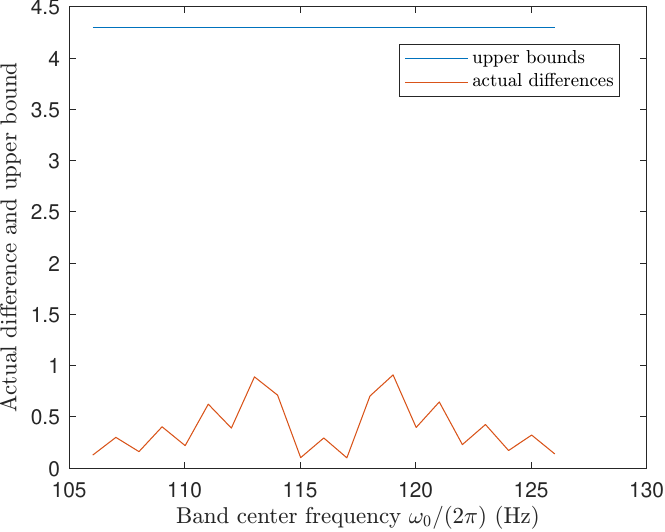}
        \caption{}
        \label{fig:T_2_3}
    \end{subfigure}
    \begin{subfigure}{0.49\textwidth}
        \centering
        \includegraphics[width=\linewidth]{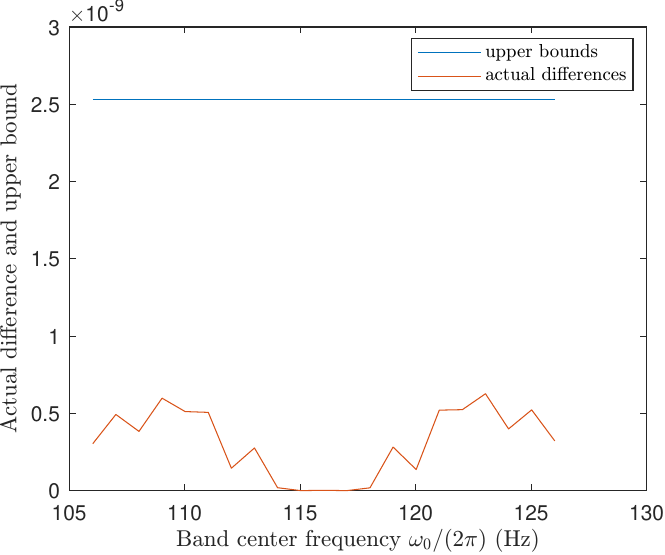}
        \caption{}
        \label{fig:T_2_4}
    \end{subfigure}
    \caption{Panel (a) shows the actual difference $\big|A^{\mathrm{ideal}}_{\mathcal B(\omega_0,\Omega)}(x^\star)-A_{\mathcal B(\omega_0,\Omega)}(x^\star)\big|$, while panel (b) shows the corresponding upper bound determined by the time-limiting, sampling, and band-limiting errors. Panels (c) and (d) compare the actual difference $\big|A^{\mathrm{est}}_{\mathcal B(\omega_0,\Omega)}(y)-A^{\mathrm{ideal}}_{\mathcal B(\omega_0,\Omega)}(x^\star)\big|$ with its upper bound for $T_1=0.5$ and $T_5=4.5$, respectively.}
\end{figure}
Figures \ref{fig:T_2_1}, \ref{fig:T_2_3}, and \ref{fig:T_2_4} show that both $\big|A^{\mathrm{ideal}}_{\mathcal B(\omega_0,\Omega)}(x^\star)-A_{\mathcal B(\omega_0,\Omega)}(x^\star)\big|$ and $\big|A^{\mathrm{est}}_{\mathcal B(\omega_0,\Omega)}(y)-A^{\mathrm{ideal}}_{\mathcal B(\omega_0,\Omega)}(x^\star)\big|$ decrease rapidly as $T$ increases. The theoretical upper bounds consistently dominate the observed errors and capture the same qualitative trend.
\begin{remark}
The infinite-sample benchmark $A^{\infty}_{\mathcal B(\omega_0,\Omega)}(x^\star)$ cannot be computed explicitly, even when both $x^\star$ and $\hat x^\star$ are known. For this reason, we merge the terms $\big|A^{\mathrm{est}}_{\mathcal B(\omega_0,\Omega)}(y)-A^{\infty}_{\mathcal B(\omega_0,\Omega)}(x^\star)\big|$ and $\big|A^{\infty}_{\mathcal B(\omega_0,\Omega)}(x^\star)-A^{\mathrm{ideal}}_{\mathcal B(\omega_0,\Omega)}(x^\star)\big|$
into the single observable quantity $\big|A^{\mathrm{est}}_{\mathcal B(\omega_0,\Omega)}(y)-A^{\mathrm{ideal}}_{\mathcal B(\omega_0,\Omega)}(x^\star)\big|$. 
\end{remark}

The upper bound for $\big|A^{\mathrm{ideal}}_{\mathcal B(\omega_0,\Omega)}(x^\star)-A_{\mathcal B(\omega_0,\Omega)}(x^\star)\big|$ is conservative. This is because our analysis is based on the classical Landau--Pollak approximation result in Lemma \ref{lemma:function_approximation_PSWFs}, which expresses the finite-term PSWF truncation error in terms of the time-tail energy of the band-limited signal. Combined with the frequency-domain Sobolev regularity imposed on $\hat x^\star$ over $\mathcal I_{W_{\mathrm{cut}}}$ and with the general smoothstep functions, this yields an explicit bound on the tail energy and hence on the truncation error. The resulting estimate is sufficiently explicit for analysis and interpretation, even though sharper PSWF approximation results are available in other settings \cite{PSWF8,PSWF4}. Those results are formulated for different regularity assumptions, typically involving time-domain Sobolev spaces on compact intervals, and do not directly address the joint structure used here: frequency-domain Sobolev regularity on $\mathcal I_{W_{\mathrm{cut}}}$ together with time-domain regularity on $\mathbb R\setminus\mathcal I_T$.
}
\subsection{Effect of the Half-Bandwidth $\Omega$}
{This experiment examines how the proposed band $L^2$-norm estimator behaves when the half-bandwidth $\Omega$ of the target band $\mathcal B(\omega_0,\Omega)$ varies. We used hyperbolic secant mixtures as the input signal $x^\star(t)$, with parameters $f_1 \in \mathscr R([120,140])$, $f_2 \in \mathscr R([310,480])$, $f_3\in\mathscr R([850,950])$, $\{\tau_i\in\mathscr R([0.055,0.065])\}_{i=1}^3$, and $\{A_i\in\mathscr R([10,20])\}_{i=1}^3$. The observation vector $y$ was constructed with $W_{\mathrm{hs}}=2\pi\cdot 1100$ and $T = 0.8$. We applied the band $L^2$-norm estimator to estimate the band $L^2$-norm over 21 bands whose center frequencies, in Hz, were $\{f_3+(i-11)\}_{i=1}^{21}$. Five half bandwidths were considered, with $\Omega_i/(2\pi)\in\mathscr R([2i, 2i+1]))$, $i\in[5]$.

Figures \ref{fig:B_3} and \ref{fig:B_4} show that the empirical errors broadly follow the shapes of the corresponding upper bounds. Although Theorem \ref{theorem:main_theorem} indicates that the time-limiting, sampling, and band-limiting terms increase with $\Omega$, the total approximation error does not exhibit a simple monotone dependence on $\Omega$. This indicates that, in the present example, the finite-term truncation term is dominant.

Corollary \ref{corollary:bounded_tail_energy} explains this behavior through the quantities $\gamma^+_{\hat x^\star,s,\omega_0,\Omega}$ in \eqref{eq:gamma_ft} and the side regions $\mathcal I^+_{\omega_0,\Omega,\kappa}$ in \eqref{eq:I_ft}. As seen in Figure \ref{fig:B_2}, the derivatives of $\hat x^\star(\omega)$ are close to zero over the relevant frequency range, so that $\gamma^+_{\hat x^\star,s,\omega_0,\Omega}$ is essentially determined by the magnitude of $\hat x^\star$ on the side regions. Consequently, the truncation error is governed less by the width of the target band itself than by the spectral values near the corresponding side intervals. This mechanism is visible in Figures \ref{fig:B_3} and \ref{fig:B_4}. For example, when the band center is near $f_3$, the side regions associated with smaller and larger half-bandwidths probe different parts of the spectrum, leading to markedly different approximation errors and upper bounds. The same phenomenon appears when the band center is shifted away from $f_3$. Overall, the experiment confirms that the dependence on $\Omega$ is mediated primarily through the spectral behavior on the side regions entering Corollary \ref{corollary:bounded_tail_energy}, rather than through a simple monotone scaling in $\Omega$.
}

\begin{figure}[H]
    \centering
    \begin{subfigure}{0.49\textwidth}
        \centering
        \includegraphics[width=\linewidth]{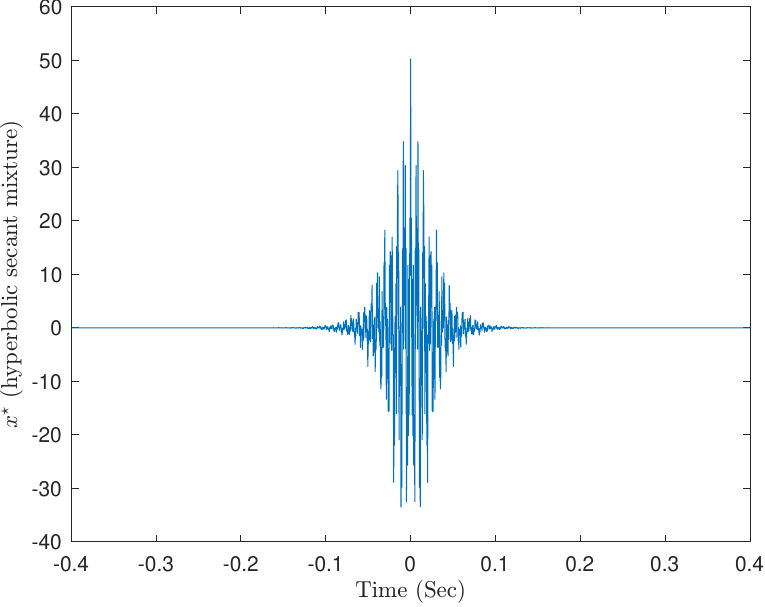}
        \caption{}
        \label{fig:B_1}
    \end{subfigure}
    \begin{subfigure}{0.49\textwidth}
        \centering
        \includegraphics[width=\linewidth]{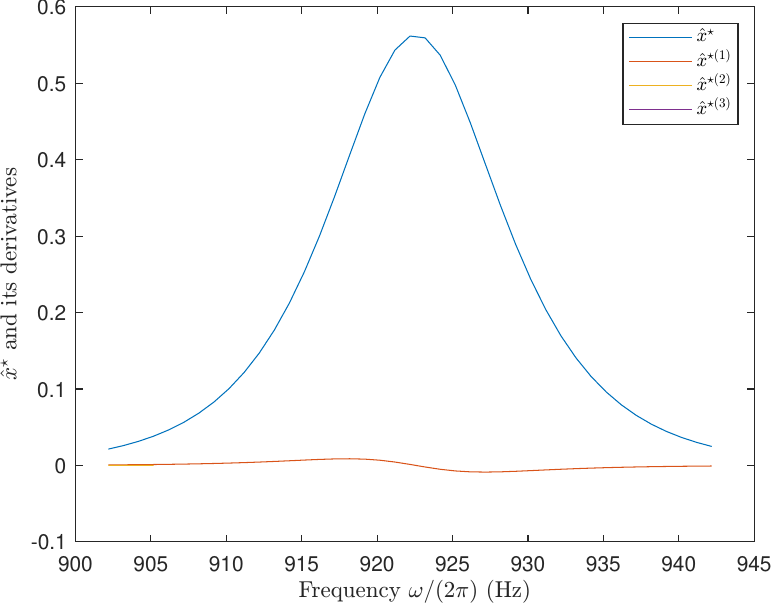}
        \caption{}
        \label{fig:B_2}
    \end{subfigure}
    \begin{subfigure}{0.49\textwidth}
        \centering
        \includegraphics[width=\linewidth, height=6.1cm]{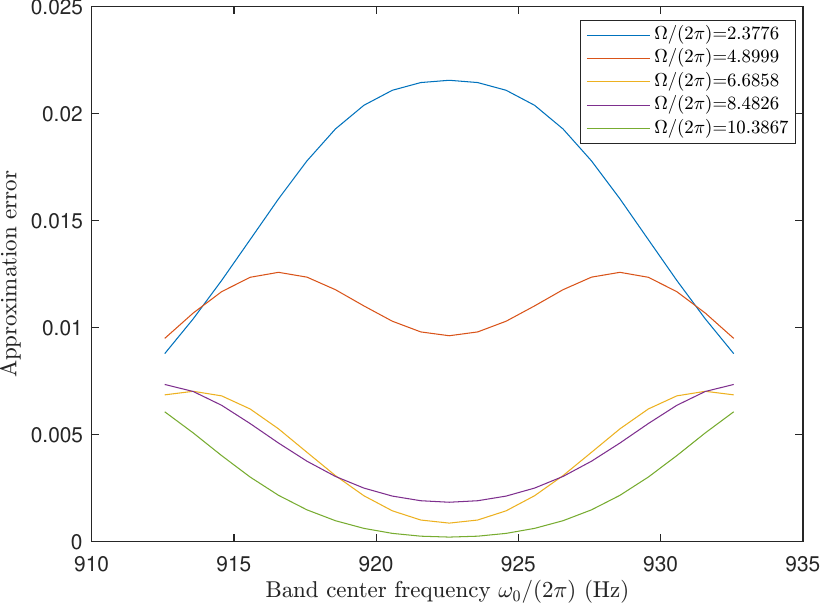}
        \caption{}
        \label{fig:B_3}
    \end{subfigure}
    \begin{subfigure}{0.49\textwidth}
        \centering
        \includegraphics[width=\linewidth, height=6.1cm]{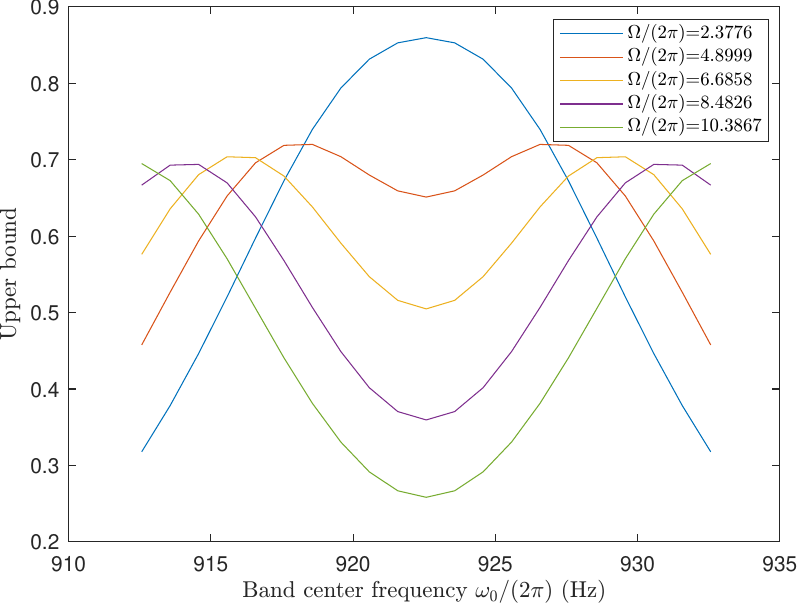}
        \caption{}
        \label{fig:B_4}
    \end{subfigure}
    \caption{Panel (a) shows the waveform $y$. Panel (b) displays $\hat x^\star(\omega)$ and its derivatives over the 21 bands. Panels (c) and (d) show the approximation error $\big|A_{\mathcal B(\omega_0,\Omega)}(x^\star)-A^{\mathrm{est}}_{\mathcal B(\omega_0,\Omega)}(y)\big|$ and its upper bound, respectively, for five half-bandwidths. The comparison illustrates the dependence of the error on the spectral behavior over the enlarged band $\mathcal B(\omega_0,\Omega+\kappa)$ and the side regions $\mathcal I^+_{\omega_0,\Omega,\kappa}$ in Corollary \ref{corollary:bounded_tail_energy}.}
\end{figure}

\subsection{Effect of Additive Noise}\label{sec:noise_experiment}
{This experiment investigates how additive noise affects the band $L^2$-norm estimator. We considered all three signal families as the underlying signals $x^\star(t)$. The Gaussian mixtures were parameterized by $\{f_{1j}\in\mathscr R([50,460])\}_{j=1}^3$, $\{\sigma_{1j}\in\mathscr R([0.08,0.21])\}_{j=1}^3$, and $\{A_{1j}\in\mathscr R([4,25])\}_{j=1}^3$; the hyperbolic secant mixtures by $\{f_{2j}\in\mathscr R([100,580])\}_{j=1}^3$, $\{\tau_{2j}\in\mathscr R([0.009,0.03])\}_{j=1}^3$, and $\{A_{2j}\in\mathscr R([3,20])\}_{j=1}^3$; and the B-spline mixtures by $\{f_{3j}\in\mathscr R([110,690])\}_{j=1}^3$, $\{\tau_{3j}\in\mathscr R([0.4,0.65])\}_{j=1}^3$, $\{r_{3j}\in\mathscr R(\{3,4,5\})\}_{j=1}^{3}$, and $\{A_{3j}\in\mathscr R([7,14])\}_{j=1}^3$. Three clean observation vectors $y_i$, $i\in[3]$, were generated by sampling the Gaussian, hyperbolic secant, and B-spline mixtures with $(W_{\mathrm{hs},1},T_1)=(2\pi\cdot600,2.5)$, $(W_{\mathrm{hs},2},T_2)=(2\pi\cdot700,3)$, and $(W_{\mathrm{hs},3},T_3)=(2\pi\cdot800,4.5)$, respectively. Each clean vector was then contaminated by additive white Gaussian noise with $\mathrm{SNR}\in\{-10,-5,0,5,10\}$ dB. For each simulation, we estimated the band $L^2$-norm over 21 bands. The half-bandwidths were set to 4, 3, and 2 Hz for the Gaussian, hyperbolic secant, and B-spline mixtures, respectively, and the corresponding band-center frequencies, in Hz, were $\{f_{11}+(i-11)\}_{i=1}^{21}$, $\{f_{23}+(i-11)\}_{i=1}^{21}$, and $\{f_{32}+(i-11)\}_{i=1}^{21}$.

In \ref{proof:corollary_noise_perturbation_analysis}, we derive 
\begin{equation*}
    \big|A^{\mathrm{est}}_{\mathcal B(\omega_0,\Omega)}(y_{\eta})-A_{\mathcal B(\omega_0,\Omega)}(x^\star)\big|\leq\big|A^{\mathrm{est}}_{\mathcal B(\omega_0,\Omega)}(y)-A_{\mathcal B(\omega_0,\Omega)}(x^\star)\big|+\Delta\sqrt{\sum\limits^{\lfloor\frac{2c}{\pi}\rfloor}_{n=0}\big|C^{(m)}_n(\eta,\omega_0,\Omega)\big|^2}    
\end{equation*}
and hence
\begin{equation*}
    \big|A^{\mathrm{est}}_{\mathcal B(\omega_0,\Omega)}(y_{\eta})-A_{\mathcal B(\omega_0,\Omega)}(x^\star)\big|\leq\big|A^{\mathrm{est}}_{\mathcal B(\omega_0,\Omega)}(y)-A_{\mathcal B(\omega_0,\Omega)}(x^\star)\big|+\varepsilon\sqrt{\Delta\sum\limits^{\lfloor\frac{2c}{\pi}\rfloor}_{n=0}\Big(\frac{\Omega^2\Delta^2}{2}+\min\{1,\lambda_n(c)+\frac{\Omega\Delta}{\pi}\}\Big)}.
\end{equation*}
To isolate the effect of noise, we retain the actual noiseless approximation error $\big|A^{\mathrm{est}}_{\mathcal B(\omega_0,\Omega)}(y)-A_{\mathcal B(\omega_0,\Omega)}(x^\star)\big|$
instead of replacing it by the upper bound from Theorem \ref{theorem:main_theorem}. The first inequality therefore yields a tight bound, while the second gives a more loose worst-case bound.
\begin{figure}[H]
    \centering
    \begin{subfigure}{0.32\textwidth}
        \centering
        \includegraphics[width=\linewidth]{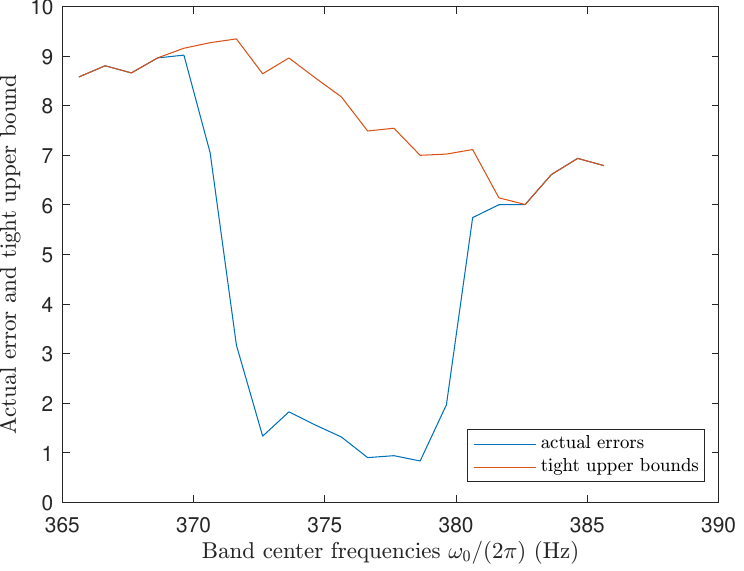}
        \caption{}
        \label{fig:noise_Gaussian_1}
    \end{subfigure}
    \begin{subfigure}{0.32\textwidth}
        \centering
        \includegraphics[width=\linewidth]{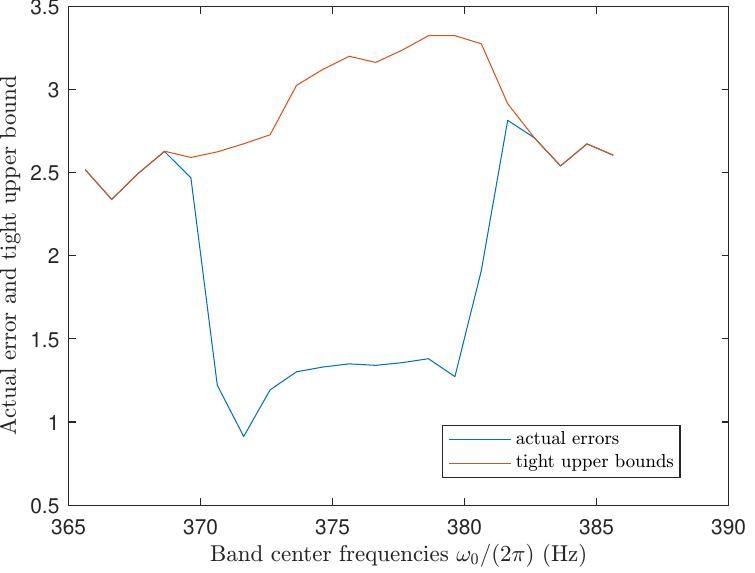}
        \caption{}
        \label{fig:noise_Gaussian_2}
    \end{subfigure}
    \begin{subfigure}{0.32\textwidth}
        \centering
        \includegraphics[width=\linewidth]{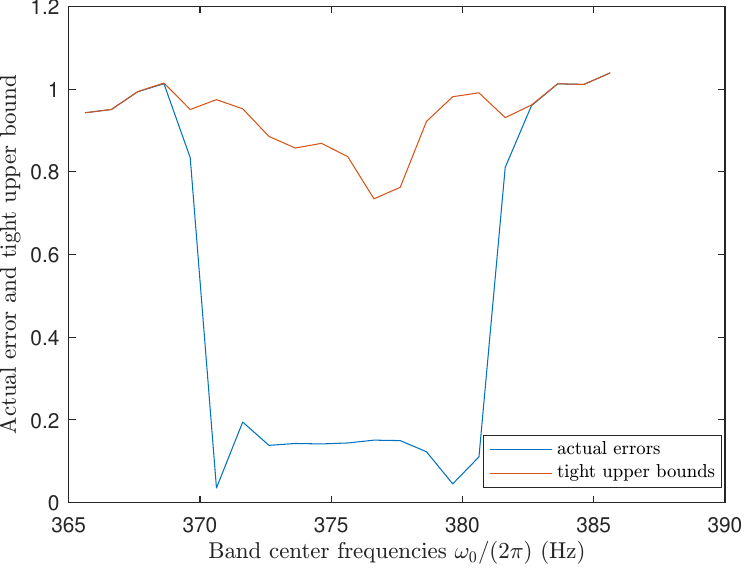}
        \caption{}
        \label{fig:noise_Gaussian_3}
    \end{subfigure}
    \begin{subfigure}{0.32\textwidth}
        \centering
        \includegraphics[width=\linewidth, height=3.8cm]{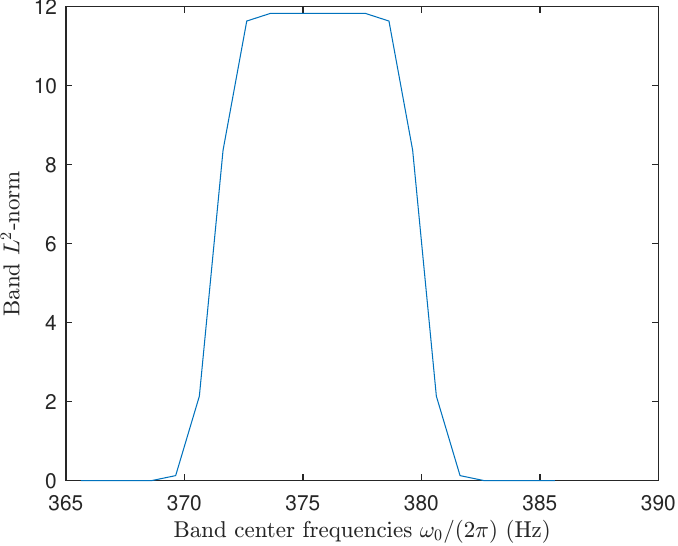}
        \caption{}
        \label{fig:noise_Gaussian_4}
    \end{subfigure}
    \begin{subfigure}{0.32\textwidth}
        \centering
        \includegraphics[width=\linewidth, height=3.8cm]{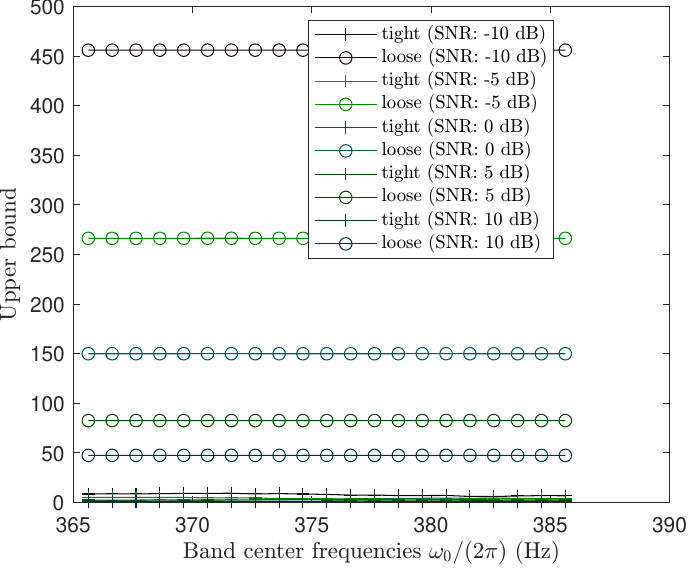}
        \caption{}
        \label{fig:noise_Gaussian_5}
    \end{subfigure}
    \begin{subfigure}{0.32\textwidth}
        \centering
        \includegraphics[width=\linewidth, height=3.8cm]{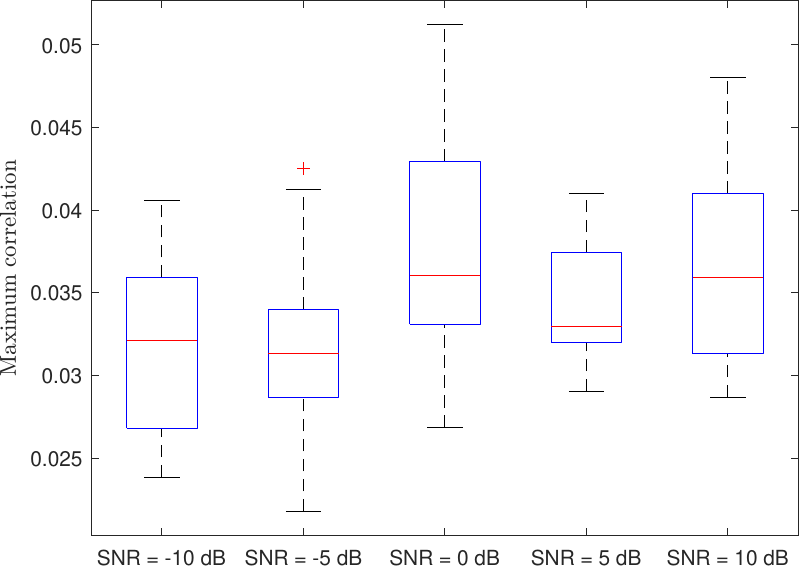}
        \caption{}
        \label{fig:noise_Gaussian_6}
    \end{subfigure}
    \caption{Gaussian-mixture example. Panels (a)–(c) show the actual approximation errors of the band $L^2$-norm estimator together with their tight upper bounds for $\mathrm{SNR}=-10,0$, and $10$ dB. Panel (d) shows the true band $L^2$-norm. Panel (e) compares the tight and loose upper bounds across SNR levels. Panel (f) reports the maximum correlations between the noise and sampled modulated PSWFs. The approximation error exhibits a pronounced dip in bands where the true band $L^2$-norm is large. Moreover, the noise correlates weakly with all bands, leading to a visible gap between the tight and loose upper bounds.}
    \label{fig:noise_Gaussian}
\end{figure} 
\begin{figure}[H]
    \centering
    \begin{subfigure}{0.32\textwidth}
        \centering
        \includegraphics[width=\linewidth]{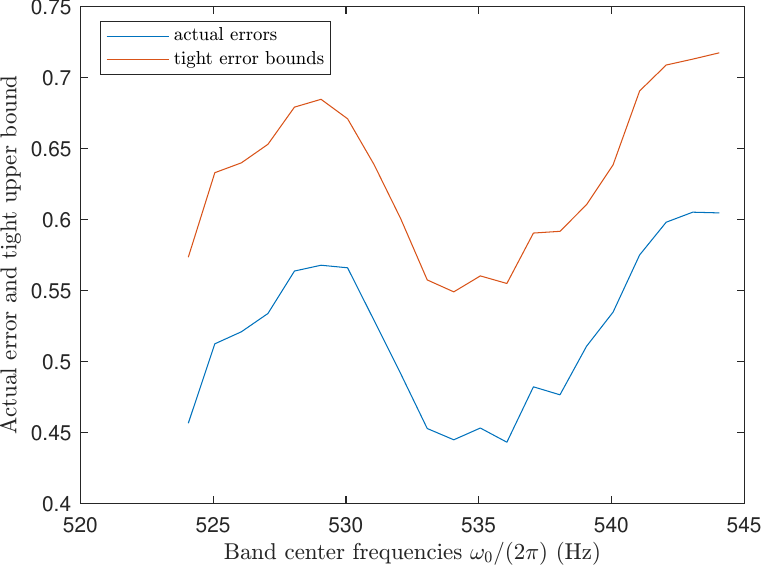}
        \caption{}
        \label{fig:noise_hyperbolic_secant_1}
    \end{subfigure}
    \begin{subfigure}{0.32\textwidth}
        \centering
        \includegraphics[width=\linewidth]{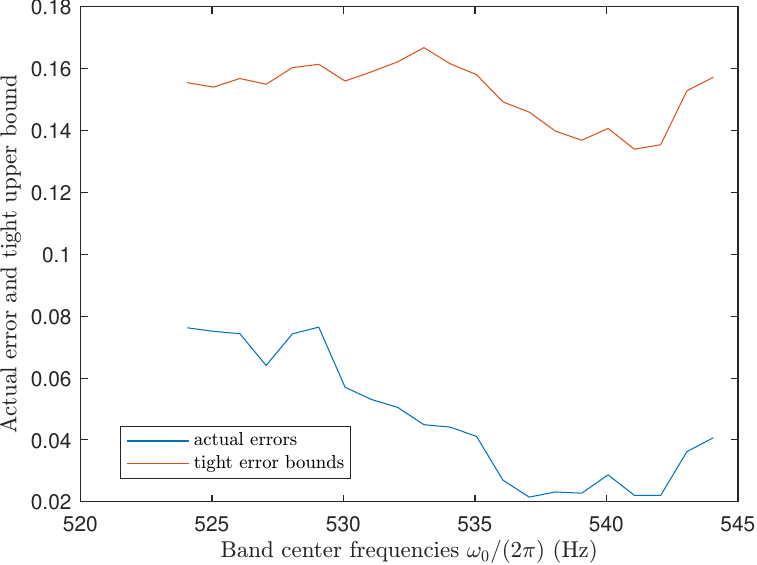}
        \caption{}
        \label{fig:noise_hyperbolic_secant_2}
    \end{subfigure}
    \begin{subfigure}{0.32\textwidth}
        \centering
        \includegraphics[width=\linewidth]{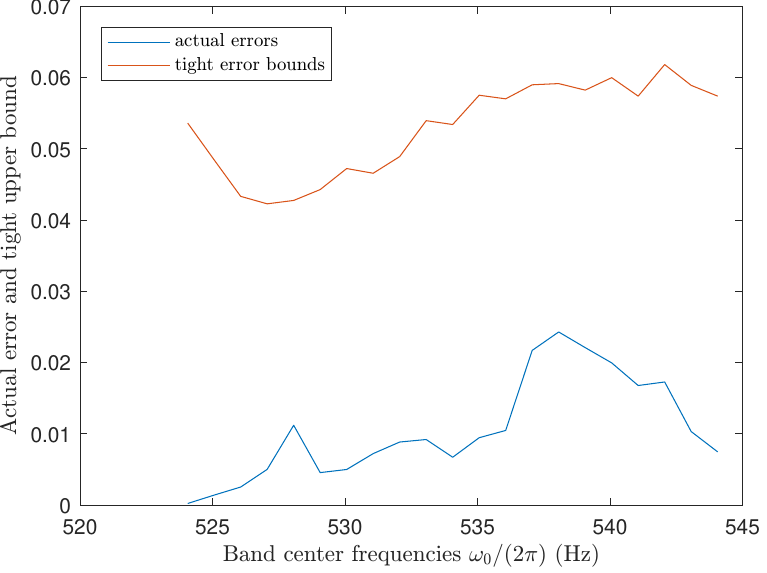}
        \caption{}
        \label{fig:noise_hyperbolic_secant_3}
    \end{subfigure}
    \begin{subfigure}{0.32\textwidth}
        \centering
        \includegraphics[width=\linewidth, height=3.8cm]{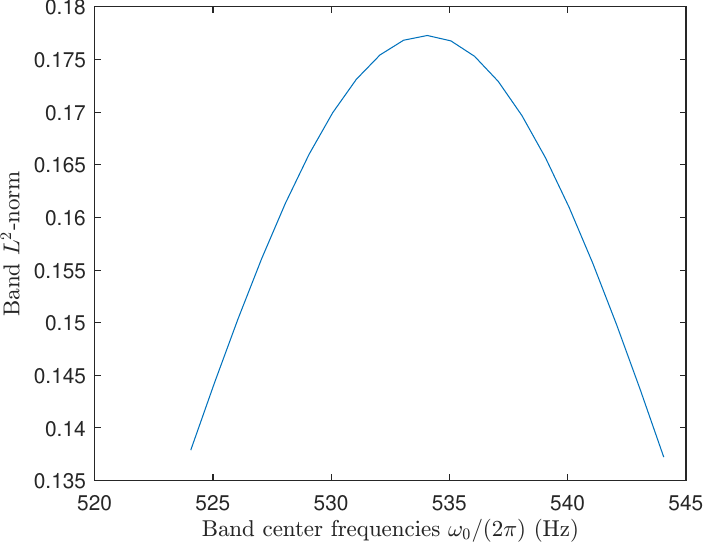}
        \caption{}
        \label{fig:noise_hyperbolic_secant_4}
    \end{subfigure}
    \begin{subfigure}{0.32\textwidth}
        \centering
        \includegraphics[width=\linewidth, height=3.8cm]{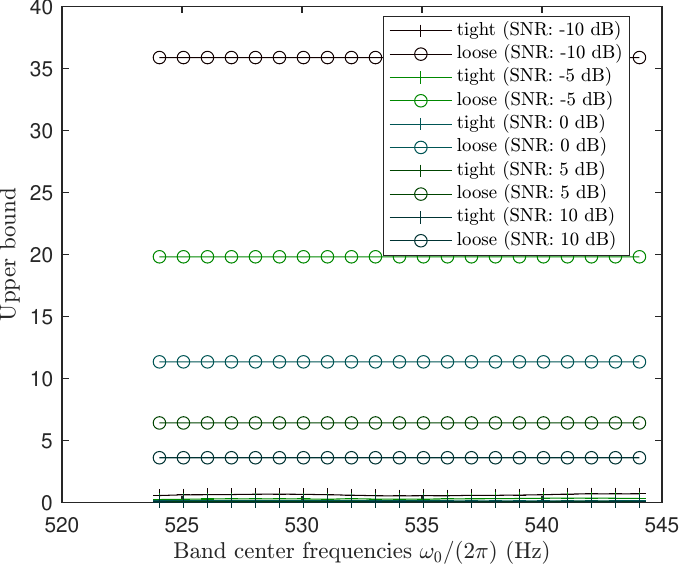}
        \caption{}
        \label{fig:noise_hyperbolic_secant_5}
    \end{subfigure}
    \begin{subfigure}{0.32\textwidth}
        \centering
        \includegraphics[width=\linewidth, height=3.8cm]{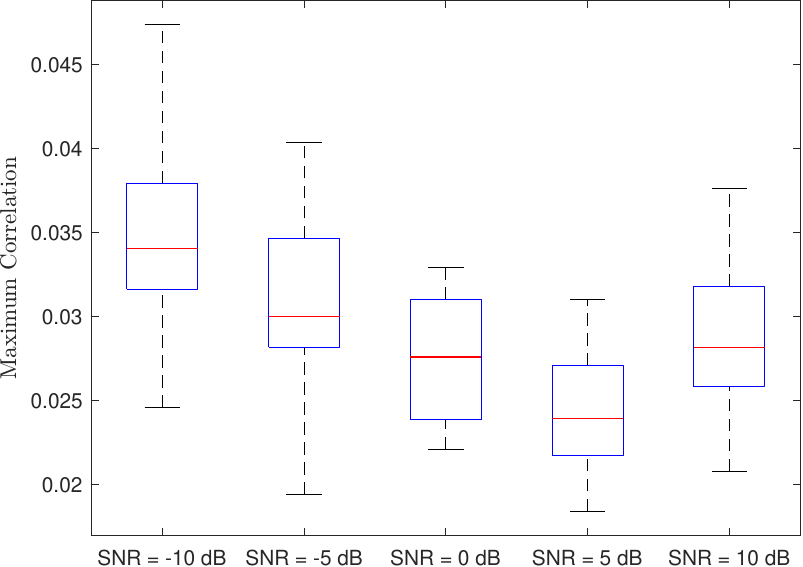}
        \caption{}
        \label{fig:noise_hyperbolic_secant_6}
    \end{subfigure}
    \caption{Hyperbolic secant mixtures. The six panels have the same interpretation as in the Gaussian-mixture case. In contrast to the Gaussian-mixture case, the approximation error follows the tight upper bound more closely because the true band $L^2$-norm is comparatively flat across bands.}
    \label{fig:noise_hyperbolic_secant}
\end{figure}

Figures \ref{fig:noise_Gaussian_1}--\ref{fig:noise_Gaussian_3} reveal an interesting phenomenon: the approximation error exhibits a pronounced dip in those bands where the true band $L^2$-norm is large. To explain this, write \begin{align*}
&\big|A^{\mathrm{est}}_{\mathcal B(\omega_0,\Omega)}(y_{\eta})-A_{\mathcal B(\omega_0,\Omega)}(x^\star)\big|\\
=&\,\Bigg|\Delta\sqrt{\sum_{n=0}^{\lfloor\frac{2c}{\pi}\rfloor}\big|C^{(m)}_n(y,\omega_0,\Omega)\big|^2+2\sum_{n=0}^{\lfloor\frac{2c}{\pi}\rfloor}\Re\langle C^{(m)}_n(y,\omega_0,\Omega),C^{(m)}_n(\eta,\omega_0,\Omega)\rangle+\sum_{n=0}^{\lfloor\frac{2c}{\pi}\rfloor}\big|C^{(m)}_n(\eta,\omega_0,\Omega)\big|^2
}-A_{\mathcal B(\omega_0,\Omega)}(x^\star)\Bigg|\\
=&\,\Bigg|\sqrt{E^{\mathrm{est}}_{\mathcal B(\omega_0,\Omega)}(y) +\delta^{(m)}_{\mathcal B(\omega_0,\Omega)}(y,\eta)  +E^{\mathrm{est}}_{\mathcal B(\omega_0,\Omega)}(\eta)}
-A_{\mathcal B(\omega_0,\Omega)}(x^\star)\Bigg|,
\end{align*}
where $\delta^{(m)}_{\mathcal B(\omega_0,\Omega)}(y,\eta):=2\Delta^2\sum\limits^{\lfloor\frac{2c}{\pi}\rfloor}_{n=0}\Re\langle C^{(m)}_n(y,\omega_0,\Omega),C^{(m)}_n(\eta,\omega_0,\Omega)\rangle$ is the cross term. For white Gaussian noise, the energy of $\eta$ is spread nearly uniformly across frequency, so $E^{\mathrm{est}}_{\mathcal B(\omega_0,\Omega)}(\eta)$ is approximately flat across bands. The cross term also fluctuates only mildly, since the vector $\big(C^{(m)}_n(\eta,\omega_0,\Omega)\big)_{n=0,\ldots,\lfloor 2c/\pi\rfloor}$ is approximately isotropic and therefore has only weak correlation with $\big(C^{(m)}_n(y,\omega_0,\Omega)\big)_{n=0,\ldots,\lfloor 2c/\pi\rfloor}$ when $\lfloor 2c/\pi\rfloor$ is sufficiently large. As a result, the error behaves approximately like an expression of the form $|\sqrt{Q^2+V}-Q|$, where $Q$ represents the true band $L^2$-norm and $V$ varies comparatively slowly across bands. When $Q$ varies much more strongly than $V$, the error necessarily dips in regions where $Q$ is large. The same behavior is observed for the B-spline mixtures; see Figures \ref{fig:noise_B_splines_1}--\ref{fig:noise_B_splines_4}. By contrast, for the hyperbolic secant mixtures, the true band $L^2$-norm in Figure \ref{fig:noise_hyperbolic_secant_4} is comparatively flat, so the influence of $V$ is more visible. In that case, the approximation error follows the trend of the tight bound much more closely; see Figures \ref{fig:noise_hyperbolic_secant_1}--\ref{fig:noise_hyperbolic_secant_3}.

A second discrepancy between the worst-case theory and the empirical results is seen in Figures \ref{fig:noise_Gaussian_5}, \ref{fig:noise_hyperbolic_secant_5}, and \ref{fig:noise_B_splines_5}. To bound $\Delta\sqrt{\sum\limits^{\lfloor\frac{2c}{\pi}\rfloor}_{n=0}\big|C^{(m)}_n(\eta,\omega_0,\Omega)\big|^2}$ we apply the Cauchy--Schwarz inequality to the inner products between the noise vector $\eta$ and the sampled modulated PSWFs $\big(\psi_n(c,t_j)e^{-i\omega_0 t_j}\big)_{j\in[m]}$, $n=0,\ldots,\lfloor 2c/\pi \rfloor$. Figures \ref{fig:noise_Gaussian_6}, \ref{fig:noise_hyperbolic_secant_6}, and \ref{fig:noise_B_splines_6} show that the maximal correlations are small: approximately $0.032$--$0.036$ for the Gaussian mixtures, $0.025$--$0.034$ for the hyperbolic secant mixtures, and $0.023$--$0.025$ for the B-spline mixtures. Consequently, the loose bounds are substantially larger than the tight bounds, but their gap is still explained by these correlation factors.

Overall, the experiment confirms the deterministic stability result, while also showing that in the white-noise setting the worst-case bounds are conservative. This suggests that sharper noise-stability estimates may be possible under additional structural assumptions on the perturbation.
}


\section{Conclusion}
In this work, we formulated direct band-energy estimation from finite sampled data as a finite-sample approximation problem for a spectral functional and proposed a multitaper-type estimator based on sampled continuous PSWFs. The estimator is constructed by first estimating the band $L^2$-norm and then squaring it to obtain the band energy. Under explicit time- and frequency-domain regularity and energy-concentration assumptions, we established deterministic and non-asymptotic error bounds with a four-term decomposition that separates the effects of time limitation, sampling, band limitation, and finite-term truncation. We further derived a deterministic stability result under additive noise. Moreover, for efficient computation, we developed a PSWF-reuse strategy that makes the proposed procedure effective in repeated multiple-band estimation. Numerical experiments confirmed the favorable performance of the proposed method relative to existing integrated periodicogram-type methods, illustrated the computational advantages of the proposed strategy, and corroborated the validity of theoretical predictions with respect to observation length, band selection, and additive noise. Future work includes developing probabilistic error guarantees under stochastic noise models on the basis of the deterministic framework established here.


\appendix
\section{General Smoothstep Functions}\label{sec:general_smoothstep}
{The third-order smoothstep function was introduced in \cite{Smoothstep1}, and the general $n$th order smoothstep function in \cite{Smoothstep2}. We adopt the scaled version of \cite{Smoothstep2}, given by
\begin{equation}
    \label{smoothstep}
    S_n(x):=\begin{cases}
        0&\text{if}\,\,x\leq0\\
        \sum\limits^n_{k=0}\binom{n+k}{k}\binom{2n+1}{n-k}(-1)^kx^{k+n+1}&\text{if}\,\,0\leq x\leq1\\1&\text{if}\,\,x\geq1
    \end{cases},\quad n\in\mathbb N.
\end{equation}
Thus, $S_n(x)$ provides a smooth transition between 0 and 1. To the best of our knowledge, explicit closed-form expressions for the derivatives of $S_n$ and the $L^1$- and $L^{\infty}$-norms of these derivatives did not appear in the literature. We compile and derive them as follows.
\begin{enumerate}
\item[\textbf{1.}]\textbf{Smoothness.}
The function $S_n$ is $n$-times differentiable. In particular,
\begin{align}
    \label{eq:smoothstep_first_derivative}
    &S_n'(x)=\frac{(2n+1)!}{(n!)^2}x^n(1-x)^n,\\
    \label{eq:smoothstep_m_derivatives}
    &S_n^{(m)}(x) = \frac{(2n+1)!}{(n!)^2}(m-1)!x^{n-m+1}(1-x)^{n-m+1}P^{(n-m+1,n-m+1)}_{m-1}(1-2x),
\end{align}
for $0\le x\le 1$, $2\le m\le n$, where $P^{(n-m+1,n-m+1)}_{m-1}(x)$ is a Jacobi polynomial \cite{Jacobi}.
Equation \eqref{eq:smoothstep_first_derivative} follows from the binomial theorem, while \eqref{eq:smoothstep_m_derivatives} follows from \eqref{eq:smoothstep_first_derivative} and the Rodrigues formula (Eq. (4.3.1) of \cite{Jacobi}) as follows:
\begin{equation*}
    P^{(\alpha,\beta)}_{r}(u)=\frac{(-1)^r}{2^rr!}(1-u)^{-\alpha}(1+u)^{-\beta}\frac{d^r}{du^r}[(1-u)^{\alpha+r}(1+u)^{\beta+r}],
\end{equation*}
with $u=1-2x$, $r=m-1$, $\alpha=n-m+1$, and $\beta=n-m+1$.
\item[\textbf{2.}] \textbf{$L^1$-integrability.}  
The functions $S_n$ and $S_n^{(m)}$ for $1\le m\le n$ all belong to $L^1(\mathbb R)$. In particular,
\begin{equation}
    \label{eq:smoothstep_derivatives_L_1}
    \begin{cases}
        \|S_n\|_{L^1([0,1])}=\frac{1}{2}\\\|S_n^{(m)}\|_{L^1(\mathbb R)}=\frac{(2n+1)!}{n!}\frac{(n-m+1)!}{(2n-2m+3)!},&1\leq m\leq n
    \end{cases}.
\end{equation}
If $2\le m\le n$, the following bound can be derived:
\begin{equation*}
    \begin{aligned}
        \|S_n^{(m)}\|_{L^1(\mathbb R)}&=\frac{(2n+1)!}{(n!)^2}(m-1)!\int^1_0|x^{n-m+1}(1-x)^{n-m+1}P^{(n-m+1,n-m+1)}_{m-1}(1-2x)|dx\\&\leq\frac{(2n+1)!}{(n!)^2}(m-1)!\sup\limits_{x\in[-1,1]}|P^{(n-m+1,n-m+1)}_{m-1}(x)|B(n-m+2,n-m+2),\\&\hspace{0.6cm}\text{where}\,\,B(p,q)\,\,\text{is the Beta function}\,\text{\cite{Jacobi}}.\\&=\frac{(2n+1)!}{(n!)^2}(m-1)!P^{(n-m+1,n-m+1)}_{m-1}(1)\frac{(n-m+1)!(n-m+1)!}{(2n-2m+3)!}\\&=\frac{(2n+1)!}{(n!)^2}(m-1)!\frac{n!}{(m-1)!(n-m+1)!}\frac{(n-m+1)!(n-m+1)!}{(2n-2m+3)!}\\&=\frac{(2n+1)!}{n!}\frac{(n-m+1)!}{(2n-2m+3)!}
    \end{aligned}
\end{equation*}  
The first inequality follows from H\"older’s inequality. The second equality uses the definition of the Beta function together with the fact that $P^{(n-m+1,n-m+1)}_{m-1}$ is maximal at $x=1$ (Sec. 7.32 of \cite{Jacobi}). The final equality uses Eq. (4.1.1) of \cite{Jacobi}.  
For $m=1$, $\|S_n'\|_{L^1(\mathbb R)}=\int^1_0|S_n'(x)|dx=S_n(1)-S_n(0)=1$, which is consistent with the above bound when $m=1$. For $m=0$, note that $S_n(x)+S_n(1-x)$ is constant because $S_n'(x)=S_n'(1-x)$. Evaluating at $x=0$ yields $S_n(x)+S_n(1-x)=S_n(0)+S_n(1)=1$. Thus, $\|S_n(x)\|_{L^1([0,1])}=\int^1_0S_n(x)dx=\int^1_0(1-S_n(1-x))dx=1-\int^1_0S_n(x)dx$, implying $\|S_n\|_{L^1([0,1])}=\frac{1}{2}$.
\item[\textbf{3.}] \textbf{$L^{\infty}$-integrability.}  
The functions $S_n$ and $S_n^{(m)}$ for $1\le m\le n$ all belong to $L^{\infty}(\mathbb R)$. Specifically,
\begin{equation}
    \label{eq:smoothstep_derivatives_L_infinity}
    \begin{cases}
        \|S_n\|_{L^{\infty}(\mathbb R)}=1\\
        \|S_n^{(m)}\|_{L^{\infty}(\mathbb R)}=\frac{(2n+1)!}{4^{n-m+1}n!(n-m+1)!},&1\leq m\leq n
    \end{cases}.
\end{equation}
For $2\le m\le n$, we have the bound:
\begin{equation*}
    \begin{aligned}
        \|S_n^{(m)}\|_{L^{\infty}}&=\frac{(2n+1)!}{(n!)^2}(m-1)!\sup\limits_{x\in[0,1]}|x(1-x)|^{n-m+1}|P^{(n-m+1,n-m+1)}_{m-1}(1-2x)|\\&\leq\frac{(2n+1)!}{(n!)^2}(m-1)!4^{-(n-m+1)}\frac{n!}{(m-1)!(n-m+1)!}\\&=\frac{(2n+1)!}{4^{n-m+1}n!(n-m+1)!}
    \end{aligned}
\end{equation*}
For $m=1$, $\|S_n'\|_{L^{\infty}(\mathbb R)}=\frac{(2n+1)!}{(n!)^2}\|x^n(1-x)^n\|_{L^\infty(\mathbb R)}=\frac{(2n+1)!}{(n!)^2}4^{-n}$, which is again consistent with the general bound. For $m=0$, $\|S_n\|_{L^{\infty}(\mathbb R)}=1$.
\end{enumerate}
Scaling the smoothstep functions as $S_n(x/\kappa)$ with $\kappa>0$ preserves their smoothness while enabling approximation of a unit step function. Because $\frac{d^m}{dx^m}S_n(\frac{x}{\kappa})=\kappa^{-m}S_n^{(m)}(\frac{x}{\kappa})$, the relation between the $L^1$-norms of the scaled and unscaled derivatives is
\begin{equation}
    \label{eq:scaled_smoothstep_derivatives_L_1}
    \|\frac{d^m}{dx^m}S_n(\frac{.}{\kappa})\|_{L^1(\mathbb R)}=\kappa^{-m}\|S_n^{(m)}(\frac{.}{\kappa})\|_{L^1(\mathbb R)}=\kappa^{1-m}\|S_n^{(m)}\|_{L^1(\mathbb R)}.
\end{equation}  
Similarly, the relation for the $L^{\infty}$-norm is
\begin{equation}
    \label{eq:scaled_smoothstep_derivatives_L_infinity}
    \|\frac{d^m}{dx^m}S_n(\frac{.}{\kappa})\|_{L^{\infty}(\mathbb R)}=\kappa^{-m}\|S_n^{(m)}(\frac{.}{\kappa})\|_{L^{\infty}(\mathbb R)}=\kappa^{-m}\|S_n^{(m)}\|_{L^{\infty}(\mathbb R)}.
\end{equation}
}
\section{Proofs of Auxiliary Results}
\subsection{Proof of Proposition $\ref{proposition:exact_representation}$}\label{proof:proposition_exact_representation}
{Consider the Fourier transform pair $\tilde{x}^\star_{\mathcal B(\omega_0,\Omega)}(t):=(P_{\mathcal B(\omega_0,\Omega)}x^\star)(t)\leftrightarrow\hat {\tilde{x}}^\star_{\mathcal B(\omega_0,\Omega)}(\omega):=\hat x^\star(\omega)\boldsymbol{1}_{\omega\in\mathcal B(\omega_0,\Omega)}$, which can be derived by the definition of the band-limiting operator \eqref{eq:band_limiting_operator}.
Clearly, $\tilde{x}^\star_{\mathcal B(\omega_0,\Omega)}(t)\in\mathscr B(\mathcal B(\omega_0,\Omega))$. Since the modulated PSWFs $\{\psi_n(c,t)e^{i\omega_0t}\,|\,c=\frac{\Omega T}{2},n=0,1,2,\cdots\}$ form a complete orthonormal set in $\mathscr B(\mathcal B(\omega_0,\Omega))$, the function $\tilde{x}_{\mathcal B(\omega_0,\Omega)}^\star$ can be expanded as follows:
\begin{equation}
    \label{eq:x_tilde_expansion}
    \tilde{x}_{\mathcal B(\omega_0,\Omega)}^\star=\sum\limits^{\infty}_{n=0}\langle\tilde{x}_{\mathcal{B}(\omega_0,\Omega)}^\star,\psi_ne^{i\omega_0t}\rangle\psi_n(c,t)e^{i\omega_0t}
\end{equation}
The inner product appearing in \eqref{eq:x_tilde_expansion} can be simplified to $\Delta C^\star_n(x^\star,\omega_0,\Omega)$ as follows.
\begin{align*}
    \langle\tilde{x}^\star_{\mathcal B(\omega_0,\Omega)},\psi_ne^{i\omega_0t}\rangle&=\int^{\infty}_{-\infty}\tilde{x}^\star_{\mathcal B(\omega_0,\Omega)}(t)\psi_n(c,t)e^{-i\omega_0t}dt\\&=\frac{1}{2\pi}\int^{\infty}_{-\infty}\hat{\tilde{x}}^\star_{\mathcal B(\omega_0,\Omega)}(\omega)\overline{\hat{\psi_n}}(c,\omega-\omega_0)d\omega\\&=\frac{1}{2\pi}\int^{\infty}_{-\infty}\big(\hat{x}^\star(\omega)\boldsymbol{1}_{\omega\in\mathcal B(\omega_0,\Omega)}\big)\big(\overline\ell_n\psi_n(c,\frac{T}{2\Omega}(\omega-\omega_0))\boldsymbol{1}_{\omega\in\mathcal B(\omega_0,\Omega)}\big)d\omega\\&=\Delta C^\star_n(x^\star,\omega_0,\Omega)
\end{align*}
The second equality follows from the generalized Parseval theorem \eqref{eq:Parseval}, the third equality follows from \eqref{eq:ell_n}, and the last one follows from \eqref{eq:C_n_ideal}. The energy of $x^\star$ within $\mathcal B(\omega_0,\Omega)$ is then obtained as 
\begin{align*}
    E_{\mathcal B(\omega_0,\Omega)}(x^\star)=E_{\mathbb R}(\tilde{x}^\star_{\mathcal B(\omega_0,\Omega)})=\sum\limits^{\infty}_{n=0}|\langle\tilde{x}^\star_{\mathcal B(\omega_0,\Omega)},\psi_ne^{i\omega_0t}\rangle|^2=\Delta^2\sum\limits^{\infty}_{n=0}\Big|C^\star_n(x^\star,\omega_0,\Omega)\Big|^2.
\end{align*}
\hspace*{15.5cm}Q.E.D.
}
\subsection{Proof of Lemma $\ref{lemma:difference_C_m_C_inf}$}\label{proof:lemma_difference_C_m_C_inf}
{\begin{equation*}
    \begin{aligned}
        &C^{(m)}_n(y,\omega_0,\Omega)-C^{(\infty)}_n(x^\star,\omega_0,\Omega)=-\sum\limits_{|k|>\lfloor\frac{m}{2}\rfloor}x^\star_{d}[k]\overline{u_{d,n}}[k]\\\Rightarrow&\,\Big|C^{(m)}_n(y,\omega_0,\Omega)-C^{(\infty)}_n(x^\star,\omega_0,\Omega)\Big|\leq\left(\sum\limits_{|k|>\lfloor\frac{m}{2}\rfloor}|x^\star_{d}[k]|^2\right)^{1/2}\left(\sum\limits_{|k|>\lfloor\frac{m}{2}\rfloor}|\overline{u_{d,n}}[k]|^2\right)^{1/2}
    \end{aligned}
\end{equation*}  
To bound the two terms $\sum\limits_{|k|>\lfloor\frac{m}{2}\rfloor}|x^\star_{d}[k]|^2$ and $\sum\limits_{|k|>\lfloor\frac{m}{2}\rfloor}|\overline{u_{d,n}}[k]|^2$, we appeal to the Peano kernel theorem \eqref{eq:peano_theorem}.\\
\hspace*{0.5cm}Let $L_k\in C^{2}([(k-1/2)\Delta,(k+1/2)\Delta])\to\mathbb C$ where $|k|>\lfloor\frac{m}{2}\rfloor$ be the linear functional defined as
\begin{equation*}
    L_k(g)=\Delta g(k\Delta)-\int^{(k+1/2)\Delta}_{(k-1/2)\Delta}g(t)dt.
\end{equation*} 
For any $1^{st}$ polynomial $g(t)=at+b$ (with arbitrary constants $a$ and $b$), we have
\begin{equation*}
    \begin{aligned}
        L_k(g)&=\Delta(ak\Delta+b)-\int^{(k+1/2)\Delta}_{(k-1/2)\Delta}(at+b)dt=0.
    \end{aligned}
\end{equation*} 
Hence, by the Peano kernel theorem, 
\begin{equation*}
    L_k(g)=\int^{(k+1/2)\Delta}_{(k-1/2)\Delta}L_k[(t-\theta)_+]g''(\theta)d\theta,
\end{equation*}
and the corresponding Peano kernel of $L_k$ is 
\begin{equation*}
    \begin{aligned}
        L_k\big((t-\theta)_+\big)&=\Delta(k\Delta-\theta)_+-\int^{(k+1/2)\Delta}_{(k-1/2)\Delta}(t-\theta)_+dt\\&=\Delta(k\Delta-\theta)_+-\int^{(k+1/2)\Delta}_{\theta}(t-\theta)dt\\&=\Delta(k\Delta-\theta)_+-\frac{1}{2}((k+1/2)\Delta-\theta)^2\\&=-\frac{1}{2}(\Delta/2-|k\Delta-\theta|)^2.
    \end{aligned}
\end{equation*}
We may bound $|L_k(g)|$ by H\"older's inequality:
\begin{equation*}
    \begin{aligned}
        |L_k(g)|\leq&\,\sup\limits_{\theta\in[(k-1/2)\Delta,(k+1/2)\Delta]}|-\frac{1}{2}(\Delta/2-|k\Delta-\theta|)^2|\int^{(k+1/2)\Delta}_{(k-1/2)\Delta}|g''(\theta)|d\theta=\,\frac{\Delta^2}{8}\int^{(k+1/2)\Delta}_{(k-1/2)\Delta}|g''(\theta)|d\theta.
    \end{aligned}
\end{equation*}
Consequently,
\begin{equation*}
    \begin{aligned}
        \Big|\Delta\sum\limits_{|k|>\lfloor\frac{m}{2}\rfloor}g(k\Delta)-\int_{|t|>T/2+\Delta/2}g(t)dt\Big|\leq&\,\sum\limits_{|k|>\lfloor\frac{m}{2}\rfloor}\Big|\Delta g(k\Delta)-\int^{(k+1/2)\Delta}_{(k-1/2)\Delta}g(t)dt\Big|\\=&\,\sum\limits_{|k|>\lfloor\frac{m}{2}\rfloor}|L_k(g)|\\\leq&\,\sum\limits_{|k|>\lfloor\frac{m}{2}\rfloor}\frac{\Delta^2}{8}\int^{(k+1/2)\Delta}_{(k-1/2)\Delta}|g''(t)|dt\\=&\,\frac{\Delta^2}{8}\int_{|t|>T/2+\Delta/2}|g''(t)|dt,
    \end{aligned}
\end{equation*}
and therefore   
\begin{equation*}
    \begin{aligned}
        \sum\limits_{|k|>\lfloor\frac{m}{2}\rfloor}g(k\Delta)&\leq\frac{1}{\Delta}\left(\int_{|t|>T/2+\Delta/2}g(t)dt+\frac{\Delta^2}{8}\int_{|t|>T/2+\Delta/2}|g''(t)|dt\right)\\&\leq\frac{1}{\Delta}\left(\int_{|t|>T/2}g(t)dt+\frac{\Delta^2}{8}\int_{|t|>T/2}|g''(t)|dt\right).
    \end{aligned}
\end{equation*}
For $g(t)=|x^\star(t)|^2$, we obtain
\begin{equation}
    \label{eq:x_star_d_energy bound}
    \begin{aligned}
         \sum\limits_{|k|>\lfloor\frac{m}{2}\rfloor}|x^\star_{d}[k]|^2&=\sum\limits_{|k|>\lfloor\frac{m}{2}\rfloor}|x^\star(k\Delta)|^2\\&\leq\frac{1}{\Delta}\left(\epsilon_T^2\|x^\star\|^2_{L^2(\mathbb R)}+\frac{\Delta^2}{8}\|\frac{d^2}{dt^2}|{x^\star}|^2\|_{L^1(\mathbb R\setminus\mathcal I_T)}\right).
    \end{aligned}
\end{equation}
Similarly, for $g(t)=|u_n(t)|^2=|\psi_n(c,t)|^2$, we have
\begin{equation*}
    \begin{aligned}
        \sum\limits_{|k|>\lfloor\frac{m}{2}\rfloor}|\overline{u_{d,n}}[k]|^2&=\sum\limits_{|k|>\lfloor\frac{m}{2}\rfloor}|u_{d}(k\Delta)|^2\\&\leq\frac{1}{\Delta}\left(\int_{|t|>T/2}|\psi_n(c,t)|^2dt+\frac{\Delta^2}{8}\int_{|t|>T/2}|\frac{d^2}{dt^2}\psi_n^2(c,t)|dt
        \right)\\&=\frac{1}{\Delta}\left(1-\lambda_n(\frac{\Omega T}{2})+\frac{\Delta^2}{8}\int_{|t|>T/2}|\frac{d^2}{dt^2}\psi_n^2(c,t)|dt\right).
    \end{aligned}
\end{equation*}
We bound the tail integral $\int_{|t|>T/2}|\frac{d^2}{dt^2}\psi_n^2(c,t)|dt$ as follows. Since $\frac{d^2}{dt^2}\psi_n^2(c,t)=2\bigl(\psi_n''(c,t)\psi_n(c,t)+|\psi_n'(c,t)|^2\bigr)$,
we have
\begin{align*}
\int_{|t|>T/2}\bigl|\frac{d^2}{dt^2}\psi_n^2(c,t)\bigr|dt
&\le 2\int_{|t|>T/2}\Bigl(|\psi_n''(c,t)\psi_n(c,t)|+|\psi_n'(c,t)|^2\Bigr)dt \\
&\le 2\|\psi_n''\|_{L^2(\mathbb R)}\|\psi_n\|_{L^2(\mathbb R)}
   +2\|\psi_n'\|_{L^2(\mathbb R)}^2,
\end{align*}
by Cauchy-Schwarz. Since $\psi_n$ is band-limited to $[-\Omega,\Omega]$, 
Parseval implies $\|\psi_n^{(k)}\|_{L^2(\mathbb R)}\leq\Omega^k\|\psi_n\|_{L^2(\mathbb R)}$ for $k=1,2$. Combining these estimates, we obtain $\int_{|t|>T/2}\bigl|\frac{d^2}{dt^2}\psi_n^2(c,t)\bigr|dt\le 4\Omega^2\|\psi_n\|_{L^2(\mathbb R)}^2
=4\Omega^2$. Hence, 
\begin{equation}
    \label{eq:udn_energy_bound}
    \sum\limits_{|k|>\lfloor\frac{m}{2}\rfloor}|\overline{u_{d,n}}[k]|^2\leq\frac{1}{\Delta}\left(1-\lambda_n(\frac{\Omega T}{2})+\frac{\Omega^2\Delta^2}{2}\right).
\end{equation}
Combining \eqref{eq:x_star_d_energy bound} and \eqref{eq:udn_energy_bound}, we conclude that
\begin{equation}
    \begin{aligned}
        &\Big|C^{(m)}_n(y,\omega_0,\Omega)-C^{(\infty)}_n(x^\star,\omega_0,\Omega)\Big|\\\leq&\,\frac{1}{\Delta}\left(\epsilon_T^2\|x^\star\|^2_{L^2(\mathbb R)}+\frac{\Delta^2}{8}\|\frac{d^2}{dt^2}|{x^\star}|^2\|_{L^1(\mathbb R\setminus\mathcal I_T)}\right)^{1/2}\left(1-\lambda_n(\frac{\Omega T}{2})+\frac{\Omega^2\Delta^2}{2}\right)^{1/2}.
    \end{aligned}
\end{equation}
\hspace*{15.5cm}Q.E.D.
}
\subsection{Proof of Lemma $\ref{lemma:difference_C_inf_C_ideal}$}\label{proof:lemma_difference_C_inf_C_ideal}
{\begin{equation*}
    \begin{aligned}
        \Big|C^{(\infty)}_n(x^\star,\omega_0,\Omega)-C^\star_n(x^\star,\omega_0,\Omega)\Big|&=\Big|\frac{\overline{\ell_n}}{2\pi\Delta}\sum\limits_{k\not=0}\int^{\omega_0+\Omega}_{\omega_0-\Omega}\hat x^\star(\omega-k\frac{2\pi}{\Delta})\psi_n(c,\frac{T}{2\Omega}(\omega-\omega_0))d\omega\Big|\\&\leq\Big|\frac{\ell_n}{2\pi\Delta}\Big|\sqrt{\frac{2\Omega}{T}\lambda_n(c)}\left(\int^{\omega_0+\Omega}_{\omega_0-\Omega}\left(\sum\limits_{k\not=0}\hat x^\star(\omega-k\frac{2\pi}{\Delta})\right)^2d\omega\right)^{1/2} \\&\leq\alpha_{fd}\Big|\frac{\ell_n}{2\pi\Delta}\Big|\sqrt{\frac{2\Omega}{T}\lambda_n(c)}\left(\int^{\omega_0+\Omega}_{\omega_0-\Omega}\left(\sum\limits_{k\not=0}\Big|\omega-k\frac{2\pi}{\Delta}\Big|^{-r}\right)^2d\omega\right)^{1/2}\\&=\frac{\alpha_{fd}}{\Delta\sqrt{2\pi}}\left(\int^{\omega_0+\Omega}_{\omega_0-\Omega}\left(\sum\limits_{k\not=0}\Big|\omega-k\frac{2\pi}{\Delta}\Big|^{-r}\right)^2d\omega\right)^{1/2}
    \end{aligned}
\end{equation*}  
The first inequality follows from the Cauchy-Schwarz inequality together with the orthogonal property of PSWFs \eqref{eq:PSWF_orthogonal}. The second inequality is a consequence of the spectral localization assumption (Assumption \ref{assumption:SL}). The final equality follows from the identities $\lambda_n(c)=\frac{2c}{\pi}[R^{(1)}_{0n}(c,1)]^2$ \eqref{eq:lambda_n}, $c=\frac{\Omega T}{2}$ \eqref{eq:c} and $\ell_n=\frac{\pi}{i^n\Omega R^{(1)}_{0n}(c,1)}$ \eqref{eq:ell_n}. The term $\sum\limits_{k\not=0}\Big|\omega-k\frac{2\pi}{\Delta}\Big|^{-r}$ for $\omega\in\mathcal B(\omega_0,\Omega)$ can be bounded above as follows.
\begin{equation*}
    \begin{aligned}
        \sum\limits_{k\not=0}\Big|\omega-k\frac{2\pi}{\Delta}\Big|^{-r}&=\sum\limits_{k\geq1}(2kW_{\mathrm{hs}}-\omega)^{-r}+\sum\limits_{k\leq-1}(\omega+2|k|W_{\mathrm{hs}})^{-r}\\&\leq2\sum\limits_{k\geq1}[(2k-1)W_{\mathrm{hs}}]^{-r}\\&\leq 2W_{\mathrm{hs}}^{-r}\left(1+2^{-r}\sum\limits_{k\geq1}k^{-r}\right)\\&=2W_{\mathrm{hs}}^{-r}\left(1+2^{-r}\left(1+\int^{\infty}_1x^{-r}dx\right)\right)\\&=2W_{\mathrm{hs}}^{-r}\left(1+2^{-r}\left(1+\frac{1}{r-1}\right)\right)\\&\leq3W_{\mathrm{hs}}^{-r}
    \end{aligned}
\end{equation*}
The first inequality holds because $W_{\mathrm{hs}}\geq W_{\mathrm{cut}}>\omega_0+\Omega\geq\omega>\omega_0-\Omega>-W_{\mathrm{cut}}\geq -W_{\mathrm{hs}}$. The final inequality is justified by the assumption $r\geq2$.
Thus, 
\begin{equation}
    \Big|C^{(\infty)}_n(x^\star,\omega_0,\Omega)-C^\star_n(x^\star,\omega_0,\Omega)|\leq\frac{1}{\Delta}\frac{6\Omega\alpha_{fd}}{\sqrt{2\pi}}W_{\mathrm{hs}}^{-r}\leq\frac{1}{\Delta}\sqrt{6}\Omega\alpha_{fd}W_{\mathrm{hs}}^{-r}.
\end{equation} 
\hspace*{15.5cm}Q.E.D.
}
\subsection{Proof of Corollary $\ref{corollary:bounded_tail_energy}$}\label{proof:corollary_bounded_tail_energy}
{The upper bound for the finite-term truncation error involves the tail energy $\|P_{\mathcal B(\omega_0,\Omega)}x^\star\|_{L^2(\mathbb R\setminus\mathcal I_T)}=\int_{|t|>T/2}|P_{\mathcal B(\omega_0,\Omega)}x^\star(t)|^2dt$. Using the frequency-differentiation Fourier transform pair $tf(t)\leftrightarrow i\frac{d}{d\omega}\hat f(\omega)$ repeatedly, we obtain $t^sf(t)=\frac{1}{2\pi}\int_{\mathbb R}i^s{\hat f}^{(s)}(\omega)e^{i\omega t}d\omega$. Then, $|f(t)|\leq\frac{\|{\hat f}^{(s)}\|_{L^1(\mathbb R)}}{2\pi}|t|^{-s}$ whenever $\hat f^{(s)}\in L^1(\mathbb R)$. Thus, we seek to invoke the additional spectral regularity assumption (Assumption \ref{assumption:SR}). However, because $P_{\mathcal B(\omega_0,\Omega)}x^\star$ is obtained by hard thresholding $\hat x^\star(\omega)$, its Fourier transform is not $s$-times differentiable. Instead, we consider $x^\star_{\phi}(t)\leftrightarrow\hat{x}^\star_{\phi}(\omega):=\hat x^\star(\omega)\phi(\omega)$, where $\phi(\omega)=\begin{cases}
     1,&\omega\in\mathcal B(\omega_0,\Omega)\\0,&|\omega-\omega_0|\geq\Omega+\kappa\\\in[0,1],&\omega\in\mathcal I^+_{\omega_0,\Omega,\kappa}
 \end{cases}$, $\kappa>0$, $\mathcal I^+_{\omega_0,\Omega,\kappa}$ is defined in \eqref{eq:I_ft}, and $\phi(\omega)$ is at least $s$-times differentiable with each derivative in $L^1(\mathbb R)$. Such a window can be constructed using the smoothstep functions introduced in $\ref{sec:general_smoothstep}$. With this windowing, $\hat{x}^\star_{\phi}(\omega)$ is $s$-times differentiable with derivatives in $L^1(\mathbb R)$, so that $|x^\star_\phi(t)|\leq\frac{\|\hat x^{\star(s)}_{\phi}\|_{L^1(\mathbb R)}}{2\pi}|t|^{-s}$. Hence, 
 \begin{equation}
     \label{eq:x_star_phi_tail}
     \|x^\star_\phi\|_{L^1(\mathbb R\setminus\mathcal I_T)}\leq\frac{\|\hat x^{\star(s)}_{\phi}\|_{L^1(\mathbb R)}}{2\pi}\sqrt{\int_{|t|>T/2}|t|^{-2s}dt}=\frac{2^{s-1}\|\hat x^{\star(s)}_\phi\|_{L^1(\mathbb R)}}{\pi\sqrt{2s-1}}T^{1/2-s},
 \end{equation}
 which shows the need to bound $\|\hat x^{\star(s)}_\phi\|_{L^1(\mathbb R)}$. By the triangle inequality,
\begin{align*}
\|\hat x^{\star(s)}_\phi\|_{L^1(\mathbb R)}&=\|\sum\limits^s_{i=0}\binom{s}{i}\hat x^{\star(i)}\phi^{(s-i)}\|_{L^1(\mathbb R)}\leq\sum\limits^s_{i=0}\binom{s}{i}\|\hat x^{\star(i)}\phi^{(s-i)}\|_{L^1(\mathbb R)}.
\end{align*}
Using Hölder’s inequality together with \eqref{eq:smoothstep_derivatives_L_1} and \eqref{eq:scaled_smoothstep_derivatives_L_1}, we obtain
 \begin{align}
    \label{eq:x_star_phi_derivative_bound1}
     \|\hat x^{\star(s)}_{\phi}\|_{L^1(\mathbb R)}&\leq\,2\sum\limits^{s-1}_{i=0}\left(\frac{s}{s-i}\binom{s-1}{i}\|\hat x^{\star(i)}\|_{L^{\infty}(\mathcal I^+_{\omega_0,\Omega,\kappa})}\kappa^{1-(s-i)}\frac{(2s+1)!}{s!}\frac{(i+1)!}{(2i+3)!}\right)+\|\hat x^{\star(s)}\phi\|_{L^1(\mathbb R)}\nonumber\\&\leq\,\frac{2(2s+1)!}{(s-1)!}(1+\kappa^{-1})^{s-1}\sup\limits_{i=0,1,\cdots,s-1}\|\hat x^{\star(i)}\|_{L^{\infty}(\mathcal I^+_{\omega_0,\Omega,\kappa})}\sup\limits_{i=0,1,\cdots,s-1}\frac{(i+1)!}{(s-i)(2i+3)!}\nonumber\\&\hspace{0.5cm}+\|\hat x^{\star(s)}\|_{L^1(\mathcal B(\omega_0,\Omega+\kappa))}\nonumber\\&\leq\gamma^+_{\hat x^\star,s,\omega_0,\Omega}(\kappa)\frac{(2s+1)!}{3s!}(1+\kappa^{-1})^{s-1},
 \end{align}
 where $\gamma^+_{\hat x^\star,s,\omega_0,\Omega}(\kappa)$ is defined in \eqref{eq:gamma_ft}. Similarly, by Hölder’s inequality together with \eqref{eq:smoothstep_derivatives_L_infinity} and \eqref{eq:scaled_smoothstep_derivatives_L_infinity},
 \begin{align}
    \label{eq:x_star_phi_derivative_bound2}
     \|\hat x^{\star(s)}_\phi\|_{L^1(\mathbb R)}&\leq\,\sum\limits^{s-1}_{i=0}\left(\frac{s}{s-i}\binom{s-1}{i}\|\hat x^{\star(i)}\|_{L^1(\mathcal I^+_{\omega_0,\Omega,\kappa})}\kappa^{-(s-i)}\frac{(2s+1)!}{s!(i+1)!}4^{-i-1}\right)+\|\hat x^{\star(s)}\phi\|_{L^1(\mathbb R)}\nonumber\\&\leq\,\frac{(2s+1)!}{2(s-1)!}(1+\kappa^{-1})^{s-1}\sup\limits_{i=0,1,...,s-1}\|\hat x^{\star(i)}\|_{L^{\infty}(\mathcal I^+_{\omega_0,\Omega,\kappa})}\sup\limits_{i=0,1,...,s-1}\frac{4^{-i}}{(s-i)(i+1)!}\nonumber\\&\hspace{0.5cm}+\|\hat x^{\star(s)}\|_{L^1(\mathcal B(\omega_0,\Omega+\kappa))}\nonumber\\&\leq\gamma^+_{\hat x^\star,s,\omega_0,\Omega}(\kappa)\frac{(2s+1)!}{2s!}(1+\kappa^{-1})^{s-1}.
 \end{align}
 The second inequality holds because $\|\hat x^{\star(i)}\|_{L^1(\mathcal I^+_{\omega_0,\Omega,\kappa})}\leq2\kappa\|\hat x^{\star(i)}\|_{L^{\infty}(\mathcal I^+_{\omega_0,\Omega,\kappa})}$. Combining \eqref{eq:x_star_phi_derivative_bound1} and \eqref{eq:x_star_phi_derivative_bound2} gives
 \begin{equation}
     \label{eq:x_star_phi_derivative_final_bound}
    \|\hat x^{\star(s)}_\phi\|_{L^1(\mathbb R)}\leq\gamma^+_{\hat x^\star,s,\omega_0,\Omega}(\kappa)\frac{(2s+1)!}{3s!}(1+\kappa^{-1})^{s-1}.
 \end{equation}
 \hspace*{0.5cm}We now derive an upper bound for the tail energy $\|P_{\mathcal B(\omega_0,\Omega)x^\star}\|_{L^2(\mathbb R\setminus\mathcal I_T)}$:
 \begin{align*}
    \|P_{\mathcal B(\omega_0,\Omega)}x^\star\|_{L^2(\mathbb R\setminus\mathcal I_T)}-\|x^\star_\phi\|_{L^2(\mathbb R\setminus\mathcal I_T)}&=\sqrt{\int_{|t|>T/2}|(P_{\mathcal B(\omega_0,\Omega)}x^\star)(t)|^2dt}-\sqrt{\int_{|t|>T/2}|x^\star_\phi(t)|
     ^2dt}\\&\leq\,\sqrt{\int_{|t|>T/2}|(P_{\mathcal B(\omega_0,\Omega)}x^\star)(t)-x^\star_\phi(t)|^2dt}\\&\leq\,\sqrt{\int_{\mathbb R}|(P_{\mathcal B(\omega_0,\Omega)}x^\star)(t)-x^\star_\phi(t)|^2dt}\\&=\,\sqrt{\frac{1}{2\pi}\int_{\mathbb R}|\hat x^\star(\omega)\boldsymbol{1}_{\omega\in\mathcal B(\omega_0,\Omega)}-\hat x^\star(\omega)\phi(\omega)|^2d\omega}\\&=\,\sqrt{\frac{1}{2\pi}\int_{\omega\in\mathcal I^+_{\omega_0,\Omega,\kappa}}|\hat x^\star(\omega)\phi(\omega)|^2d\omega}\\&\leq\,\sqrt{\frac{\kappa}{\pi}}\|\hat x^\star\|_{L^{\infty}(\mathcal I^+_{\omega_0,\Omega,\kappa})}
 \end{align*}
 The first inequality follows from the reverse triangle inequality \eqref{eq:reverse_triangle_inequality}.
Hence,
 \begin{align}
    \label{eq:P_B_x_star_tail_energy_bound} 
     \|P_{\mathcal B(\omega_0,\Omega)}x^\star\|_{L^2(\mathbb R\setminus\mathcal I_T)}&\leq\,\gamma^+_{\hat x^\star,s,\omega_0,\Omega}(\kappa)\frac{(2s+1)!}{3\pi s!\sqrt{2s-1}}(2+2\kappa^{-1})^{s-1}T^{1/2-s}+\sqrt{\frac{\kappa}{\pi}}\|\hat x^\star\|_{L^{\infty}(\mathcal I^+_{\omega_0,\Omega,\kappa})}\nonumber\\&\leq\,\gamma^+_{\hat x^\star,s,\omega_0,\Omega}(\kappa)\left[\frac{(2s+1)!}{3\pi s!\sqrt{2s-1}}(2+2\kappa^{-1})^{s-1}T^{1/2-s}+\sqrt{\frac{\kappa}{\pi}}\right]\nonumber\\&\leq\,\gamma_{\hat x^\star,s,\omega_0,\Omega}^+(\kappa)Q_{s,T}(\kappa), \forall\kappa>0
 \end{align}
 where $Q_{s,T}(\kappa)$ is defined in \eqref{eq:Qst_ft}.
Since \eqref{eq:P_B_x_star_tail_energy_bound} holds for all $\kappa>0$,
 \begin{equation}
     \|P_{\mathcal B(\omega_0,\Omega)}x^\star\|_{L^2(\mathbb R\setminus\mathcal I_T)}\leq\inf\limits_{\kappa>0}\gamma^+_{\hat x^\star,s,\omega_0,\Omega}(\kappa)Q_{s,T}(\kappa).
 \end{equation}
 \hspace*{15.5cm}Q.E.D.
 }
 \subsection{Proof of Corollary $\ref{corollary:bounded_tail_energy_special_case}$}\label{proof:corollary_bounded_tail_energy_special_case}
 {Let $Z:=\frac{(2s+1)!2^{s-1}}{3\pi s!\sqrt{2s-1}}T^{1/2-s}$. Then, $Q_{s,T}(\kappa)=Z(1+\kappa^{-1})^{s-1}+\sqrt{\frac{\kappa}{\pi}}$. Although a closed-form expression for the minimum of $Q_{s,T}(\kappa)$ is not available, we know that it decreases as $Z$ decreases, and the corresponding minimizer also becomes smaller. It is therefore reasonable to choose $\kappa$ as a function of $Z$ such that smaller values of $Z$ correspond to smaller values of $\kappa$. In this way, we can obtain a useful upper bound for $Q_{s,T}(\kappa)$. Specifically, we set $\kappa=2Z^{1/s}$ and assume $\kappa<1$. Then,
 \begin{align}
     Q_{s,T}(2Z^{1/s})=Z(1+(2Z^{1/s})^{-1})^{s-1}+\sqrt{\frac{2Z^{1/s}}{\pi}}&\leq Z(2(2Z^{1/s})^{-1})^{s-1}+\sqrt{\frac{2}{\pi}}Z^{1/2s}\nonumber\\&\leq Z^{1/s}+Z^{1/2s}\nonumber\\&\leq2Z^{1/2s}\nonumber\\&=2\left(\frac{2^{s-1}(2s+1)!}{3\pi s!\sqrt{2s-1}}\right)^{\frac{1}{2s}}T^{\frac{1}{4s}-\frac{1}{2}}.
 \end{align}
 The margin is $\kappa_T=2Z^{1/s}=\left(\frac{2^{2s-1}(2s+1)!}{3\pi s!\sqrt{2s-1}}\right)^{1/s}T^{\frac{1}{2s}-1}$.
 The requirement $\kappa=\kappa_T<1$ is satisfied whenever 
 \begin{align}
     T>\left(\frac{(2s+1)!2^{2s-1}}{3\pi s!\sqrt{2s-1}}\right)^{\frac{2}{2s-1}}.
 \end{align}
\hspace*{15.5cm}Q.E.D.
}
\subsection{Proof of Corollary $\ref{corollary:noise_perturbation_analysis}$}\label{proof:corollary_noise_perturbation_analysis}
{$A^{\mathrm{est}}_{\mathcal B(\omega_0,\Omega)}(y_{\eta})-A_{\mathcal B(\omega_0,\Omega)}(x^\star)=\Delta\sqrt{\sum\limits^{\lfloor\frac{2c}{\pi}\rfloor}_{n=0}\Big|C^{(m)}_n(y,\omega_0,\Omega)+C^{(m)}_n(\eta,\omega_0,\Omega)\Big|^2}-A_{\mathcal B(\omega_0,\Omega)}(x^\star)$\\
If $A^{\mathrm{est}}_{\mathcal B(\omega_0,\Omega)}(y_{\eta})-A_{\mathcal B(\omega_0,\Omega)}(x^\star)<0$, then
\begin{align*}
    0&>\Delta\sqrt{\sum\limits^{\lfloor\frac{2c}{\pi}\rfloor}_{n=0}\Big|C^{(m)}_n(y,\omega_0,\Omega)+C^{(m)}_n(\eta,\omega_0,\Omega)\Big|^2}-A_{\mathcal B(\omega_0,\Omega)}(x^\star)\\&\geq\Delta\sqrt{\sum\limits^{\lfloor\frac{2c}{\pi}\rfloor}_{n=0}\Big|C^{(m)}_n(y,\omega_0,\Omega)\Big|^2}-A_{\mathcal B(\omega_0,\Omega)}(x^\star)-\Delta\sqrt{\sum\limits^{\lfloor\frac{2c}{\pi}\rfloor}_{n=0}\Big|C^{(m)}_n(\eta,\omega_0,\Omega)\Big|^2}.
\end{align*}
The second inequality follows from the reverse triangle inequality \eqref{eq:reverse_triangle_inequality}.
If $A^{\mathrm{est}}_{\mathcal B(\omega_0,\Omega)}(y_{\eta})-A_{\mathcal B(\omega_0,\Omega)}(x^\star)\geq0$, then
\begin{align*}
    0&\leq\Delta\sqrt{\sum\limits^{\lfloor\frac{2c}{\pi}\rfloor}_{n=0}\Big|C^{(m)}_n(y,\omega_0,\Omega)+C^{(m)}_n(\eta,\omega_0,\Omega)\Big|^2}-A_{\mathcal B(\omega_0,\Omega)}(x^\star)\\&\leq\Delta\sqrt{\sum\limits^{\lfloor\frac{2c}{\pi}\rfloor}_{n=0}\Big|C^{(m)}_n(y,\omega_0,\Omega)\Big|^2}-A_{\mathcal B(\omega_0,\Omega)}(x^\star)+\Delta\sqrt{\sum\limits^{\lfloor\frac{2c}{\pi}\rfloor}_{n=0}\Big|C^{(m)}_n(\eta,\omega_0,\Omega)\Big|^2}.
\end{align*}
Hence, we obtain 
\begin{align}
    \label{eq:band_L^2_estimator_intermediate_bound_noise}
    \Big|A^{\mathrm{est}}_{\mathcal B(\omega_0,\Omega)}(y_{\eta})-A_{\mathcal B(\omega_0,\Omega)}(x^\star)\Big|&\leq\Big|A^{\mathrm{est}}_{\mathcal B(\omega_0,\Omega)}(y)-A_{\mathcal B(\omega_0,\Omega)}(x^\star)\Big|+\Delta\sqrt{\sum\limits^{\lfloor\frac{2c}{\pi}\rfloor}_{n=0}\Big|C^{(m)}_n(\eta,\omega_0,\Omega)\Big|^2}.
\end{align}
By the triangle inequality and the Cauchy-Schwarz inequality,
\begin{align}
    \label{eq:finite_sample_band_coefficient_noise_bound}
    \Delta\sqrt{\sum\limits^{\lfloor\frac{2c}{\pi}\rfloor}_{n=0}\Big|C^{(m)}_n(\eta,\omega_0,\Omega)\Big|^2}&=\Delta\sqrt{\sum\limits^{\lfloor\frac{2c}{\pi}\rfloor}_{n=0}\Big|\sum\limits^m_{j=1}\eta_j\psi_n(c,t_j)e^{-i\omega_0t_j}\Big|^2}\nonumber\\&\leq\Delta\sqrt{\sum\limits^{\lfloor\frac{2c}{\pi}\rfloor}_{n=0}\Big(\sum\limits^m_{j=1}\big|\eta_j\psi_n(c,t_j)e^{-i\omega_0t_j}\big|\Big)^2}\nonumber\\&\leq\Delta\sqrt{\sum\limits^{\lfloor\frac{2c}{\pi}\rfloor}_{n=0}\Big(\sum\limits^m_{j=1}\big|\eta_j\big|^2\Big)\Big(\sum\limits^m_{j=1}\big|\psi_n(c,t_j)e^{-i\omega_0t_j}\big|^2\Big)}=\Delta\varepsilon\sqrt{\sum\limits^{\lfloor\frac{2c}{\pi}\rfloor}_{n=0}\sum\limits^m_{j=1}\big|\psi_n(c,t_j)\big|^2}.
\end{align}
Using a similar argument as in $\ref{proof:lemma_difference_C_m_C_inf}$, we bound the energy of discrete sampled PSWFs within $\mathcal I_T$ as follows.
\begin{align}
    \sum\limits^m_{j=1}|\psi_n(c,t_j)|^2&\leq\frac{1}{\Delta}\Big(\int_{|t|\leq\frac{T}{2}+\frac{\Delta}{2}}|\psi_n(c,t)|^2dt+\frac{\Delta^2}{8}\int_{|t|\leq\frac{T}{2}+\frac{\Delta}{2}}|\frac{d^2}{dt^2}\psi_n^2(c,t)|dt\Big)\nonumber\\&\leq\frac{1}{\Delta}\Big(\lambda_n(c)+\int_{t\in[-\frac{T}{2}-\frac{\Delta}{2},-\frac{T}{2}]\cup[\frac{T}{2},\frac{T}{2}+\frac{\Delta}{2}]}|\psi_n(c,t)|^2dt+\frac{\Delta^2}{8}\int_{\mathbb R}|\frac{d^2}{dt^2}\psi_n^2(c,t)|dt\Big)\nonumber\\&\leq\frac{1}{\Delta}\Big(\lambda_n(c)+\int_{t\in[-\frac{T}{2}-\frac{\Delta}{2},-\frac{T}{2}]\cup[\frac{T}{2},\frac{T}{2}+\frac{\Delta}{2}]}|\psi_n(c,t)|^2dt+\frac{\Omega^2\Delta^2}{2}\Big)
\end{align}
Since $|\psi_n(c,t)|=\frac{1}{2\pi}|\int^{\Omega}_{-\Omega}\hat\psi_n(c,\omega)d\omega|\leq\frac{1}{2\pi}\|\hat\psi_n\|_{L^2(\mathbb R)}\|\boldsymbol{1}(\omega\in\mathcal I_{\Omega})\|_{L^2(\mathbb R)}=\sqrt{\frac{\Omega}{\pi}}$, we obtain 
\begin{align}
    \label{eq:PSWF_finite_interval_energy_bound1}
    \sum\limits^m_{j=1}|\psi_n(c,t_j)|^2\leq\frac{1}{\Delta}\Big(\lambda_n(c)+\frac{\Delta\Omega}{\pi}+\frac{\Omega^2\Delta^2}{2}\Big).
\end{align}
We also have 
\begin{align}
    \label{eq:PSWF_finite_interval_energy_bound2}
    \sum\limits^m_{j=1}|\psi_n(c,t_j)|^2&\leq\frac{1}{\Delta}\Big(\int_{|t|\leq\frac{T}{2}+\frac{\Delta}{2}}|\psi_n(c,t)|^2dt+\frac{\Delta^2}{8}\int_{|t|\leq\frac{T}{2}+\frac{\Delta}{2}}|\frac{d^2}{dt^2}\psi^2_n(c,t)|dt\Big)\nonumber\\&\leq\frac{1}{\Delta}\Big(1+\frac{\Omega^2\Delta^2}{2}\Big).
\end{align}
Combining \eqref{eq:PSWF_finite_interval_energy_bound1} and \eqref{eq:PSWF_finite_interval_energy_bound2} yields
\begin{equation}
    \label{eq:PSWF_finite_interval_energy_final_bound}
    \sum\limits^m_{j=1}|\psi_n(c,t_j)|^2\leq\frac{1}{\Delta}\Big(\frac{\Omega^2\Delta^2}{2}+\min\{1,\lambda_n(c)+\frac{\Omega\Delta}{\pi}\}\Big).
\end{equation}
Finally, combining \eqref{eq:band_L^2_estimator_intermediate_bound_noise}, \eqref{eq:finite_sample_band_coefficient_noise_bound}, and \eqref{eq:PSWF_finite_interval_energy_final_bound}, we obtain
\begin{align}
    \label{eq:band_L^2_estimator_final_bound_noise}
    &\Big|A^{\mathrm{est}}_{\mathcal B(\omega_0,\Omega)}(y_{\eta})-A_{\mathcal B(\omega_0,\Omega)}(x^\star)\Big|\nonumber\\\leq&\Big|A^{\mathrm{est}}_{\mathcal B(\omega_0,\Omega)}(y)-A_{\mathcal B(\omega_0,\Omega)}(x^\star)\Big|+\varepsilon\sqrt{\Delta\sum\limits^{\lfloor\frac{2c}{\pi}\rfloor}_{n=0}\Big(\frac{\Omega^2\Delta^2}{2}+\min\{1,\lambda_n(c)+\frac{\Omega\Delta}{\pi}\}\Big)}.
\end{align}
\hspace*{15.5cm}Q.E.D.
}


\section{Supplementary Figures}
{\begin{figure}[H]
    \centering
    \begin{subfigure}{0.32\textwidth}
        \centering
        \includegraphics[width=\linewidth]{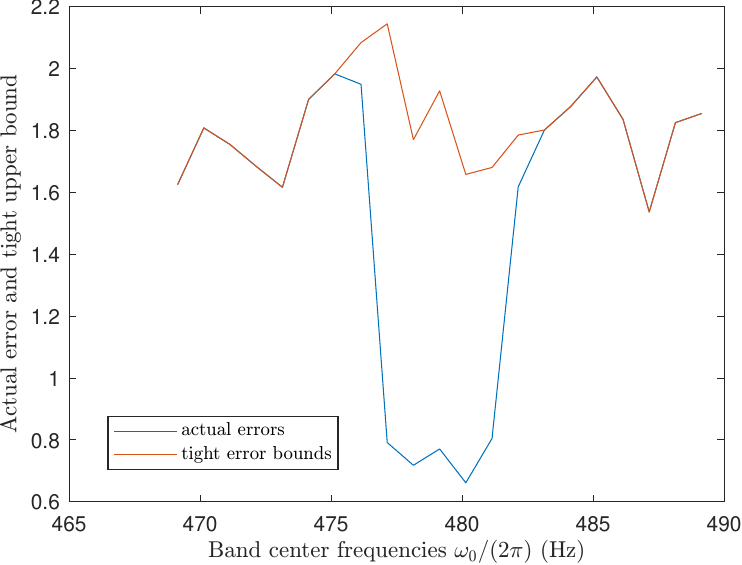}
        \caption{}
        \label{fig:noise_B_splines_1}
    \end{subfigure}
    \begin{subfigure}{0.32\textwidth}
        \centering
        \includegraphics[width=\linewidth]{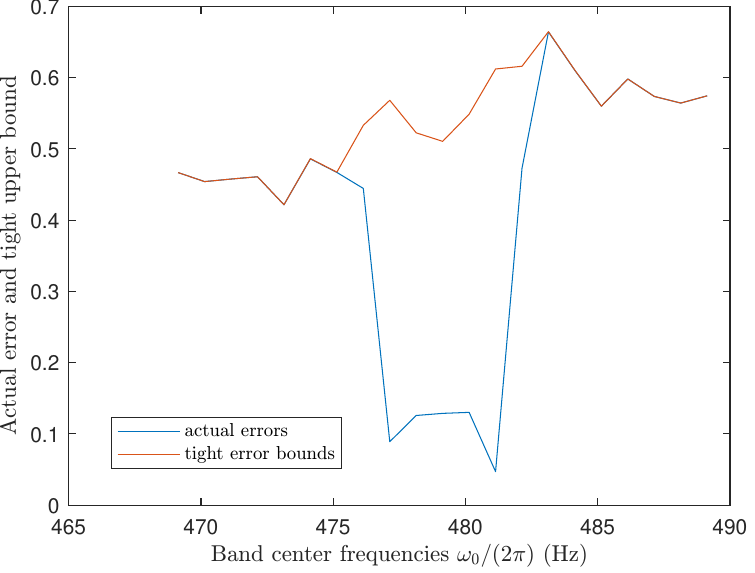}
        \caption{}
        \label{fig:noise_B_splines_2}
    \end{subfigure}
    \begin{subfigure}{0.32\textwidth}
        \centering
        \includegraphics[width=\linewidth]{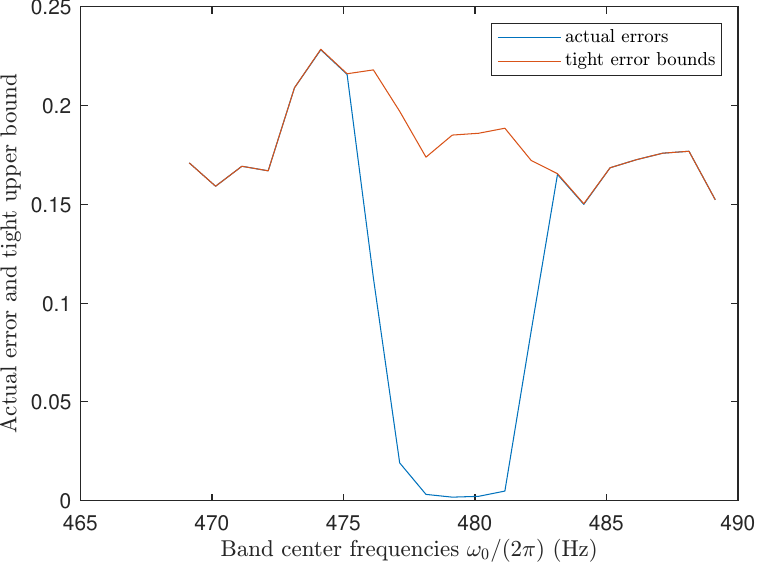}
        \caption{}
        \label{fig:noise_B_splines_3}
    \end{subfigure}
    \begin{subfigure}{0.32\textwidth}
        \centering
        \includegraphics[width=\linewidth, height=3.8cm]{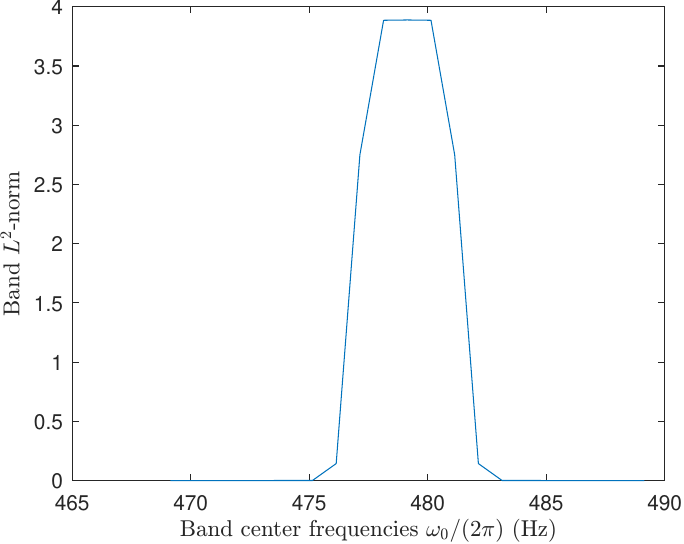}
        \caption{}
        \label{fig:noise_B_splines_4}
    \end{subfigure}
    \begin{subfigure}{0.32\textwidth}
        \centering
        \includegraphics[width=\linewidth, height=3.8cm]{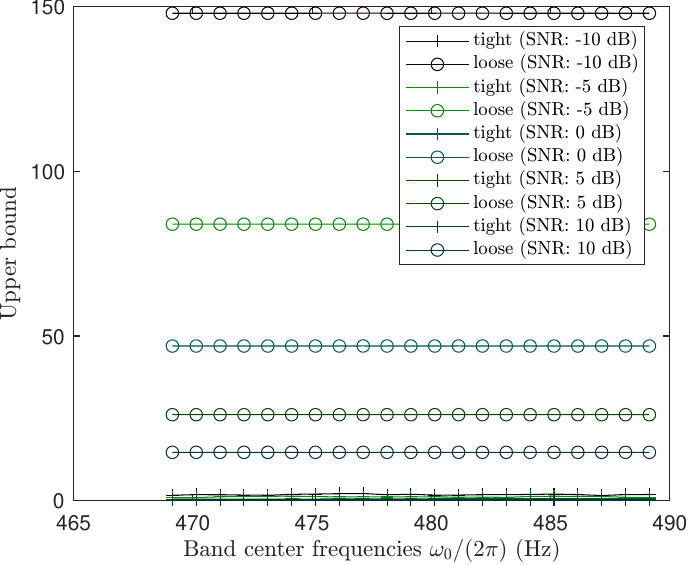}
        \caption{}
        \label{fig:noise_B_splines_5}
    \end{subfigure}
    \begin{subfigure}{0.32\textwidth}
        \centering
        \includegraphics[width=\linewidth, height=3.8cm]{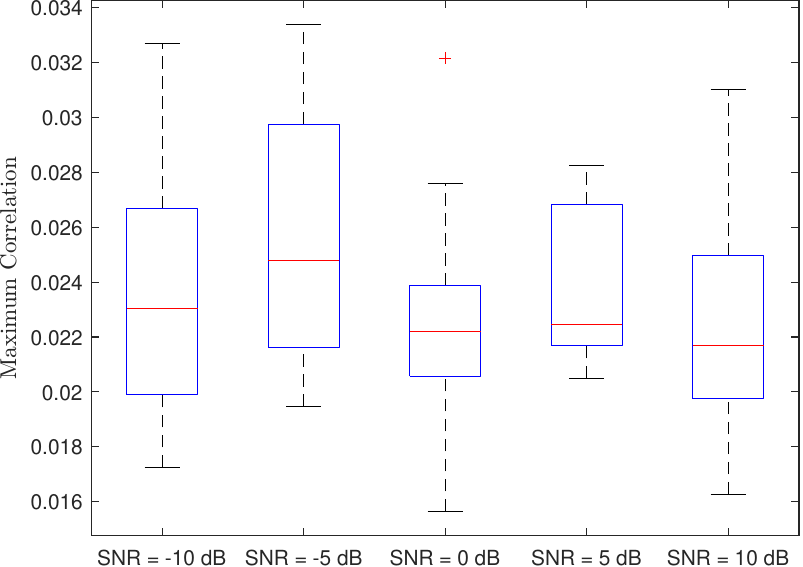}
        \caption{}
        \label{fig:noise_B_splines_6}
    \end{subfigure}
    \caption{B-spline mixtures. The six panels have the same interpretation as in the Gaussian-mixture case. Overall, the B-spline-mixture case exhibits behavior similar to the Gaussian-mixture case.}
    \label{fig:noise_B_splines}
\end{figure}
}


\section*{Data availability}

The numerical experiments in this study use synthetically generated
signals. The MATLAB scripts for reproducing the reported experiments
have been deposited as a private Mendeley Data draft for peer review.
The dataset will be made publicly available upon publication of the
article, and the final public DOI will be added to this statement.


\section*{Acknowledgements}
{This work was supported by the National Science and Technology Council of Taiwan under the contract NSTC 114-2221-E-002 -126 -MY3, and by a scholarship from the German Academic Exchange Service (DAAD). The authors also gratefully acknowledge Prof. Holger Rauhut of LMU Munich, Germany, for many insightful discussions and constructive suggestions on this work, particularly with regard to the rigor of the theoretical framework and the clarity of the mathematical arguments.} 


\bibliographystyle{elsarticle-num} 
\bibliography{references.bib}

@article{PSWF1,
  author  = {Slepian, David and Pollak, Henry O.},
  title   = {Prolate Spheroidal Wave Functions, {Fourier} Analysis and Uncertainty {I}},
  journal = {Bell System Technical Journal},
  volume  = {40},
  number  = {1},
  pages   = {43--63},
  year    = {1961}
}

@article{PSWF2,
  author  = {Landau, Henry J. and Pollak, Henry O.},
  title   = {Prolate Spheroidal Wave Functions, {Fourier} Analysis and Uncertainty {III}: The Dimension of the Space of Essentially Time- and Band-Limited Signals},
  journal = {Bell System Technical Journal},
  volume  = {41},
  number  = {4},
  pages   = {1295--1336},
  year    = {1962}
}

@article{PSWF3,
  author  = {Slepian, David},
  title   = {Prolate Spheroidal Wave Functions, {Fourier} Analysis, and Uncertainty {V}: The Discrete Case},
  journal = {Bell System Technical Journal},
  volume  = {57},
  number  = {5},
  pages   = {1371--1430},
  year    = {1978}
}

@article{PSWF4,
  author  = {Wang, Li-Lian},
  title   = {A Review of Prolate Spheroidal Wave Functions from the Perspective of Spectral Methods},
  journal = {Journal of Mathematical Study},
  volume  = {50},
  number  = {2},
  pages   = {101--143},
  year    = {2017}
}

@book{PSWF5,
  author    = {Flammer, Carson},
  title     = {Spheroidal Wave Functions},
  publisher = {Dover Publications},
  year      = {2014}
}

@inproceedings{PSWF6,
  author       = {Zhiyue Lin and Richard W. McCallum and Hongyu Wang},
  title        = {Computation and Performance of the Prolate-Spheroidal Wave Function Window in Spectral Estimation},
  booktitle    = {Proceedings of the 1996 IEEE International Conference on Acoustics, Speech, and Signal Processing ({ICASSP})},
  year         = {1996},
  volume       = {5},
  pages        = {2976--2978},
  doi          = {10.1109/ICASSP.1996.550179}
}

@inproceedings{PSWF7,
  author    = {Ding, Jian-Jiun and Pei, Soo-Chang},
  title     = {Reducing Sampling Error by Prolate Spheroidal Wave Functions and Fractional {Fourier} Transform},
  booktitle = {Proceedings of the {ICASSP} 2005},
  volume    = {4},
  pages     = {217--220},
  year      = {2005}
}

@article{PSWF8,
  author  = {Bonami, Aline and Karoui, Abderrazek},
  title   = {Approximations in {Sobolev} Spaces by Prolate Spheroidal Wave Functions},
  journal = {Applied and Computational Harmonic Analysis},
  volume  = {42},
  number  = {3},
  pages   = {361--377},
  year    = {2017}
}

@article{PSWF9,
  author       = {Abderrazek Karoui and Taher Moumni},
  title        = {Spectral analysis of the finite Hankel transform and circular prolate spheroidal wave functions},
  journal      = {Journal of Computational and Applied Mathematics},
  year         = {2009},
  volume       = {233},
  number       = {2},
  pages        = {315--333},
  doi          = {10.1016/j.cam.2009.07.037}
}

@article{PSWF10,
  author       = {Yan Tian},
  title        = {Superconvergence and fast implementation of the barycentric prolate differentiation},
  journal      = {Journal of Computational and Applied Mathematics},
  year         = {2022},
  volume       = {410},
  pages        = {114191},
  doi          = {10.1016/j.cam.2022.114191}
}

@article{Multitaper1,
  author  = {Thomson, David J.},
  title   = {Spectrum Estimation and Harmonic Analysis},
  journal = {Proceedings of the IEEE},
  volume  = {70},
  number  = {9},
  pages   = {1055--1096},
  year    = {1982},
  doi     = {10.1109/PROC.1982.12433}
}

@article{Multitaper2,
  author  = {Karnik, Santhosh and Romberg, Justin and Davenport, Mark A.},
  title   = {{Thomson}'s Multitaper Method Revisited},
  journal = {IEEE Transactions on Information Theory},
  volume  = {68},
  number  = {7},
  pages   = {4864--4891},
  year    = {2022},
  doi     = {10.1109/TIT.2022.3151415}
}

@article{Multitaper3,
  author       = {Thomas P. Bronez},
  title        = {On the Performance Advantage of Multitaper Spectral Analysis},
  journal      = {IEEE Transactions on Signal Processing},
  year         = {1992},
  volume       = {40},
  pages        = {2941--2946},
  doi          = {10.1109/78.175738}
}

@article{Multitaper4,
  author       = {Joakim And{\'e}n and Jos{\'e} Luis Romero},
  title        = {Multitaper Estimation on Arbitrary Domains},
  journal      = {SIAM Journal on Imaging Sciences},
  year         = {2020},
  volume       = {13},
  number       = {3},
  pages        = {1565--1594},
  doi          = {10.1137/19M1278338}
}

@article{Welch1,
  author  = {Welch, Peter D.},
  title   = {The Use of Fast Fourier Transform for the Estimation of Power Spectra: A Method Based on Time Averaging Over Short, Modified Periodograms},
  journal = {IEEE Transactions on Audio and Electroacoustics},
  volume  = {15},
  number  = {2},
  pages   = {70--73},
  year    = {1967},
  doi     = {10.1109/TAU.1967.1161901}
}

@book{Welch2,
  author    = {Stoica, Petre and Moses, Randolph L.},
  title     = {Spectral Analysis of Signals},
  publisher = {Prentice Hall},
  year      = {2005}
}

@book{Spectralanalysis,
  author    = {Percival, Donald B. and Walden, Andrew T.},
  title     = {Spectral Analysis for Physical Applications: Multitaper and Conventional Univariate Techniques},
  publisher = {Cambridge University Press},
  year      = {1993}
}

@article{Window,
  author  = {Harris, Fredric J.},
  title   = {On the Use of Windows for Harmonic Analysis with the Discrete {Fourier} Transform},
  journal = {Proceedings of the IEEE},
  volume  = {66},
  number  = {1},
  pages   = {51--83},
  year    = {1978},
  doi     = {10.1109/PROC.1978.10837}
}

@article{Integratedperiodogram1,
  author       = {Ulf Grenander and Murray Rosenblatt},
  title        = {Statistical Spectral Analysis of Time Series Arising from Stationary Stochastic Processes},
  journal      = {The Annals of Mathematical Statistics},
  year         = {1953},
  volume       = {24},
  number       = {4},
  pages        = {537--558},
  doi          = {10.1214/aoms/1177728913}
}

@article{Integratedperiodogram2,
  author  = {Ibragimov, Ildar A.},
  title   = {On Estimation of the Spectral Function of a Stationary {Gaussian} Process},
  journal = {Theory of Probability and Its Applications},
  volume  = {8},
  number  = {4},
  pages   = {366--401},
  year    = {1963}
}

@article{Integratedperiodogram3,
  author       = {M. S. Ginovyan},
  title        = {Asymptotically Efficient Nonparametric Estimation of Nonlinear Spectral Functionals},
  journal      = {Acta Applicandae Mathematicae},
  year         = {2003},
  volume       = {78},
  pages        = {145--154},
  doi          = {10.1023/A:1025708727313}
}

@article{Integratedperiodogram4,
  author       = {Mamikon S. Ginovyan},
  title        = {Efficient Estimation of Spectral Functionals for Continuous-Time Stationary Models},
  journal      = {Acta Applicandae Mathematicae},
  year         = {2011},
  volume       = {115},
  number       = {2},
  pages        = {233--254},
  doi          = {10.1007/s10440-011-9617-7}
}

@article{Integratedperiodogram5,
  author       = {Mamikon S. Ginovyan and Artur A. Sahakyan},
  title        = {Estimation of Spectral Functionals for {L}{\'e}vy-Driven Continuous-Time Linear Models with Tapered Data},
  journal      = {Electronic Journal of Statistics},
  year         = {2019},
  volume       = {13},
  number       = {1},
  pages        = {255--283},
  doi          = {10.1214/18-EJS1525}
}

@article{Integratedperiodogram6,
  author  = {Ginovyan, Mamikon S. and Sahakyan, Artur A.},
  title   = {Statistical Inference for Stationary Linear Models with Tapered Data},
  journal = {Statistics Surveys},
  volume  = {15},
  pages   = {154--194},
  year    = {2021}
}

@article{Integratedperiodogram7,
  author       = {Sourav Das and Suhasini Subba Rao and Junho Yang},
  title        = {Spectral Methods for Small Sample Time Series: A Complete Periodogram Approach},
  journal      = {Journal of Time Series Analysis},
  year         = {2021},
  volume       = {42},
  number       = {5-6},
  pages        = {597--621},
  doi          = {10.1111/jtsa.12584}
}

@inproceedings{Smoothstep1,
  author    = {Hazimeh, Hussein and Ponomareva, Natalia and Mol, Petros and Tan, Zhenyu and Mazumder, Rahul},
  title     = {The Tree Ensemble Layer: Differentiability Meets Conditional Computation},
  booktitle = {Proceedings of the 37th International Conference on Machine Learning ({ICML})},
  pages={4138--4148},
  series    = {Proceedings of Machine Learning Research},
  year      = {2020}
}

@misc{Smoothstep2,
  author       = {Kolouri, Soheil and Nadjahi, Kimia and Simsekli, Umut and Shahrampour, Shahin},
  title        = {Generalized Sliced Distances for Probability Distributions},
  howpublished = {arXiv preprint},
  year         = {2020},
  eprint       = {2002.12537},
  archivePrefix= {arXiv},
  primaryClass = {stat.ML}
}

@book{Jacobi,
  author    = {Szeg{\"o}, G{\'a}bor},
  title     = {Orthogonal Polynomials},
  series    = {American Mathematical Society Colloquium Publications},
  volume    = {23},
  edition   = {4},
  publisher = {American Mathematical Society},
  year      = {1975}
}

@book{Numericalanalysis,
  author    = {Scott, L. Ridgway},
  title     = {Numerical Analysis},
  publisher = {Princeton University Press},
  year      = {2011},
  isbn      = {9780691146867}
}

@book{FT,
  title={Table of integrals, series, and products},
  author={Gradshteyn, Izrail Solomonovich and Ryzhik, Iosif Moiseevich},
  year={2014},
  publisher={Academic press}
}

@article{Bsplines,
  author  = {Unser, Michael},
  title   = {Splines: A Perfect Fit for Signal and Image Processing},
  journal = {IEEE Signal Processing Magazine},
  volume  = {16},
  number  = {6},
  pages   = {22--38},
  year    = {1999},
  doi     = {10.1109/79.799930}
}

@book{Noise1,
  author    = {Foucart, Simon and Rauhut, Holger},
  title     = {A Mathematical Introduction to Compressive Sensing},
  series    = {Applied and Numerical Harmonic Analysis},
  publisher = {Birkh{\"a}user},
  year      = {2013}
}

@book{Noise2,
  author    = {Christensen, Ole},
  title     = {An Introduction to Frames and Riesz Bases},
  edition   = {2},
  series    = {Applied and Numerical Harmonic Analysis},
  publisher = {Birkh{\"a}user},
  year      = {2016}
}

@book{Noise3,
  author    = {Aster, Richard C. and Borchers, Brian and Thurber, Clifford H.},
  title     = {Parameter Estimation and Inverse Problems},
  publisher = {Elsevier},
  year      = {2018}
}

@article{Application1,
  author  = {Benwell, Christopher S. Y. and Davila-P{\'e}rez, Paula and Fried, Peter J. and Jones, Richard N. and Travison, Thomas G. and Santarnecchi, Emiliano and Pascual-Leone, Alvaro and Shafi, Mouhsin M.},
  title   = {{EEG} Spectral Power Abnormalities and Their Relationship with Cognitive Dysfunction in Patients with {Alzheimer}'s Disease and Type 2 Diabetes},
  journal = {Neurobiology of Aging},
  volume  = {85},
  pages   = {83--95},
  year    = {2020},
  doi     = {10.1016/j.neurobiolaging.2019.10.004}
}

@article{Application2,
  author  = {Pfurtscheller, Gert and Lopes da Silva, F. H.},
  title   = {Event-Related {EEG}/{MEG} Synchronization and Desynchronization: Basic Principles},
  journal = {Clinical Neurophysiology},
  volume  = {110},
  number  = {11},
  pages   = {1842--1857},
  year    = {1999},
  doi     = {10.1016/S1388-2457(99)00141-8}
}

@article{Application3,
  author  = {Shaffer, Fred and Ginsberg, Jay P.},
  title   = {An Overview of Heart Rate Variability Metrics and Norms},
  journal = {Frontiers in Public Health},
  volume  = {5},
  pages   = {258},
  year    = {2017},
  doi     = {10.3389/fpubh.2017.00258}
}

@book{Application4,
  author       = {P. P. Vaidyanathan},
  title        = {Multirate Systems and Filter Banks},
  publisher    = {Prentice Hall},
  address      = {Englewood Cliffs, NJ},
  year         = {1993},
  isbn         = {9780136057185}
}

@article{Application5,
  author  = {Muzaffar, Muhammad Umair and Sharqi, Rula},
  title   = {A Review of Spectrum Sensing in Modern Cognitive Radio Networks},
  journal = {Telecommunication Systems},
  volume  = {85},
  number  = {2},
  pages   = {347--363},
  year    = {2024},
  doi     = {10.1007/s11235-023-01079-1}
}

@article{Uncertainty,
  author  = {Donoho, David L. and Stark, Philip B.},
  title   = {Uncertainty Principles and Signal Recovery},
  journal = {SIAM Journal on Applied Mathematics},
  volume  = {49},
  number  = {3},
  pages   = {906--931},
  year    = {1989}
}

@article{Wavelet1,
  author  = {Daubechies, Ingrid},
  title   = {The Wavelet Transform, Time-Frequency Localization and Signal Analysis},
  journal = {IEEE Transactions on Information Theory},
  volume  = {36},
  number  = {5},
  pages   = {961--1005},
  year    = {1990},
  doi     = {10.1109/18.57199}
}

@book{Wavelet2,
  author    = {Mallat, St{\'e}phane},
  title     = {A Wavelet Tour of Signal Processing},
  publisher = {Academic Press},
  year      = {1999}
}

@article{Filterbank,
  author       = {J{\'e}r{\^o}me Antoni},
  title        = {Orthogonal-like fractional-octave-band filters},
  journal      = {The Journal of the Acoustical Society of America},
  year         = {2010},
  volume       = {127},
  number       = {2},
  pages        = {884--895},
  doi          = {10.1121/1.3273888}
}

@article{Energydetection,
  author       = {Harry Urkowitz},
  title        = {Energy Detection of Unknown Deterministic Signals},
  journal      = {Proceedings of the IEEE},
  year         = {1967},
  volume       = {55},
  number       = {4},
  pages        = {523--531},
  doi          = {10.1109/PROC.1967.5573}
}

@book{Sobolev,
  author    = {Adams, Robert A. and Fournier, John J. F.},
  title     = {Sobolev Spaces},
  publisher = {Elsevier},
  year      = {2003}
}

@book{Fourier,
  author    = {Stein, Elias M. and Shakarchi, Rami},
  title     = {Fourier Analysis: An Introduction},
  publisher = {Princeton University Press},
  year      = {2011}
}

@book{Lp,
  author    = {Axler, Sheldon},
  title     = {Measure, Integration \& Real Analysis},
  publisher = {Springer},
  year      = {2020},
  doi       = {10.1007/978-3-030-33143-6}
}

@book{Chebfun,
  editor    = {Driscoll, Tobin A. and Hale, Nicholas and Trefethen, Lloyd N.},
  title     = {Chebfun Guide},
  publisher = {Pafnuty Publications},
  address   = {Oxford},
  year      = {2014}
}

@misc{Signalprocessing,
  author       = {{The MathWorks, Inc.}},
  title        = {Signal Processing Toolbox},
  howpublished = {\url{https://www.mathworks.com/products/signal.html}},
  year         = {2024},
  note         = {Used with {MATLAB} {R}2024b. Accessed: Dec.\ 2025.}
}

@misc{Matlab2024b,
  author       = {{The MathWorks, Inc.}},
  title        = {{MATLAB} ({R}2024b)},
  howpublished = {\url{https://www.mathworks.com/products/matlab.html}},
  year         = {2024},
  note         = {Version 24.2 ({R}2024b). Accessed: Dec.\ 2025.}
}


\end{document}